\documentclass[11pt,reqno]{amsart} 
\pdfoutput=1

\usepackage[displaymath, mathlines]{lineno}

\usepackage{graphicx}
\graphicspath{{"/Users/erik/Dropbox/Spin_glass_project/Crisanti_Sommers/graphics/"}}

\usepackage{ifthen}

\usepackage{microtype}

\usepackage{xparse}

\usepackage[top=1.25in,bottom=1.25in,left=1.25in,right=1.25in]{geometry} 

\usepackage{xcolor}

\usepackage{enumitem}

\usepackage[labelfont=bf]{caption}

\usepackage[font={small},labelfont=md]{subfig}

\usepackage[noadjust,nosort]{cite}

\usepackage{amsmath, amsthm, amssymb, tikz, mathtools, array, mathrsfs, tensor}

\usepackage{polynom}

\usepackage{esint}

\usepackage{multicol}

\usepackage{dsfont}
	\newcommand{\one}{\mathds{1}}

\usepackage{framed}

\usepackage{stackengine}

\usepackage{upref}

\usepackage{hyperref}
	\hypersetup{colorlinks=true,linkcolor=blue,citecolor=blue,urlcolor=black} 

\numberwithin{equation}{section}


\newcommand{\red}{}

\newcommand{\eq}[1]{\begin{linenomath}\postdisplaypenalty=0\begin{align*} #1 \end{align*}\end{linenomath}}

\NewDocumentCommand{\eeq}{om}{\begin{linenomath}\postdisplaypenalty=0\begin{align} \IfNoValueTF{#1}{}{\tag{#1}} \begin{split} #2 \end{split} \end{align}\end{linenomath}}

\newcommand{\eeqs}[1]{\begin{linenomath}\postdisplaypenalty=0\begin{align} #1 \end{align}\end{linenomath}}

\newcommand{\stackref}[2]{
\readlist*\mylist{#1}
\stackrel{\mbox{\footnotesize\foreachitem\x\in\mylist[]{\ifnum\xcnt=1\else,\fi\eqref{\x}}}}{#2}
}
\newcommand{\stackrefp}[2]{
\readlist*\mylist{#1}
\stackrel{\hphantom{\mbox{\footnotesize\foreachitem\x\in\mylist[]{\ifnum\xcnt=1\else,\fi\eqref{\x}}}}}{#2}
}
\newcommand{\stackrefpp}[3]{
\readlist*\mylist{#1}
\readlist*\mylistt{#2}
\stackrel{\parbox{\widthof{\footnotesize\foreachitem\x\in\mylistt[]{\ifnum\xcnt=1\else,\fi\eqref{\x}}}}{\centering\footnotesize\foreachitem\x\in\mylist[]{{\ifnum\xcnt=1\else,\fi\eqref{\x}}}}}{#3}
}

\def\eps{\varepsilon}
\def\vphi{\varphi}

\newcommand{\E}{\mathbb{E}}

\newcommand{\R}{\mathbb{R}}

\newcommand{\T}{\mathbb{T}}

\renewcommand{\AA}{\mathcal{A}}

\newcommand{\DD}{\mathcal{D}}

\newcommand{\PP}{\mathcal{P}}

\newcommand{\SSS}{\mathscr{S}}
\newcommand{\TTT}{\mathscr{T}}

\newcommand{\iprod}[2]{\langle #1,\, #2\rangle}

\newcommand{\vc}[1]{{\boldsymbol #1}}

\newcommand{\wt}[1]{\widetilde{#1}}
\newcommand{\wh}[1]{\widehat{#1}}
\newcommand{\wc}[1]{\widecheck{#1}}
\DeclareFontFamily{U}{mathx}{\hyphenchar\font45}
\DeclareFontShape{U}{mathx}{m}{n}{
      <5> <6> <7> <8> <9> <10>
      <10.95> <12> <14.4> <17.28> <20.74> <24.88>
      mathx10
      }{}
\DeclareSymbolFont{mathx}{U}{mathx}{m}{n}
\DeclareFontSubstitution{U}{mathx}{m}{n}
\DeclareMathAccent{\widecheck}{0}{mathx}{"71}
\DeclareMathAccent{\wideparen}{0}{mathx}{"75}

\newcommand{\cc}{\mathrm{c}} 
\newcommand{\dd}{\mathrm{d}} 
\DeclareMathOperator{\Supp}{Supp}

            \makeatletter
            \DeclareFontFamily{OMX}{MnSymbolE}{}
            \DeclareSymbolFont{MnLargeSymbols}{OMX}{MnSymbolE}{m}{n}
            \SetSymbolFont{MnLargeSymbols}{bold}{OMX}{MnSymbolE}{b}{n}
            \DeclareFontShape{OMX}{MnSymbolE}{m}{n}{
                <-6>  MnSymbolE5
               <6-7>  MnSymbolE6
               <7-8>  MnSymbolE7
               <8-9>  MnSymbolE8
               <9-10> MnSymbolE9
              <10-12> MnSymbolE10
              <12->   MnSymbolE12
            }{}
            \DeclareFontShape{OMX}{MnSymbolE}{b}{n}{
                <-6>  MnSymbolE-Bold5
               <6-7>  MnSymbolE-Bold6
               <7-8>  MnSymbolE-Bold7
               <8-9>  MnSymbolE-Bold8
               <9-10> MnSymbolE-Bold9
              <10-12> MnSymbolE-Bold10
              <12->   MnSymbolE-Bold12
            }{}
            
            \let\llangle\@undefined
            \let\rrangle\@undefined
            \DeclareMathDelimiter{\llangle}{\mathopen}%
                                 {MnLargeSymbols}{'164}{MnLargeSymbols}{'164}
            \DeclareMathDelimiter{\rrangle}{\mathclose}%
                                 {MnLargeSymbols}{'171}{MnLargeSymbols}{'171}
            \makeatother
            
    \DeclareFontFamily{U}{matha}{\hyphenchar\font45}
    \DeclareFontShape{U}{matha}{m}{n}{ <-6> matha5 <6-7> matha6 <7-8>
    matha7 <8-9> matha8 <9-10> matha9 <10-12> matha10 <12-> matha12 }{}
    \DeclareSymbolFont{matha}{U}{matha}{m}{n}
    \DeclareFontFamily{U}{mathx}{\hyphenchar\font45}
    \DeclareFontShape{U}{mathx}{m}{n}{ <-6> mathx5 <6-7> mathx6 <7-8>
    mathx7 <8-9> mathx8 <9-10> mathx9 <10-12> mathx10 <12-> mathx12 }{}
    \DeclareSymbolFont{mathx}{U}{mathx}{m}{n}
    
    \DeclareMathDelimiter{\llbrack} {4}{matha}{"76}{mathx}{"30}
    \DeclareMathDelimiter{\rrbrack} {5}{matha}{"77}{mathx}{"38}
            
\newcommand{\ext}{\mathrm{ext}}

\DeclareMathOperator*{\Motimes}{\text{\raisebox{0.25ex}{\scalebox{0.8}{$\bigotimes$}}}}

\newenvironment{proofclaim}
{\begin{proof}}
{\renewcommand{\qedsymbol}{$\square$ (Claim)}
\end{proof}
\renewcommand{\qedsymbol}{$\square$}
}

\newtheorem{thm}{Theorem}[section]
\newtheorem{prop}[thm]{Proposition}
\newtheorem{cor}[thm]{Corollary}
\newtheorem{lemma}[thm]{Lemma}
\newtheorem{claim}[thm]{Claim}
\newtheorem{theirthm}{Theorem} 

\newtheorem{defn}[thm]{Definition}
\newtheorem{eg}[thm]{Example}

\newtheorem{remark}[thm]{Remark}

\renewcommand{\thefootnote}{\fnsymbol{footnote}}

\title[Crisanti--Sommers formula and simultaneous symmetry breaking]{Crisanti--Sommers formula and simultaneous symmetry breaking in multi-species spherical spin glasses}



\author{Erik Bates}
\thanks{E.B. was partially supported by NSF grant DMS-1902734} 
\address{\newline Department of Mathematics \newline University of Wisconsin--Madison \newline Van Vleck Hall \newline 480 Lincoln Drive \newline Madison, Wisconsin 53706-1324 
\newline \textup{\tt ewbates@wisc.edu}}

\author{Youngtak Sohn}
\thanks{Y.S. was partially supported by NSF grant DMS-1954337}
\address{\newline Department of Mathematics \newline Massachusetts Institute of Technology \newline Building 2 \newline 77 Massachusetts Avenue \newline Cambridge, Massachusetts 02139-4307 
\newline \textup{\tt youngtak@mit.edu}}

\begin{document}
\bibliographystyle{acm}

\renewcommand{\thefootnote}{\arabic{footnote}} \setcounter{footnote}{0}

\begin{abstract}
There is a rich history of expressing the limiting free energy of mean-field spin glasses as a variational formula over probability measures on $[0,1]$, where the measure represents the similarity (or ``overlap'') of two independently sampled spin configurations.
At high temperatures, the formula's minimum is achieved at a measure which is a point mass, meaning sample configurations are asymptotically orthogonal up to a magnetic field correction.
At low temperatures, though, a very different behavior emerges known as \textit{replica symmetry breaking} (RSB). 
The deep wells in the energy landscape create more rigid structure, and the optimal overlap measure is no longer a point mass.
The exact size of its support remains in many cases an open problem.

Here we consider these themes for multi-species spherical spin glasses.
Following a companion work in which we establish the Parisi variational formula, here we present this formula's Crisanti--Sommers representation.
In the process, we gain new access to a problem unique to the multi-species setting.
Namely, if RSB occurs for one species, does it necessarily occur for other species as well?
We provide sufficient conditions for the answer to be \textit{yes}.
For instance, we show that if two species share any quadratic interaction, then RSB for one implies RSB for the other.
Moreover, the level of symmetry breaking must be identical, even in cases of full RSB.
In the presence of an external field, any type of interaction suffices.
\end{abstract}

\maketitle
\vspace{-1\baselineskip}
\tableofcontents



\section{Introduction and background}
Mean-field spin glasses are meant to be mathematically tractable models of disordered magnetism.
A central advantage of mean-field models is the fact (or in certain cases, the hope) that the free energy density can be expressed as an explicit variational formula.
Most famously, the limiting free energy of the Sherrington--Kirkpatrick (SK) model \cite{sherrington-kirkpatrick75} is given by the Parisi formula \cite{parisi79,parisi80}, a fact proved rigorously by Guerra \cite{guerra03} and Talagrand \cite{talagrand06}.
This formula has since been generalized to mixed $p$-spin models \cite{panchenko14,auffinger-chen17}.

Spherical spin glasses have the further advantage of admitting an alternative formulation of their Parisi formulas, namely the Crisanti--Sommers (C--S) representation \cite{crisanti-sommers92}.
Indeed, because the C--S formula has a simpler objective function than the Parisi formula, it allows for finer analysis of spherical models. 
In the mathematical literature, this began with Talagrand \cite{talagrand06II}, who showed that the critical points of the two formulas coincide.
Since the C--S functional is strictly convex, uniqueness of the minimizer followed as a trivial corollary.
For comparison, the analogous result for Ising spin glasses is highly non-trivial \cite{auffinger-chen15II}.

Beyond the free energy itself, there is particular interest in \textit{Parisi measures}, the name given to minimizers of Parisi formulas.
These measures give the distribution of an overlap structure, which in turn describes the similarity of independently sampled spin configurations \cite{mezard-parisi-virasoro87,talagrand03}.
In the context of spherical models, the C--S formula makes possible certain explicit calculations involving the free energy, which can give additional access to the Parisi measures.
Numerous works have capitalized on this fact, for instance \cite{panchenko-talagrand07,auffinger-chen15I,jagannath-tobasco18,auffinger-zeng19}.

In light of this critical role played by the C--S representation, its generalization has been sought and indeed proved in a number of settings,
including mixed $p$-spin models \cite{chen13}, zero temperature models \cite{chen-sen17,jagannath-tobasco17II}, and models with vector spins \cite{ko?}.
The first goal of this paper is to add multi-species models to this list.
This is enabled by our companion work \cite{bates-sohn22}, in which we establish the Parisi formula for this setting.
In terms of logical dependence, the present paper relies on \cite{bates-sohn22} but not conversely.

Our second goal is to investigate the nature of symmetry breaking for multi-species models.
In the single-species models mentioned thus far, the order parameter for the Parisi and C--S formulas is a probability measure (typically on $[0,1]$).
A model can then be classified by whether the optimizing measure is a point mass (called \textit{replica symmetry}, abbreviated RS) or instead supported on multiple points (called \textit{replica symmetry breaking}, RSB).
Moreover, as one varies the temperature from high to low, there is often a phase transition from RS to RSB.
More intriguing than this transition is what lies beyond it: if a model is RSB, then what \textit{level} of symmetry breaking does it exhibit?
Namely, is the Parisi measure supported on exactly $k$ points (called $k$-RSB, e.g.~\cite[Thm.~4]{panchenko-talagrand07}, \cite[Prop.~3]{chen-sen17},
\cite[Thm.~5]{auffinger-chen18II}, \cite[Thm.~1]{auffinger-zeng19}) or possibly infinitely many points ($\infty$-RSB, e.g.~\cite[Thm.~1.1]{auffinger-chen-zeng20})?
In the latter case, it is expected that the support must contain a nonempty interval (called full RSB, or FRSB, e.g.~\cite[Prop.~2]{chen-sen17}, \cite[Thm.~4]{auffinger-chen18II}).
For discussion of this prediction, see \cite{montanari-ricci03,oppermann-sherrington05,oppermann-schmidt-sherrington07}. 

In multi-species models, questions of symmetry breaking are further complicated for several reasons.
At a most basic level, formulas for the free energy are more elaborate, and the analysis required for a rigorous study of Parisi measures is often already quite delicate.
In addition, it is not clear if or when multi-species formulas admit unique minimizers.
This speaks to the loss of certain tools (such as convexity), as well as possible complications in studying energy landscapes.
Finally, and perhaps most fundamental, there is no longer just one measure under consideration, but rather one measure for each species.
Therefore, \textit{a priori} it may be that the level of symmetry breaking is not uniform across species.

This possibility was raised by Panchenko in \cite{panchenko15I}, where the Parisi formula was proved for the multi-species (Ising) SK model.
Some related questions of symmetry breaking have been addressed in \cite{bates-sloman-sohn19,dey-wu21}, but to our knowledge, there has no been previous work---nor any serious predictions---comparing levels of symmetry breaking across species.
To this end, we will prove that if two species share a quadratic interaction (or equivalent thereof, see Example \ref{first_simultaneous_example}), then symmetry breaking occurs in one species only when occurs in the other, a behavior we call \textit{simultaneous} symmetry breaking.
Furthermore, the level of symmetry breaking must be the same.
We stress that the hypothesis of quadratic interaction can be weakened depending on which species have nonzero external fields (see the last part of Example \ref{first_simultaneous_example}).
Going beyond pairwise comparisons, we provide a second result that in suitable circumstances allows one to conclude that a third species (or fourth, and so on) will also exhibit simultaneity. 
The added value here is that it may be otherwise not possible to conclude that the third species is simultaneous with one of the other two individually; see Example \ref{chaining_example}.
Our results hold for positive definite covariance structures, at least for Parisi measures.
One can also talk about ``Crisanti--Sommers measures'', and then this convexity assumption is no longer needed.

The paper is organized as follows.
\begin{itemize}
\item In Section \ref{def_sec}, we define the multi-species spherical spin glass model.
The exact setting is not necessary for the rest of the paper, but at least provides a minimal amount of physical context for the mathematics that follows.
It was shown in \cite{bates-sohn22} that the limiting free energy of this model is given by a Parisi variational formula, which is recalled in
Section \ref{parisi_sec}.

\item In Section \ref{cs_sec}, we define the Crisanti--Sommers variational formula and state our first main result, namely the equivalence of the Parisi and C--S formulas (Theorem \ref{parisi_cs}).
The proof of this equivalence rests on two families of identities---one for the Parisi formula and one for the C--S formula---which must be satisfied by any minimizer.
The two families are related to each other via integration by parts.

\item In Section \ref{rsb_sec}, we formally introduce the notions of symmetry breaking and simultaneity (Definition \ref{simultaneous_defn}).
Our second main result is the aforementioned statements about simultaneous symmetry breaking (Theorems \ref{simultaneous_thm_1} and \ref{simultaneous_thm_2}).
We also include Examples \ref{decoupled_example},  \ref{first_simultaneous_example}, \ref{chaining_example}, and \ref{incomplete_example} to illustrate applications of these theorems. 
As in the previous bullet point, the proofs rely on the two families of minimizer identities.

\item These identities are finally revealed in Section \ref{identities_sec}: Theorem \ref{parisi_min_thm} for Parisi minimizers, and Theorem \ref{cs_min_thm} for Crisanti--Sommers minimizers. 
The proof of Theorems \ref{simultaneous_thm_1} and \ref{simultaneous_thm_2} is remarkably brief given the identities, and so we write it immediately after.

\item The remaining proofs are partitioned into three parts.
First, continuity of the C--S functional is shown in Section \ref{continuity_proof}.
Following this preliminary step, the identities satisfied by minimizers are verified in Section \ref{identities_proof}, and then the C--S formula is confirmed in Section \ref{parisi_cs_proof}.
While none of the arguments have appeared at this level of generality before, the most novel work comes in Section \ref{identities_proof}.
\end{itemize}
Finally, for a broader review of related literature, we refer the reader to \cite[Sec.~1.4]{bates-sohn22}, and also to the inexhaustibly useful \cite{talagrand11I,talagrand11II,panchenko13}.

\subsection{The setting: multi-species spherical spin glasses} \label{def_sec}
Let $\SSS$ be a finite set, whose elements index the various species.
Suppose we write each positive integer $N$ as a sum of nonnegative integers, $N = \sum_{s\in\SSS}N^s$.
We then define the following product of spheres:
\eq{
\T_N \coloneqq 
\Motimes_{s\in\SSS} S_{N^s}, \quad \text{where} \quad
 S_n \coloneqq \{\sigma\in\R^n:\, \|\sigma\|_2^2 = n\}.
}
We assume that 
\eeq[H1]{ \label{lambda_assumption}
\lim_{N\to\infty} \frac{N^s}{N} = \lambda^s\in(0,1] \quad \text{for each $s\in\SSS$}.
}
An element of $\T_N$ will be written $\sigma = (\sigma(s))_{s\in\SSS}$, where $\sigma(s)\in S_{N^s}$.
The \textit{overlap} between two configurations $\sigma^1,\sigma^2\in\T_N$ is the following vector belonging to $[-1,1]^\SSS$:
\eq{ 
\vc R(\sigma^1,\sigma^2) \coloneqq (R^s(\sigma^1,\sigma^2))_{s\in\SSS}, \quad \text{where} \quad
R^s(\sigma^1,\sigma^2) \coloneqq \frac{\iprod{\sigma^1(s)}{\sigma^2(s)}}{N^s},
}
and $\iprod{\cdot}{\cdot}$ denotes the Euclidean inner product.

For each integer $p\geq1$, assume $\vc\Delta^2_{p} = (\Delta^2_{s_1,\dots,s_p})_{ s_1,\dots,s_p\in\SSS}$ is a symmetric $p$-dimensional tensor with nonnegative entries.
Let $(\beta_p)_{p\geq1}$ be a sequence of nonnegative numbers such that
\eeq[H2]{\label{decay_condition}
\sum_{p\ge1}\beta_p(1+\eps)^p\sum_{s\in\SSS^p}\Delta_{s_1,\dots,s_p}^2\lambda^{s_1}\cdots\lambda^{s_p} \quad \text{for some $\eps>0$},
}
so that the following function is well-defined and analytic on some open set containing $[-1,1]^\SSS$:
\eeq{ \label{xi_def}
\xi(\vc q) \coloneqq \sum_{p\geq1}\beta_p\sum_{s\in\SSS^p}\Delta^2_{s_1,\dots,s_p}\lambda^{s_1}\cdots\lambda^{s_p}q^{s_1}\cdots q^{s_p}, \quad \vc q = (q^s)_{s\in\SSS}\in[-1,1]^\SSS.
}
Let $H_N$ be a centered Gaussian processes on $\T_N$ whose covariance function is
\eeq{ \label{HN_def}
\E[H_N(\sigma^1)H_N(\sigma^2)]
= \xi(\vc R(\sigma^1,\sigma^2)).
}
The \textit{free energy} associated to $H_N$ with an external field $\vc h = (h_s)_{s\in\SSS}$ is the quantity
\eeq{ \label{free_energy_def}
F_N \coloneqq \frac{1}{N}\log\int_{\T_N}\exp\Big( H_N(\sigma)+\sum_{s\in\SSS}h_s\iprod{\sigma(s)}{\vc 1}\Big)\ \dd\sigma,
}
where $\dd\sigma$ denotes the product measure under which $\sigma(s)$ is uniformly distributed on $S_{N^s}$, for each $s\in\SSS$.
In order for the upcoming Parisi formula to hold, it is necessary to assume the Hessian of $\xi$ is nonnegative definite on the nonnegative orthant:
\eeq[H3]{ \label{xi_convex}
\nabla^2 \xi(\vc q)\geq 0 \quad \text{for $\vc q\in [0,1]^\SSS$}.
}
For convenience, we also define the following functions involving derivatives of $\xi$:
\eeq{ 
\xi^s(\vc q)
\coloneqq \frac{1}{\lambda^s}\frac{\partial \xi}{\partial q^s}(\vc q), \qquad 
\label{gamma_theta_def}
\theta(\vc q) 
\coloneqq \vc q\cdot\nabla\xi(\vc q) - \xi(\vc q).
}
Throughout the rest of the paper, we only consider the restriction of $\xi$ to $[0,1]^\SSS$.
That is, the reader should always assume $q^s\geq0$ for all $s\in\SSS$.

\subsection{The Parisi formula} \label{parisi_sec}
The Parisi functional has several inputs, one of which is an element from the following space of functions.

\begin{defn} \label{lambda_av_def}
Given $\vc\lambda=(\lambda^s)_{s\in\SSS}$, a map $\Phi = (\Phi^s)_{s\in\SSS} \colon [0,1]\to[0,1]^\SSS$ is said to be $\vc\lambda$-\textit{admissible} if each coordinate $\Phi^s$ is non-decreasing and continuous, and jointly they satisfy
\eeq{ \label{admissible_def}
\sum_{s\in\SSS}\lambda^s\Phi^s(q) = q \quad \text{for all $q\in[0,1]$}.
}
If $\zeta$ is a Borel probability measure on $[0,1]$, then $(\zeta,\Phi)$ is called a \textit{$\vc\lambda$-admissible pair}.
\end{defn}


Given any $\vc\lambda$-admissible pair $(\zeta,\Phi)$, for each $s\in\SSS$ we define the following function:
\eeq{ \label{ds_def}
d^s(q) \coloneqq \int_q^1 \zeta\big([0,u]\big)(\xi^s\circ \Phi)'(u)\ \dd u, \quad q\in[0,1].
}
Here we must point out because $\Phi^s$ is monotone, the derivatives appearing above (and below) exist on a set of full Lebesgue measure.
For any vector $\vc b = (b^s)_{s\in\SSS}$ satisfying the constraint
\eeq{ \label{parisi_constraint}
b^s > d^s(0) \quad \text{for each $s\in\SSS$},
}
we consider the quantity
\eeq{ \label{A_def}
A(\zeta,\Phi,\vc b) \coloneqq
\sum_{s\in\SSS} \frac{\lambda^s}{2}\Big[\frac{h_s^2+\red{\xi^s(\vc 0)}}{b^s-d^s(0)}+b^s -1-\log b^s + \int_0^1\frac{(\xi^s\circ \Phi)'(q)}{b^s-d^s(q)}\ &\dd q\Big]\\
- \frac{1}{2}\int_0^1 \zeta\big([0,q]\big)(\theta\circ\Phi)'(q)\ &\dd q.
}

\begin{theirthm}[{Parisi formula, \cite[Thm.~1.3 and Rmk.~1.4]{bates-sohn22}}] \label{parisi_thm}
Assuming \eqref{lambda_assumption}, \eqref{decay_condition}, and \eqref{xi_convex}, we have
\eeq{ \label{parisi_formula}
\lim_{N\to\infty} F_N = \inf_{\zeta,\Phi,\vc b} A(\zeta,\Phi,\vc b) \quad \mathrm{a.s.},
}
where the infimum is over triples satisfying \eqref{parisi_constraint}.
\end{theirthm}

The proof of Theorem \ref{parisi_thm} requires separate verification of matching upper and lower bounds, together with a standard concentration inequality to show $F_N$ concentrates around its mean.
The rest of this paper is mostly divorced from these arguments, as we focus exclusively on the right-hand side of \eqref{parisi_formula} (except for a brief moment in the proof of Theorem \ref{parisi_cs}).

\section{Main results}

\subsection{The Crisanti--Sommers formula} \label{cs_sec}
The input to the C--S functional will be a $\vc\lambda$-admissible pair $(\zeta,\Phi)$, but we restrict attention to those pairs for which 
there is some $q_*\in[0,1)$ such that
\eeq{ \label{gap_assumption}
\zeta\big([0,q_*]\big) = 1 \qquad \text{and} \qquad \Phi^s(q_*)< 1 \quad \text{for all $s\in\SSS$}.
}
Fixing such a pair, we define
\eeq{ \label{delta_def}
\Delta^s(q) \coloneqq \int_q^1 \zeta\big([0,u]\big)(\Phi^s)'(u)\ \dd u,
}
and then the \textit{Crisanti--Sommers functional} is given by
\eeq{ \label{B_def}
B(\zeta,\Phi) \coloneqq \sum_{s\in\SSS}\frac{\lambda^s}{2}\Big[\red{h_s^2}\Delta^s(0)+\int_0^{q_*}\frac{(\Phi^s)'(q)}{\Delta^s(q)}\ \dd q+\log \Delta^s(q_*)&\Big]\\
+\frac{1}{2}\int_0^1\zeta\big([0,q]\big)(\xi\circ \Phi)'(q)\ \dd q&.
}

\begin{remark} \label{q_star_remark}
Observe that if $q'\in(q_*,1)$ and $\Phi^s(q')<1$, then
\eq{
\int_0^{q'}\frac{(\Phi^s)'(q)}{\Delta^s(q)}\ \dd q
- \int_0^{q_*}\frac{(\Phi^s)'(q)}{\Delta^s(q_*)}\ \dd q =\int_{q_*}^{q'}\frac{(\Phi^s)'(q)}{1-\Phi^s(q)}\ \dd q
= \log\frac{1-\Phi^s(q_*)}{1-\Phi^s(q')}
= \log \frac{\Delta^s(q_*)}{\Delta^s(q')}.
}
Therefore, the exact choice of $q_*$ does not affect the value of $B(\zeta,\Phi)$ so long as \eqref{gap_assumption} is satisfied.
If no such $q_*$ exists, we can simply take $B(\zeta,\Phi)=\infty$.
\end{remark}

Our first main result is the identification of Parisi minimizers with C--S minimizers.
Let us introduce a positive definite version of \eqref{xi_convex}: 
\eeq[H3$'$]{ \label{xi_strictly_convex}
\nabla^2\xi(\vc q) > 0 \quad \text{for all $q\in[0,1]^\SSS\setminus\{\vc 0\}$}.
}

\begin{thm} \label{parisi_cs}
Assume \eqref{decay_condition} and \eqref{xi_convex}. Then
\eq{ 
\inf_{\zeta,\Phi,\vc b} A(\zeta,\Phi,\vc b) = \inf_{\zeta,\Phi} B(\zeta,\Phi).
}
If we further assume \eqref{xi_strictly_convex}, then the set of minimizers for each side is the same, in the sense that
\eq{
A(\zeta,\Phi,\vc b)= \inf A \quad \text{for some $\vc b$}  \quad \iff \quad B(\zeta,\Phi) = \inf B.
}
\end{thm}

Because of Theorem \ref{parisi_thm}, we immediately obtain the following corollary.

\begin{cor}[Crisanti--Sommers formula] \label{cs_thm}
Assuming \eqref{lambda_assumption}, \eqref{decay_condition}, and \eqref{xi_convex}, we have
\eeq{ \label{cs_formula}
\lim_{N\to\infty} F_N = \inf_{\zeta,\Phi} B(\zeta,\Phi) \quad \mathrm{a.s.},
}
where the infimum is over $\vc\lambda$-admissible pairs satisfying \eqref{gap_assumption} for some $q_*\in[0,1)$.
\end{cor}

Theorem \ref{parisi_cs} is the central aim of the paper.
The key observation is Lemma \ref{big_lemma}, which gives a sufficient condition for $A(\zeta,\Phi,\vc b)$ to be equal to $B(\zeta,\Phi)$.
The rest of the work is to show that minimizers satisfy this condition, precisely because of the identities we will soon state as Theorems \ref{parisi_min_thm} and \ref{cs_min_thm}.
This line of reasoning is carried out in Section \ref{parisi_cs_proof} and was inspired by Talagrand's argument in the single-species case \cite[Sec.~4]{talagrand06II}.
What is different here---apart from the presence of multiple species---is that we work with completely general measures $\zeta$, not just those with finite support.
This significantly complicates the relevant calculations.
On the other hand, Theorems \ref{parisi_min_thm} and \ref{cs_min_thm} are analogous to \cite[Prop.~2.1 and Lem.~4.3]{talagrand06II}.
For these results, the difficulty in generalizing Talagrand's arguments is less about the support of $\zeta$, and more about the interactions between species.
As it turns out, what is necessary to overcome these complications is intertwined with the issue of simultaneous symmetry breaking, which we discuss next.

\subsection{Simultaneous symmetry breaking} \label{rsb_sec}
In the single-species case (i.e.~$|\SSS|=1$), the formulas \eqref{parisi_formula} and \eqref{cs_formula} reduce to those given in \cite{talagrand06II,chen13}.
This is because the only map $\Phi:[0,1]\to[0,1]$ satisfying Definition \ref{lambda_av_def} is the identity function, and so the input to these formulas is just the one-dimensional measure $\zeta$.
Furthermore, in this case it is a trivial matter to check that the Crisanti--Sommers formula is strictly convex in $\zeta$, thus leading to a unique minimizer.
This measure, which we will keep denoting by $\zeta$, is called the \textit{Parisi measure}, and is the functional order parameter for classifying the spin glass model into one of two phases.
If $\zeta$ is a point mass, then the model is said to be \textit{replica symmetric} (RS); otherwise the model is \textit{replica symmetry breaking} (RSB).
In the latter case, there is a further classification based on the `level' of symmetry breaking.
Namely, if $|\Supp(\zeta)| = k+1$, where $k\in\{1,2,\dots,\infty\}$, then we say the model is \textit{$k$-RSB}.
If the model is $\infty$-RSB, then one can ask the even subtler question of whether $\Supp(\zeta)$ contains a nonempty interval; this behavior is called \textit{full RSB}.

In seeking to generalize this classification scheme to multi-species models, one immediately encounters a technical roadblock: the Crisanti--Sommers functional \eqref{B_def} is no longer convex in $\Phi$. 
Therefore, uniqueness of the minimizer in Corollary \ref{cs_thm} is not known (and here we mean uniqueness up to a natural pseudometric $\DD$ defined in Section \ref{continuity_proof}).
Nevertheless, we can still speak about the symmetry breaking status of any particular minimizer $(\zeta,\Phi)$.
But then we are faced with a second and more novel  complication: the analogous object to $\zeta$ from before is the pushforward measure $\zeta\circ(\Phi^s)^{-1}$, whose support very much depends on the species $s$.
In particular, some species may be in the RS phase (i.e.~$|\Supp(\zeta\circ(\Phi^s)^{-1})|=1)$ while others are RSB.
Moreover, those in the RSB phase need not have the same level of symmetry breaking.

\begin{eg} \label{decoupled_example}
If the covariance function from \eqref{HN_def} is of the form
\eq{
\xi(\vc q) = \sum_{s\in\SSS}\psi^s(q^s)
}
for some functions $(\psi^s)_{s\in\SSS}$, then there are no  interactions between species.
That is, the spin glass model is a product of independent single-species models.
Of course, each single-species model can be tuned separately to create different levels of symmetry breaking.
\end{eg}

We are interested in finding conditions under which the species must break symmetry together (if they break symmetry at all).
We thus define the following equivalence relation on $\SSS$.

\begin{defn} \label{simultaneous_defn}
Given any pair of species $s,t\in\SSS$, let us say that a $\vc\lambda$-admissible pair $(\zeta,\Phi)$ is $(s,t)$-simultaneous if the following equivalence holds for all $q_0,q_1\in\Supp(\zeta)$:
\eeq{ \label{simultaneous_def}
\Phi^s(q_0) < \Phi^s(q_1) \quad \iff \quad \Phi^t(q_0) < \Phi^t(q_1).
}
More generally, for any subset $\TTT\subset\SSS$, we say that $(\zeta,\Phi)$ is $\TTT$-simultaneous if it is $(s,t)$-simultaneous for every $s,t\in\TTT$.
\end{defn}

A more physical interpretation of this definition is
the following.

\begin{lemma} \label{simultaneous_lemma}
If $(\zeta,\Phi)$ is $(s,t)$-simultaneous, then
there is a measure-preserving and increasing bijection between $\Supp(\zeta\circ(\Phi^s)^{-1})$ and $\Supp(\zeta\circ(\Phi^t)^{-1})$.
\end{lemma} 

The proof is a matter of chasing definitions, and so we postpone it until the end of Section \ref{identities_sec}.
Clearly \eqref{simultaneous_def} is a much more straightforward condition to check, but in light of the discussion that opened this section, we really care about the interpretation offered by Lemma \ref{simultaneous_lemma}.
More specifically, we care about whether or not a \textit{minimizer} to the Crisanti--Sommers formula is $(s,t)$-simultaneous.
In order to maintain the greatest possible generality, we will give all of our results in terms of Crisanti--Sommers minimizers.
If \eqref{xi_strictly_convex} holds, then these coincide with the Parisi minimizers thanks to Theorem \ref{parisi_cs}.

To state our first theorem on simultaneous symmetry breaking, we define the set of species which have nonzero external fields:
\eeq{ \label{S_ext}
\SSS_\ext \coloneqq \{s\in\SSS:\, h_s^2>0\}.
}
The reason we make this definition is that every $s\in\SSS_\ext$ necessarily has $0\notin\Supp(\zeta\circ(\Phi^s)^{-1})$ for any minimizer $(\zeta,\Phi)$ to \eqref{cs_formula}, a fact which will become clear in Section \ref{identities_sec}.
Consequently, the presence of external fields only serves to strengthen our statements regarding simultaneous symmetry breaking, as the following result demonstrates.

\begin{thm} \label{simultaneous_thm_1}
Assume \eqref{decay_condition}.
If a minimizer $(\zeta,\Phi)$ to \eqref{cs_formula} satisfies
\eeq{ \label{simultaneous_condition_specific}
\frac{\partial\xi^s}{\partial q^t}(\Phi(q))>0 \quad \text{whenever $q\in\Supp(\zeta)$ and $\Phi^s(q)\vee\Phi^t(q)>0$},
}
then $(\zeta,\Phi)$ is $(s,t)$-simultaneous.
In particular, if
\eeq{ \label{simultaneous_condition}
\frac{\partial \xi^s}{\partial q^t}(\vc q) > 0 \quad \text{whenever $q^s\vee q^t>0$ and $q^r>0$ for all $r\in\SSS_\ext$},
}
then any minimizer to \eqref{cs_formula} is $(s,t)$-simultaneous.
\end{thm}


Here is an application of Theorem \ref{simultaneous_thm_1}.

\begin{eg} \label{first_simultaneous_example}
Notice from \eqref{xi_def} that
\eq{
\frac{\partial \xi^s}{\partial q^t} \geq \beta_2\lambda^t\Delta_{s,t}^2.
}
Hence \eqref{simultaneous_condition} holds whenever $\beta_2\Delta_{s,t}^2>0$.
But this is not a necessary condition.
For instance, if 
\eq{
\beta_p\underbrace{\Delta^2_{s,t,\dots,t}}_{\rlap{\footnotesize\text{$p-1$ copies of $t$}}}>0 \quad \text{for some $p\geq2$},
\quad\text{and} \quad
\beta_{p'}\underbrace{\Delta^2_{t,s,\dots,s}}_{\rlap{\footnotesize\text{$p'-1$ copies of $s$}}}>0 \quad \text{for some $p'\geq2$},
}
then \eqref{simultaneous_condition} again holds.
If $h_r^2>0$ for all $r\in\SSS$, then \eqref{simultaneous_condition} is implied by an even weaker condition, namely that
\eq{ 
\beta_p\Delta_{s,t,r_1,\dots,r_{p-2}}^2 > 0 \quad \text{for some $p\geq2$, $r_1,\dots,r_{p-2}\in\SSS$}.
}
\end{eg}

Notice that any model can be made to satisfy \eqref{simultaneous_condition} via an arbitrarily small perturbation.
Namely, if $\xi$ is replaced with $\xi+\eps q^sq^t$, where $\eps>0$, then \eqref{simultaneous_condition} holds.
Upon performing this modification for every pair of species, we arrive at a model in which all species are simultaneous.
In this sense, simultaneous symmetry breaking might be regarded as a ``generic'' feature of spherical spin glasses.

Of course, we also wish to address the coordination of more than just two species.
We thus state the following generalized form of Theorem \ref{simultaneous_thm_1}.
It allows one to ``chain'' simultaneity relations. 

\begin{thm} \label{simultaneous_thm_2}
Assume \eqref{decay_condition}.
Suppose $(\zeta,\Phi)$ is a minimizer to \eqref{cs_formula} that is $\TTT$-simultaneous for some $\TTT\subset\SSS$.
If
\eeq{ \label{simultaneous_condition_2_species}
\max_{t\in\TTT}\frac{\partial \xi^s}{\partial q^t}(\Phi(q)) > 0 \qquad \text{whenever} \qquad \parbox{2.7in}{\centering $q\in\Supp(\zeta)$, $\Phi^s(q)\vee \min_{t\in\TTT}\Phi^t(q)>0$, and $\Phi^r(q)>0$ $\forall$ $r\in\SSS_\ext$,}
}
then $(\zeta,\Phi)$ is $(\TTT\cup\{s\})$-simultaneous.
In particular, if
\eeq{ \label{simultaneous_condition_2}
\max_{t\in\TTT}\frac{\partial \xi^s}{\partial q^t}(\vc q) > 0 \quad \text{whenever $q^s\vee \min_{t\in\TTT}q^t>0$ and $q^r>0$ for all $r\in\SSS_\ext$},
}
then every minimizer which is $\TTT$-simultaneous is also $(\TTT\cup\{s\})$-simultaneous.
\end{thm}

Here is an application which distinguishes Theorem \ref{simultaneous_thm_2} from Theorem \ref{simultaneous_thm_1}.

\begin{eg} \label{chaining_example}
Suppose $\SSS = \{r,s,t\}$, and that we have
\eq{ 
\beta_2\Delta_{r,t}^2 > 0, \qquad
\beta_{p}\underbrace{\Delta^2_{r,s,\dots,s}}_{\rlap{\footnotesize\text{$p-1$ copies of $s$}}}>0 \quad \text{for some $p\geq2$}, \qquad
\beta_{p'}\underbrace{\Delta^2_{s,t,\dots,t}}_{\rlap{\footnotesize\text{$p'-1$ copies of $t$}}}>0 \quad \text{for some $p'\geq2$}.
}
By Theorem \ref{simultaneous_thm_1}, the first of these inequalities ensures that any minimizer to \eqref{cs_formula} will be $(r,t)$-simultaneous, as in Example \ref{first_simultaneous_example}.
Then Theorem \ref{simultaneous_thm_2} comes into effect with $\TTT=\{r,t\}$, for the second inequality ensures that $\partial^r\xi^s(\vc q)>0$ whenever $q^s>0$, while
the third inequality ensures $\partial^t\xi^s(\vc q)>0$ whenever $q^t>0$.
Therefore, \eqref{simultaneous_condition_2} is satisfied, and so for any minimizer in any model fulfilling the three conditions displayed above, all three species are simultaneous.
\end{eg}

Let us also include a case for which we do not have a complete answer.

\begin{eg} \label{incomplete_example}
Again suppose $\SSS=\{r,s,t\}$, and that all we know about the covariance function $\xi$ is 
\eeq{ \label{beta_3_assumption}
\beta_3\Delta_{r,s,t}^2>0.
}
This is not enough to conclude \eqref{simultaneous_condition} for any pair of species.
Nevertheless, if species $r$ is the ``first'' to break symmetry for some minimizer $(\zeta,\Phi)$, in the sense that
\eeq{ \label{r_is_first}
 \Phi^{s}(q)\vee \Phi^{t}(q) > 0  \quad \implies \quad \Phi^r(q) > 0\quad \text{for $q\in\Supp(\zeta)$},
}
then Theorem \ref{simultaneous_thm_1} forces $(\zeta,\Phi)$ to at least be $(s,t)$-simultaneous.
This is because
\eq{
\frac{\partial\xi^s}{\partial q^t} \geq \beta_3\Delta^2_{r,s,t}\lambda^r\lambda^t q^r,
}
and so under \eqref{beta_3_assumption} and \eqref{r_is_first}, the hypothesis \eqref{simultaneous_condition_specific} holds even though \eqref{simultaneous_condition} does not.
Note that \eqref{r_is_first} is trivially true if $h_r^2>0$, since in this case
\eqref{cs_identity} implies $\Phi^r(q)>0$ for all $q\in\Supp(\zeta)$.
Consequently, $h_r^2>0$ implies any minimizer is $(s,t)$-simultaneous.
If $h_s^2$ is also positive, then any minimizer is also $(r,t)$-simultaneous, meaning all three species are simultaneous.
\end{eg}

It should be mentioned that we have not addressed the actual \textit{existence} of symmetry breaking.
There are well-known arguments to prove symmetry breaking at sufficiently low temperatures, e.g.~\cite[Prop.~4.2]{barra-contucci-mingione-tantari15}.
For the multi-species Ising SK model, a more quantitative condition for symmetry breaking is given in \cite{bates-sloman-sohn19,dey-wu21}.
To go further and actually determine the \textit{level} of symmetry breaking is in general a famously challenging problem already for single-species models.
This is especially true for models at positive temperature, which is the setting considered here.
We leave these important questions for future work.

\subsection{Identities satisfied by minimizers} \label{identities_sec}
Now we state the essential identities which underlie all the results of Sections \ref{cs_sec} and \ref{rsb_sec}.
Proving these identities is the biggest challenge of this paper.
The arguments are perturbative and are carried out in Section \ref{identities_proof}.
First we consider Parisi minimizers.

\begin{thm} \label{parisi_min_thm}
Assume \eqref{decay_condition}.
There exists a triple $(\zeta,\Phi,\vc b)$ which achieves the minimum in \eqref{parisi_formula}, and necessarily satisfies
\begin{subequations}
\label{parisi_identity}
\eeq{
1 - \frac{1}{b^s} - \frac{h_s^2+\red{\xi^s(\vc0)}}{(b^s-d^s(0))^2} &= \int_0^1\frac{(\xi^s\circ\Phi)'(q)}{(b^s-d^s(q))^2}\ \dd q \label{parisi_identity_1} \quad \text{for all $s\in\SSS$}.
}
Furthermore, if \eqref{xi_strictly_convex} holds,
then any minimizer must also satisfy
\eeq{
\Phi^s(q) &= \frac{h_s^2+\red{\xi^s(\vc0)}}{(b^s-d^s(0))^2}+\int_0^q \frac{(\xi^s\circ\Phi)'(u)}{(b^s-d^s(u))^2}\ \dd u \quad \text{for all $q\in\Supp(\zeta)$, $s\in\SSS$}.
\label{parisi_identity_2}
}
\end{subequations}
\end{thm}

One obvious consequence of \eqref{parisi_identity_2} is that $\Phi^s(q)$ can be no smaller than $h_s^2/((b^s-d^s(0))^2$, for $q\in\Supp(\zeta)$.
So the presence of a nonzero external field on species $s$ forces the corresponding overlap to be bounded away from $0$.
We will be able to make the same observation from the following parallel result about Crisanti--Sommers minimizers. 
This is why the set $\SSS_\ext$ from \eqref{S_ext} is given special attention in Theorems \ref{simultaneous_thm_1} and \ref{simultaneous_thm_2}.


\begin{thm} \label{cs_min_thm}
Assume \eqref{decay_condition}.
There exists a $\vc\lambda$-admissible pair $(\zeta,\Phi)$ which achieves the infimum in \eqref{cs_formula}.
Furthermore, any minimizer must satisfy
\eeq{ \label{cs_identity}
\xi^s(\Phi(q))+\red{h_s^2}  = \int_0^q \frac{(\Phi^s)'(u)}{(\Delta^s(u))^2}\ \dd u \quad \text{for all $q\in\Supp(\zeta)$, $s\in\SSS$}.
}
\end{thm}


To demonstrate just how useful this identity is, let us now prove Theorem \ref{simultaneous_thm_2}.
Note that Theorem \ref{simultaneous_thm_1} is the special case when $\TTT=\{t\}$.

\begin{proof}[Proof of Theorem \ref{simultaneous_thm_2}]
From \eqref{simultaneous_condition_2_species}, choose $t\in\TTT$ such that
\eeq{ \label{simultaneous_condition_2_consequence}
\frac{\partial\xi^s}{\partial q^t}(\Phi(q)) > 0 \qquad \text{whenever} \qquad \parbox{2.2in}{\centering $q\in\Supp(\zeta)$, $\min_{r\in\TTT}\Phi^r(q)>0$, and $\Phi^r(q)>0$ $\forall$ $r\in\SSS_\ext$.}
}
For ease of notation, let us say that $\SSS = \{1,\dots,n\}$, $t=n-1$, and $s=n$.
Given $a<b$ in $[0,1]$, consider the path from $\Phi(a)$ to $\Phi(b)$ which moves one coordinate at a time.
That is, the first coordinate is moved from $\Phi^1(a)$ to $\Phi^1(b)$, then the second coordinate from $\Phi^2(a)$ to $\Phi^2(b)$, and so on.
Let $\vphi^r$ be the restriction $\partial\xi^s/\partial q^r$ to the $r^\text{th}$ line segment in this path, which is just a function of the $r^\text{th}$ coordinate:
\eq{
\vphi^s_r(q) \coloneqq \frac{\partial \xi^s}{\partial q^r}\Big|_{\{\text{$q^r = q$, $q^{j}=\Phi^{j}(b)$ for $j<r$, $q^{j}=\Phi^{j}(a)$ for $j>r$}\}}, \quad q\in[\Phi^r(a),\Phi^r(b)].
}
By the fundamental theorem of calculus, we have
\eeq{ \label{preparatory_ineq}
\xi^s(\Phi(b)) - \xi^s(\Phi(a))
= \sum_{r=1}^n\int_{\Phi^r(a)}^{\Phi^r(b)}\vphi^s_r(q)\ \dd q.
}
Now suppose $a<b$ and $a,b\in\Supp(\zeta)$.
In particular, for any $r\in\SSS_\ext$, the identity \eqref{cs_identity} shows that $\Phi^r(b)>0$.
If $\Phi^r(a) < \Phi^r(b)$ for some $r\in\TTT$, then by hypothesis we have $\Phi^r(a)<\Phi^r(b)$ for all $r\in\TTT$.
In particular, we have $\Phi^r(b)>0$ for all $r\in\TTT\cup\SSS_\ext$.
Therefore, \eqref{simultaneous_condition_2_consequence} tells us that the $(n-1)^\text{th}$ summand in  \eqref{preparatory_ineq} is strictly positive.
Hence $\xi^s(\Phi(a))<\xi^s(\Phi(b))$, and then it follows from \eqref{cs_identity} that $\Phi^s(a)<\Phi^s(b)$.
We have thus argued that if $\Phi^r(a)<\Phi^r(b)$ for some $r\in\TTT$, then $\Phi^s(a)<\Phi^s(b)$.

To establish the reverse implication, we use \eqref{simultaneous_condition_2_species} to identify $t\in\TTT$ such that
\eeq{ \label{simultaneous_condition_2_consequence_2}
\frac{\partial\xi^t}{\partial q^s}(\Phi(q)) > 0 \qquad \text{whenever} \qquad \parbox{1.7in}{\centering $q\in\Supp(\zeta)$, $\Phi^s(q)>0$, and $\Phi^r(q)>0$ $\forall$ $r\in\SSS_\ext$.}
}
Then replace \eqref{preparatory_ineq} with 
\eq{
\xi^t(\Phi(b))-\xi^t(\Phi(a)) = \sum_{r=1}^n\int_{\Phi^r(a)}^{\Phi^r(b)}\vphi^t_r(q)\ \dd q.
}
If $\Phi^s(a)<\Phi^s(b)$, then the $n^\text{th}$ summand is strictly positive by \eqref{simultaneous_condition_2_consequence_2}. 
Hence $\xi^t(\Phi(a))<\xi^t(\Phi(b))$, which forces $\Phi^t(a)<\Phi^t(b)$ thanks to \eqref{cs_identity}.
\end{proof}

As promised, we close this section by proving Lemma  \ref{simultaneous_lemma}.

\begin{proof}[Proof of Lemma \ref{simultaneous_lemma}]
Let us write $\zeta^s =\zeta\circ(\Phi^s)^{-1}$.
If we take the convention
\eeq{ \label{inverse_convention}
(\Phi^s)^{-1}(u) \coloneqq \inf\{q\in\Supp(\zeta):\, \Phi^s(q) \geq u\}, \quad u\in[0,1],
}
then $\zeta^s\big([u,1]) = \zeta\big([(\Phi^s)^{-1}(u),1]\big)$ for all $u\in[0,1]$.
Now, for any two points $q_0\leq q_1$ in $\Supp(\zeta)$, we have
\eq{
(\Phi^s)^{-1}(\Phi^s(q_1)) \leq q_0 \quad 
&\stackref{inverse_convention}{\iff} \quad
\Phi^s(q_0) = \Phi^s(q_1)  \\
&\stackref{simultaneous_def}{\iff} \quad
\Phi^t(q_0) = \Phi^t(q_1)  \quad \stackref{inverse_convention}{\iff} \quad
(\Phi^t)^{-1}(\Phi^t(q_1)) \leq q_0.
}
We thus have $(\Phi^s)^{-1}\circ\Phi^s = (\Phi^t)^{-1}\circ\Phi^t$ on $\Supp(\zeta)$, and so on the domain $\Supp(\zeta^s)$, we have
\eq{
\Phi^s\circ(\Phi^t)^{-1}\circ\Phi^t\circ(\Phi^s)^{-1}
= \Phi^s\circ(\Phi^s)^{-1}\circ\Phi^s\circ(\Phi^s)^{-1}
= \mathrm{Id}.
}
By symmetry, the same statement holds with $s$ and $t$ exchanged, and so $\Phi^t\circ(\Phi^s)^{-1}$ and $\Phi^s\circ(\Phi^t)^{-1}$ are inverses of each other.
To see that these maps are measure-preserving, we simply use the definition of pushforward measures:
For any $u\in\Supp(\zeta^s)$, we have
\eq{
\zeta^t\big([(\Phi^t\circ(\Phi^s)^{-1})(u),1]\big)
&= \zeta\big([((\Phi^t)^{-1}\circ\Phi^t\circ(\Phi^s)^{-1})(u),1]\big) \\
&= \zeta\big([((\Phi^s)^{-1}\circ\Phi^s\circ(\Phi^s)^{-1})(u),1]\big) \\
&= \zeta\big([((\Phi^s)^{-1}\circ\mathrm{Id})(u),1]\big) 
= \zeta^s\big([u,1]\big).
}
Indeed, $\Phi^t\circ(\Phi^s)^{-1}:\Supp(\zeta^s)\to\Supp(\zeta^t)$ is measure-preserving.
\end{proof}


\section{Continuity of the Crisanti--Sommers functional} \label{continuity_proof}
The main goal of this section is to prove continuity of the functional $(\zeta,\Phi)\mapsto B(\zeta,\Phi)$, stated as Proposition \ref{continuity_prop} below.
To make the discussion precise, we consider the same metric space as in \cite{bates-sohn22}, which we now describe.
Identify any $\vc\lambda$-admissible pair $(\zeta,\Phi)$ with the pushforward measure $\zeta\circ\Phi^{-1}$ on the unit hypercube $[0,1]^\SSS$, which is equipped with the $\ell^1$ norm.
Then the distance between $(\zeta_1,\Phi_1)$ and $(\zeta_2,\Phi_2)$ is taken to be the Wasserstein-1 distance between $\zeta_1\circ\Phi_1^{-1}$ and $\zeta_2\circ\Phi_2^{-1}$.
Since each coordinate of $\Phi$ is non-decreasing, this distance has a convenient representation in terms of quantile functions.
Specifically, for a probability measure $\zeta$ on $[0,1]$, define the associated quantile function:
\eq{
Q_\zeta(z) \coloneqq \inf\{q\in[0,1]:\zeta\big([0,q]\big)\geq z\}, \quad z\in[0,1].
}
Then the Wasserstein distance we have just described is given by
\eeq{ \label{pseudometric_def}
\DD\big((\zeta_1,\Phi_1),(\zeta_2,{\Phi}_2)\big)
\coloneqq \int_0^1 \|\Phi_1(Q_{\zeta_1}(z))-{\Phi}_2(Q_{\zeta_2}(z))\|_1\ \dd z.
}
Formally, $\DD$ is a pseudometric on the space of $\vc\lambda$-admissible pairs $(\zeta,\Phi)$.

\begin{remark}
The fact that $\DD$ is a pseudometric rather than a metric underlines the fact that the condition of $\vc\lambda$-admissibility in \eqref{admissible_def} is somewhat artificial.
The ``true'' functional order parameter is the $\SSS$-tuple of measures $(\zeta\circ(\Phi^s)^{-1})_{s\in\SSS}$.
The $\vc\lambda$-admissible pair $(\zeta,\Phi)$ is a mechanism for coupling these measures together in a ``synchronized'' way; see the discussion in \cite[Sec.~1.3]{bates-sohn22}.
This coupling mechanism is canonical up to modifications of $\Phi$ off the support of $\zeta$.
\end{remark}

For $\bar q<1$, let $\AA(\bar q)$ denote the collection of $\vc\lambda$-admissible pairs such that for some $q_*\in[0,1)$, we have $\zeta\big([0,q_*]\big) = 1$ and $\Phi^s(q_*) \leq \bar q$ for all $s\in\SSS$. 
In pushforward notation, this means
\eeq{ \label{A_q_bar_def}
\AA(\bar q) \coloneqq \{(\zeta,\Phi):\, \Supp(\zeta\circ(\Phi^s)^{-1})\subset[0,\bar q] \text{ for each $s\in\SSS$}\}, \quad \bar q<1.
}
When we wish to refer to all pairs satisfying \eqref{gap_assumption} for some $q_*$, we will simply write 
\eeq{ \label{AA_def}
\AA \coloneqq \bigcup_{\bar q<1}\AA(\bar q)
= \{(\zeta,\Phi):\, \Supp(\zeta\circ(\Phi^s)^{-1})\subset[0,1) \text{ for each $s\in\SSS$}\}.
}
With these definitions, we can state our continuity result.

\begin{prop} \label{continuity_prop}
Assume \eqref{decay_condition} and $\bar q<1$.
On the set $\AA(\bar q)$, the map $(\zeta,\Phi)\mapsto B(\zeta,\Phi)$ is Lipschitz continuous with respect to $\DD$ (with a Lipschitz constant depending on $\bar q$).
\end{prop}

A key consequence is the weaker statement that $B(\zeta,\Phi)$ is invariant under different representations of the measure $\zeta\circ\Phi^{-1}$.
In other words, the functional $B$ is well-defined on the quotient space of $\AA$ obtained by identifying elements $(\zeta_1,\Phi_1)$ and $(\zeta_2,\Phi_2)$ such that $\DD\big((\zeta_1,\Phi_1),(\zeta_2,\Phi_2)\big) = 0$.
This fact is crucially used in the proof of Theorem \ref{cs_min_thm}, as it allows one to modify $\Phi$ anywhere not belonging to the support of $\zeta$, without changing the value of $B(\zeta,\Phi)$.

Our strategy for proving Proposition \ref{continuity_prop} is to restrict to measures with finite support, and then appeal to a density argument.
Since calculations are easier in the finite-support case, it will be advantageous for us to use this section as an opportunity to analyze how close a minimizer's support can be to $1$.
Indeed, since \eqref{gap_assumption} is not maintained under closure, it will be necessary for us to keep these supports separated from $1$.
This is accomplished by Lemma \ref{away_from_1_lemma}.



Given any $\bar q\in[0,1)$, let us consider $(\zeta,\Phi)\in\AA(\bar q)$ such that $\zeta$ is supported on finitely many points.
Every such pair corresponds to a sequence of weights
\begin{subequations}
\label{discrete_1}
\eeq{ \label{weights}
0 = m_0 < m_1 < \dots < m_k = 1,
}
together with sequences of points for each species:
\eeq{ \label{points}
0 = q_0^s \leq q_1^s \leq \cdots \leq q_k^s \leq \bar q < q_{k+1}^s = 1.
}
Namely, if we define the convex combination
\eeq{
q_r \coloneqq \sum_{s\in\SSS}\lambda^sq^s_r,
}
then \eqref{weights} and \eqref{points} collectively encode the measure
\eeq{ \label{discrete_zeta}
\zeta = \sum_{r=1}^k m_r\delta_{q_r},
}
\end{subequations}
where $\delta_x$ denotes the Dirac delta measure at $x$.
Furthermore, if $\vc q_r = (q_r^s)_{s\in\SSS}$, then $\Phi$ can be any $\vc\lambda$-admissible map such that $\Phi(q_r) = \vc q_r$ for each $r\in\{1,\dots,k\}$.
For instance, $\Phi$ could be the piecewise linear map satisfying these constraints.
Writing the quantities \eqref{delta_def} and \eqref{B_def} in terms of \eqref{weights} and \eqref{points}, we have
\eeq{ \label{discretized_delta_def}
\Delta^s_r &\coloneqq \Delta^s(q_r)
= \sum_{\ell=r}^k m_\ell(q_{\ell+1}^s-q_\ell^s),
}
\eeq{
B(\zeta,\Phi) = \sum_{s\in\SSS}\frac{\lambda^s}{2}\bigg(\red{h_s^2}\Delta^s_1+\frac{q_1^s}{\Delta_1^s}+\sum_{r=1}^{k-1}\frac{1}{m_r}\log\frac{\Delta_r^s}{\Delta_{r+1}^s}+\log\Delta_k^s&\bigg) \\
+ \frac{1}{2}\sum_{r=1}^k m_r(\xi(\vc q_{r+1})-\xi(\vc q_r)&).
\label{B_discrete_def}
}
Let us define
\eeq{ \label{deriv_discrete_def}
\delta_{r,\ell}^s \coloneqq \frac{\partial\Delta^s_r}{\partial q^{s}_\ell}
= m_{\ell-1}\one_{\{\ell>r\}}-m_\ell\one_{\{\ell\geq r\}}, \quad \ell\in\{1,\dots,k\},
}
so that differentiating \eqref{B_discrete_def} results in
\eq{ 
\frac{\partial B}{\partial q_\ell^s}
= \frac{\lambda^s}{2}\bigg(\red{h_s^2}\delta_{1,\ell}^s+\frac{\one_{\{\ell=1\}}}{\Delta_1^s}-\frac{q_1^s}{(\Delta_1^s)^2}\delta^s_{1,\ell}+\sum_{r=1}^{k-1}\frac{1}{m_r}\Big(\frac{\delta^s_{r,\ell}}{\Delta^s_r}-\frac{\delta^s_{r+1,\ell}}{\Delta^s_{r+1}}\Big)-\frac{\one_{\{\ell=k\}}}{\Delta_k^s}&\bigg)
\\
+ \frac{m_{\ell-1}-m_\ell}{2}\lambda^s\xi^s(\vc q_\ell&).
}
Making appropriate substitutions using \eqref{deriv_discrete_def}, we have
\eeq{ \label{Bderivative_2}
\frac{\partial B}{\partial q_\ell^s}
= \frac{\lambda^s}{2}(m_{\ell-1}-m_\ell)\bigg(\red{h_s^2}-\frac{q_1^s}{(\Delta_1^s)^2}+\sum_{r=1}^{\ell-1}\frac{1}{m_r}\Big(\frac{1}{\Delta^s_r}-\frac{1}{\Delta^s_{r+1}}\Big)+\xi^s(\vc q_\ell)\bigg).
}
Since $\Delta_r^s\geq\Delta_{r+1}^s\geq\cdots\geq\Delta_k^s=1-q_k^s\geq 1-\bar q$, we have
\eq{
\Big|-\frac{q_1^s}{(\Delta_1^s)^2}+\sum_{r=1}^{\ell-1}\frac{1}{m_r}\Big(\frac{1}{\Delta_r^s}-\frac{1}{\Delta_{r+1}^s}\Big)\Big|
&= \frac{q_1^s}{(\Delta_1^s)^2}+\sum_{r=1}^{\ell-1}\frac{q_{r+1}^s-q_{r}^s}{\Delta_r^s\Delta_{r+1}^s} \\
&\leq \frac{1}{(1-\bar q)^2}\Big(q_1^s+\sum_{r=1}^{\ell-1}(q_{r+1}^s-q_{r}^s)\Big)
\leq \frac{1}{(1-\bar q)^2}.
}
Now \eqref{Bderivative_2} reads as
\eeq{ \label{Bderivative_2_rewrite}
\frac{\partial B}{\partial q_\ell^s} = \frac{\lambda^s}{2}(m_{\ell}-m_{\ell-1})D_\ell^s(\vc q),
\quad \text{where} \quad
|D_\ell^s(\vc q)| \leq \red{h_s^2}+\frac{1}{(1-\bar q)^2} + \xi^s(\vc 1).
}
This identity results in the following precursor to Proposition \ref{continuity_prop}.

\begin{lemma} \label{continuity_lemma}
Having fixed \eqref{weights} and \eqref{points},
consider any sequences of the form
\begin{subequations}
\label{discrete_2}
\eeq{ \label{tilde_points}
0 = p_0^s \leq p_1^s \leq \cdots \leq p_k^s \leq \bar q < p_{k+1}^s= 1, \quad s\in\SSS.
}
Let $\vc p_\ell = (p_\ell^s)_{s\in\SSS}$ and $p_\ell = \sum_{s\in\SSS}\lambda^sp^s_\ell$, and then consider the measure
\eeq{
\zeta_2 = \sum_{\ell=1}^k m_\ell\delta_{p_\ell}.
}
\end{subequations}
Let $\Phi_2$ be any $\vc\lambda$-admissible map such that $\Phi_2^s(p_\ell) = p_\ell^s$ for each $\ell$ and $s$.
We then have
\eeq{ \label{lipschitz_ineq}
|B(\zeta,\Phi) - B(\zeta_2,\Phi_2)| \leq C\DD\big((\zeta,\Phi),(\zeta_2,\Phi_2)\big),
}
where $C$ is a constant depending only on $\xi$, $(h_s)_{s\in\SSS}$, and $\bar q$.
\end{lemma}

\begin{proof}
Let us begin by understanding the right-hand side of \eqref{lipschitz_ineq}.
Observe that
\eq{
Q_\zeta(z) = q_\ell \quad \text{and} \quad
Q_{\zeta_2}(z) = p_\ell \quad \text{for $z\in(m_{\ell-1},m_{\ell}]$, $1\leq \ell\leq k$}.
}
Since $\Phi(q_\ell) = \vc q_\ell$ and $\Phi_2(p_\ell) = \vc p_\ell$, we thus have
\eeq{ \label{DD_identity}
\DD\big((\zeta,\Phi),(\zeta_2,\Phi_2)\big)
&\stackref{pseudometric_def}{=}
\int_0^1 \|\Phi(Q_{\zeta}(z))-\Phi_2(Q_{\zeta_2}(z))\|_1\ \dd z \\
&\stackrefp{pseudometric_def}{=} \sum_{\ell=1}^k\int_{m_{\ell-1}}^{m_\ell}\|\vc q_\ell - \vc p_\ell\|_1\ \dd z
=\sum_{\ell=1}^k(m_\ell-m_{\ell-1})\|\vc q_\ell-\vc p_\ell\|_1.
}
This identity gives us a target as we next 
study the left-hand side of \eqref{lipschitz_ineq}.

Consider the linear interpolation between \eqref{discrete_1} and \eqref{discrete_2}:
\eq{
q_\ell^s(t) \coloneqq (1-t)q_\ell^s + tp_\ell^s, \qquad
\zeta_t \coloneqq (1-t)\zeta + t\zeta_2, \qquad
\Phi_t \coloneqq (1-t)\Phi + t\Phi_2, \qquad t\in[0,1].
}
By differentiating with the chain rule, we have
\eq{
&|B(\zeta,\Phi)-B(\zeta_2,\Phi_2)|
\stackrefp{DD_identity}{\leq} \sup_{t\in(0,1)}\Big|\frac{\dd B(\zeta_t,\Phi_t)}{\dd t}\Big| \\
&\stackrefpp{Bderivative_2_rewrite}{DD_identity}{=}\sup_{t\in(0,1)}\Big|\sum_{s\in\SSS}\sum_{\ell=1}^k\frac{\lambda^s}{2}(m_\ell-m_{\ell-1})D_{\ell}^s(\vc q(t))\frac{\dd q_\ell^s(t)}{\dd t}\Big| \\
&\stackrefp{DD_identity}{=} \sup_{t\in(0,1)}\Big|\sum_{s\in\SSS}\sum_{\ell=1}^k\frac{\lambda^s}{2}(m_\ell-m_{\ell-1})D_{\ell}^s(\vc q(t))(p_\ell^s- q_\ell^s)\Big| \\
&\stackrefpp{Bderivative_2_rewrite}{DD_identity}{\leq} \Big(\max_{s\in\SSS}\frac{\lambda^s}{2}\Big[\red{h_s^2}+\frac{1}{(1-\bar q)^2}+\xi^s(\vc 1)\Big]\Big)\sum_{\ell=1}^k(m_\ell-m_{\ell-1})\sum_{s\in\SSS}|p_\ell^s-q_\ell^s| \\
&\stackref{DD_identity}{=} C\DD\big((\zeta,\Phi),(\zeta_2,\Phi_2)\big).
}
We have proved the desired Lipschitz inequality \eqref{lipschitz_ineq}.
\end{proof}
%

Let us pause to obtain another consequence of the derivative calculation \eqref{Bderivative_2}.
Let $\PP_k$ denote the set of probability measures on $[0,1]$ which are supported on at most $k$ points.
Recall the set $\AA(\bar q)$ from \eqref{A_q_bar_def}.
For $\bar q<1$, consider the following subset: 
\eq{
\AA_k(\bar q) \coloneqq \{(\zeta,\Phi)\in\AA(\bar q):\, \zeta\in\PP_k\}.
}
This is exactly the set of $\vc\lambda$-admissible pairs of the form \eqref{discrete_1}.
Also define 
\eeq{ \label{AA_k_def}
\AA_k \coloneqq \bigcup_{\bar q<1}\AA_k(\bar q).
}
We then have the following result, which will ultimately lead to the existence of a minimizer claimed in Theorem \ref{cs_min_thm}. 

\begin{lemma} \label{away_from_1_lemma}

There exists $\bar q<1$ such that for any positive integer $k$, we have
\eeq{ \label{mins_away_from_1}
\inf_{(\zeta,\Phi)\in\AA_k(\bar q)}B(\zeta,\Phi)
= \inf_{(\zeta,\Phi)\in\AA_k}B(\zeta,\Phi).
}
More precisely, we can take
\eeqs{ \label{bar_q_bound}
\bar q &= \max_{s\in\SSS}\frac{(1-u^s)(h_s^2+\xi^s(\vc 1))+u^s}{(1-u^s)(h_s^2+\xi^s(\vc 1))+1}, \quad \text{where} \\
u^s &\coloneqq 1-\frac{\sqrt{1+4(h_s^2+\xi^s(\vc 1))}-1}{2(h_s^2+\xi^s(\vc 1))}. \label{u_s_def}
}



\end{lemma}

\begin{proof}
For the sake of argument, let us temporarily fix the sequence \eqref{weights} and vary only the elements of \eqref{points}.
Moreover, we relex \eqref{points} to
\eeq{ \label{points_no_q_bar}
0 = q_0^s \leq q_1^s \leq \cdots \leq q_k^s \leq q_{k+1}^s = 1, \quad s\in\SSS.
}
Consider the following summation by parts:
\eq{
\frac{q_1^s}{\Delta_1^s}+\sum_{r=1}^{k-1}\frac{1}{m_r}\log\frac{\Delta_r^s}{\Delta_{r+1}^s}+\log\Delta_k^s
= \frac{q_1^s}{\Delta_1^s}-\frac{1}{m_1}\log\frac{1}{\Delta_1^s}+\sum_{r=1}^{k-1}\Big(\frac{1}{m_r}-\frac{1}{m_{r+1}}\Big)\log\frac{1}{\Delta_{r+1}^s}.
}
Since $\Delta_1^s \leq 1-q_1^s$, it is clear that the right-hand side diverges to $\infty$ as $q_1^s\nearrow1$, uniformly in $(q_r^s)_{r\geq2}$.
So in order to realize a minimal value for $B$, we may assume $q_1^s$ is at most some fixed number $u^s<1$. 
But then the expression displayed above is at least
\eq{
\Big(\frac{1}{m_{k-1}}-1\Big)\log\frac{1}{1-q_k^s}
- \frac{1}{m_1}\log\frac{1}{1-u^s}.
}
Clearly this quantity diverges to $\infty$ as $q_k^s\nearrow1$, and so a minimal value is achieved only when $q_k^s$ is at most some fixed number $\bar q^s < 1$.
Upon taking $\bar q = \max_{s\in\SSS}\bar q^s$, we have argued that the minimum value of $B$ over all sequences \eqref{points_no_q_bar} must be obtained on some collection of the form \eqref{points}.

What remains to be shown is that $\bar q$ can be chosen independently of the sequence \eqref{weights}.
Observe from \eqref{Bderivative_2} that
\eq{ 
\frac{2}{\lambda^s(m_k-m_{k-1})}\frac{\partial B}{\partial q_k^s}
&= -\red{h_s^2}+\frac{q_1^s}{(\Delta_1^s)^2}+\sum_{r=1}^{k-1}\frac{1}{m_{r}}\Big(\frac{1}{\Delta_{r+1}^s}-\frac{1}{\Delta_{r}^s}\Big) - \xi^s(\vc q_k).
}
Again because $\Delta_r^s \leq 1 - q_r^s$, the right-hand side is at least
\eeq{ \label{just_q1}
-\red{h_s^2}+\frac{q_1^s}{(1-q_1^s)^2} - \xi^s(\vc 1).
}
If $q_1^s$ is sufficiently close to one, or more specifically $q_1^s$ exceeds the value $u^s$ given in \eqref{u_s_def},
then \eqref{just_q1} is positive, meaning we are not at a minimum of $B$.
Therefore, any minimum must have $q_1^s \leq u^s$, and so
\eq{
\sum_{r=1}^{k-1}\frac{1}{m_r}\Big(\frac{1}{\Delta_{r+1}^s}-\frac{1}{\Delta_r^s}\Big)
\stackref{discretized_delta_def}{=} \sum_{r=1}^{k-1}\frac{q_{r+1}^s-q_r^s}{\Delta_{r+1}^s\Delta_{r}^s}
&\geq \sum_{r=1}^{k-1}\frac{(1-q_r^s)-(1-q_{r+1}^s)}{(1-q_{r+1}^s)(1-q_{r}^s)} \\
&= \sum_{r=1}^{k-1}\Big(\frac{1}{1-q_{r+1}^s}-\frac{1}{1-q_{r}^s}\Big) \\
&= \frac{1}{1-q_k^s} - \frac{1}{1-q_1^s}
\geq \frac{1}{1-q_k^s} - \frac{1}{1-u^s}.
}
If $q_k^s$ is larger than the value $\bar q$ given in \eqref{bar_q_bound}, then this quantity is larger than $\red{h_s^2}+\xi^s(\vc 1)$, which in light of \eqref{Bderivative_2} would again rule out the possibility of a critical point.
This conclusion, combined with the earlier argument that every sequence \eqref{weights} admits a minimizer, yields \eqref{mins_away_from_1}.
%
\end{proof}

To obtain Proposition \ref{continuity_prop} from Lemma \ref{continuity_lemma}, we just need to approximate an arbitrary $\zeta$ with measures supported on finitely many points.
This is accomplished through the following result.

\begin{lemma} \label{density_lemma}
Fix $\bar q\in[0,1)$ and assume $(\zeta,\Phi)\in\AA(\bar q)$.
Then for any $\eps_1,\eps_2>0$, there is a measure $\wt\zeta$ of the form \eqref{discrete_zeta}, i.e.~ $(\wt\zeta,\Phi)\in\AA_k(\bar q)$ for some finite $k$, such that the following inequalities hold:
\eeqs{
\DD\big((\zeta,\Phi),(\wt\zeta,\Phi)\big)
&\leq \eps_1, \label{need_for_density_1} \\
|B(\zeta,\Phi)-B(\wt\zeta,{\Phi})| &\leq \eps_2. \label{need_for_density_2}
}
\end{lemma}

During the proof of Lemma \ref{density_lemma}, we will use the following integration by parts identity.

\begin{lemma} \label{ibp_lemma}
\textup{\cite[Lem.~2.21]{bates-sohn22}} 
For any Borel probability measure $\zeta$ on $[0,1]$, any Lipschitz continuous, non-decreasing function $f\colon[0,1]\to[0,\infty)$, and any $q\in[0,1]$, we have
\eeq{ \label{by_parts_with_quantiles}
\int_{q}^1 \zeta\big([0,u]\big)f'(u)\ \dd u
= f(1)-\zeta\big([0,q]\big)f(q)-\int_{\zeta([0,q])}^1 f(Q_\zeta(z))\ \dd z.
}
In particular,
\eeq{ \label{by_parts_with_quantiles_0}
\int_0^1 \zeta\big([0,u]\big)f'(u)\ \dd u
= f(1) - \int_{0}^1 f(Q_\zeta(z))\ \dd z.
}
\end{lemma}

\begin{proof}[Proof of Lemma \ref{density_lemma}]
The argument is virtually identical to that of \cite[Prop.~2.17]{bates-sohn22}, but we include it for the reader's convenience.
Let $(\zeta,\Phi)\in\AA(\bar q)$ be given. 
That is, there is some $q_*\in[0,1)$ such that $\zeta\big([0,q_*]\big) = 1$ and $\Phi^s(q_*)\leq\bar q$ for all $s\in\SSS$.
We fix $q_*$ and $\bar q$ for the remainder of the proof.

Given any $\eps_1>0$, let $L$ be an integer so large that 
\eeq{ \label{L_choice_1}
\frac{1}{L}\sum_{s\in\SSS}\frac{1}{\lambda^s} \leq \eps_1.
}
The left-hand side is motivated by the fact that for any $\Phi$ satisfying Definition \ref{lambda_av_def}, we have
\eeq{ \label{lambda_av_consequence}
|\Phi^s(q)-\Phi^s(u)| \leq |q-u|/\lambda^s
\quad \text{for any $q,u\in[0,1]$}.
}
Let $J$ be the smallest integer such that $J/L\geq q_*$.
Based on $\zeta$, we choose a sequence 
\eeq{ \label{tilde_qs}
0= q_0 \leq  q_1 < \cdots <  q_k \leq q_* < q_{k+1}= 1
}
in the following manner:
\begin{itemize}
\item If $\zeta(\{0\}) > 0$, then set $q_1 = 0$.
\item For $j\in\{1,\dots,J-1\}$, if $\zeta\big(((j-1)/L,j/L]\big)>0$, then include $ q=j/L$ as one of the elements $ q_r$ of \eqref{tilde_qs}.
\item Finally, if $\zeta\big(((J-1)/L,J/L]\big)>0$, then set $q_k = q_*$ (otherwise, $q_k$ will be the largest number obtained from the two previous steps).
\end{itemize}
Once \eqref{tilde_qs} has been formed, define $m_r = \zeta\big([0, q_r]\big)$ for $r\in\{1,\dots,k\}$. 
The condition that $\zeta$ assign positive mass to the interval $(q_r-1/L,q_r]$ ensures that
\eq{ 
0=m_0 < m_1 < \cdots < m_k = 1.
}
Furthermore, since all zero-mass intervals are excluded, we have
\eeq{ \label{why_L}
 q_r - 1/L \leq Q_{\zeta}(z) \leq  q_r \quad \text{whenever $z\in(m_{r-1},m_r]$, $1\leq r\leq k$}.
}
Equivalently, the following implication is true:
\eeq{ \label{why_L_equivalent}
q_r \leq u \leq q_{r+1}- 1/L \quad \implies \quad \zeta\big([0,u]\big)=\zeta\big([0,q_r])=m_r.
}
Now take the approximating measure to be
\eq{
\wt\zeta = \sum_{r=1}^k(m_r-m_{r-1})q_r.
}
As before, given ${\Phi}$ we will write $\vc q_r=\Phi(q_r)$ so that for $z\in(m_{r-1},m_r]$, we have
\eeq{ \label{why_L_2}
\|{\Phi}(Q_{\zeta}(z))-\vc q_r\|
\stackref{lambda_av_consequence}{\leq} | Q_{\zeta}(z)-q_r|\sum_{s\in\SSS}\frac{1}{\lambda^s}
\stackref{why_L}{\leq} \frac{1}{L}\sum_{s\in\SSS}\frac{1}{\lambda^s}
\stackref{L_choice_1}{\leq} \eps_1.
}
Since $Q_{\wt\zeta}(z) = q_r$ for $z\in(m_{r-1},m_{r}]$, this inequality leads to
\eeq{ \label{L_choice_consequence}
\int_0^1\|\Phi(Q_\zeta(z))-{\Phi}(Q_{\wt\zeta}(z))\|_1\ \dd z
&= \sum_{r=1}^k \int_{m_{r-1}}^{m_r} \|{\Phi}(Q_{ \zeta}(z))-\vc q_r\|_1\ \dd z
\stackref{why_L_2}{\leq} \eps_1.
}
Finally, note that 
\eq{
\Supp(\wt\zeta\circ(\Phi^s)^{-1})=\{\Phi^s(q_1),\dots,\Phi^s(q_k)\}\subset[0,\Phi^s(q_*)]\subset[0,\bar q] \quad \text{for all $s\in\SSS$},
}
and so $(\wt\zeta,\Phi)\in\AA_k(\bar q)$.
This completes the proof of \eqref{need_for_density_1}.

Now we turn our attention to showing \eqref{need_for_density_2}. 
In order to distinguish between \eqref{delta_def} applied to $(\wt\zeta,{\Phi})$ as opposed to $(\zeta,\Phi)$, we will write
\eq{
\wt\Delta^s_r \coloneqq \int_{q_r}^1\zeta\big([0,u]\big)(\Phi^s)'(u)\ \dd u
= \sum_{\ell=r}^km_\ell(q_{\ell+1}^s-q_\ell^s).
}
Note that because $\zeta\big([0,q_k])=1=\wt\zeta\big([0,q_k]\big)$, we have
\eeq{ \label{comparing_deltas_at_q_star}
\Delta^s(q_k) = 1 - \Phi^s(q_k) = \wt\Delta^s_k.
}
Set $\alpha = \min_{s\in\SSS}(1-\Phi^s(q_k))$.
Given $\eps_2>0$, let $\eps_1\in(0,\alpha/7)$ be so small that
\eeq{ \label{eps_alpha_choice}
\frac{1}{\alpha-7\eps_1}-\frac{1}{\alpha} \leq \frac{\eps_2}{\max_{s\in\SSS}1/\lambda^s}.
}
Given $\eps_1$, take $L$ as above so that \eqref{L_choice_consequence} holds and whenever $|q-u|\leq1/L$, we have
\eeq{ \label{summarizing_choices}
|\Phi^s(q)-\Phi^s(u)|
\stackref{lambda_av_consequence}{\leq}
|q-u|\sum_{s\in\SSS}\frac{1}{\lambda^s}
\stackref{L_choice_1}{\leq} \eps_1.
}
%
Regardless of $\eps_1$, a simple calculus exercise shows
\eq{
\sup_{y\geq \alpha,x\in[-7\eps_1,7\eps_1]}\Big|\frac{1}{y-x}-\frac{1}{y}\Big| 
= \sup_{y\geq \alpha}\Big(\frac{1}{y-7\eps_1}-\frac{1}{y}\Big) 
= \frac{1}{\alpha-7\eps_1}-\frac{1}{\alpha}.
}
Since $\Delta^s(q)\geq \Delta^s(q_k)\geq\alpha$ for any $q\in[0,q_k]$ and $s\in\SSS$, it thus follows from \eqref{eps_alpha_choice} and \eqref{lambda_av_consequence} that
\eeq{ \label{uniform_inequality}
\sup_{x\in[-7\eps_1,7\eps_1]}\Big|\frac{(\Phi^s)'(q)}{\Delta^s(q)-x}-\frac{(\Phi^s)'(q)}{\Delta^s(q)}\Big| \leq \eps_2 \quad \text{whenever $q\in[0,q_*]$ and $(\Phi^s)'(q)$ exists}.
}
%
In light of \eqref{why_L_equivalent}, we have the following for $q\in[q_{r}, q_{r+1}]$, $0\leq r\leq k-1$:
\eq{
\Delta^s(q) 
&= \Delta^s(q_{r+1})
+\int_{q}^{(q_{r+1}-\frac{1}{L})\vee q}\zeta\big([0,u]\big)(\Phi^s)'(u)\ \dd u 
+\int_{(q_{r+1}-\frac{1}{L})\vee q}^{q_{r+1}}\zeta\big([0,u]\big)(\Phi^s)'(u)\ \dd u\\
&= \Delta^s(q_{r+1}) + m_r\big(\Phi^s((q_{r+1}-1/L)\vee q)-\Phi^s(q)\big)
+\int_{(q_{r+1}-\frac{1}{L})\vee q}^{q_{r+1}}\zeta\big([0,u]\big)(\Phi^s)'(u)\ \dd u.
}
Now, it is immediate from \eqref{summarizing_choices} that
\eq{
|\Phi^s(q_{r+1})-\Phi^s((q_{r+1}-1/L)\vee q)| \leq \eps_1.
}
In addition, by using the trivial inequality $0\leq\zeta\big([0,u]\big)\leq 1$, we obtain
\eq{
0\leq \int_{(q_{r+1}-\frac{1}{L})\vee q}^{q_{r+1}}\zeta\big([0,u]\big)(\Phi^s)'(u)\ \dd u
\leq \Phi^s(q_{r+1})-\Phi^s((q_{r+1}-1/L)\vee q) \leq \eps_1.
}
Since we defined $q_{r+1}^s$ to be ${\Phi}^s(q_{r+1})$, the three previous displays together show
\eeq{ \label{integrate_to_find_mr}
\big|\Delta^s(q) - \Delta^s( q_{r+1})-m_r\big(q_{r+1}^s -{\Phi}^s(q)\big)\big| \leq 2\eps_1
\quad \text{for $q\in[q_r,q_{r+1}]$.}
}
Next recall from \eqref{why_L} that $q_{r+1}-1/L \leq Q_{\zeta}(m_{r+1})\leq  q_{r+1}$.
Therefore, by yet another application of \eqref{summarizing_choices}, we have
\eeq{ \label{q_r1_to_Phi}
0 \leq q_{r+1}^s-\Phi^s(Q_\zeta(m_{r+1})) \leq \eps_1.
}
Using $Q_{\zeta}(m_{r+1})$ as the value of $q$ in \eqref{integrate_to_find_mr}, we now obtain the following special case:
\eq{
|\Delta^s(Q_{\zeta}(m_{r+1}))-\Delta^s(q_{r+1})|
\leq 3\eps_1, \quad 0\leq r\leq k-1.
}
Since $Q_{\wt\zeta}(m_{r+1}) = q_{r+1}$, we can employ Lemma \ref{ibp_lemma} to make the following comparison:
\eq{
&|\Delta^s(Q_{\zeta}(m_{r+1}))- \wt\Delta^s_{r+1}| \\
&\stackrefp{q_r1_to_Phi}{=}
\Big|\int_{Q_\zeta(m_{r+1})}^1\zeta\big([0,u]\big)(\Phi^s)'(u)\ \dd u - \int_{Q_{\wt\zeta}(m_{r+1})}^1\wt\zeta\big([0,u]\big)(\Phi^s)'(u)\ \dd u\Big| \\
&\stackrefpp{by_parts_with_quantiles}{q_r1_to_Phi}{\leq} \Big|\int_{m_{r+1}}^1\Phi^s(Q_\zeta(z))\ \dd z-\int_{m_{r+1}}^1\Phi^s(Q_{\wt \zeta}(z))\ \dd z\Big|
+m_{r+1}\big|\Phi^s(Q_{\zeta}(m_{r+1}))-q_{r+1}^s\big| \\
&\stackref{q_r1_to_Phi}{\leq} \int_{m_{r+1}}^1\|\Phi(Q_\zeta(z))-{\Phi}(Q_{\wt \zeta}(z))\|_1\ \dd z + \eps_1
\stackref{L_choice_consequence}{\leq} 2\eps_1.
}
The two previous displays combine to show that
\eeq{ \label{close_on_qs}
|\Delta^s( q_{r+1}) - \wt\Delta^s_{r+1}| \leq 5\eps_1, \quad 0\leq r\leq k-1.
}
Putting together \eqref{integrate_to_find_mr} and \eqref{close_on_qs}, we find
\eq{
\big|\Delta^s(q) - \wt\Delta^s_{r+1} - m_r\big(q_{r+1}^s-\Phi^s(q)\big)\big| \leq 7\eps_1 \quad \text{for all $q\in[q_r, q_{r+1}]$, $0\leq r\leq k-1$}.
}
It thus follows from \eqref{uniform_inequality} that whenever ${\Phi}'(q)$ exists and $q\in[q_r, q_{r+1}]$, we have
\eeq{ \label{follows_from_uniform_inequality}
\bigg|\frac{(\Phi^s)'(q)}{\Delta^s(q)}
-\frac{(\Phi^s)'(q)}{\wt\Delta^s_{r+1}+ m_{r}\big(q_{r+1}^s-\Phi^s(q)\big)}\bigg| \leq \eps_2.
}
Upon integration, this inequality yields the following for $r\in\{1,\dots,k-1\}$:
\eq{
\bigg|\int_{ q_{r}}^{ q_{r+1}}
\frac{(\Phi^s)'(q)}{\Delta^s(q)}\ \dd q
- \frac{1}{m_{r}}\log\frac{\wt\Delta^s_{r}}{\wt\Delta^s_{r+1}}\bigg|
\leq \eps_2( q_{r+1}- q_r).
}
When $r=0$, we have $m_0=0$, and so our conclusion from \eqref{follows_from_uniform_inequality} is instead
\eq{
\bigg|\int_0^{ q_{1}}
\frac{(\Phi^s)'(q)}{\Delta^s(q)}\ \dd q - \frac{q_1^s}{\wt\Delta^s_1}\bigg|
\leq \eps_2 q_1.
}
We conclude from the two previous displays, together with \eqref{comparing_deltas_at_q_star}, that
\eeq{ \label{for_density_final_1}
\bigg|\int_0^{q_k} \frac{(\Phi^s)'(q)}{\Delta^s(q)}\ \dd q+\log\Delta^s(q_k)
- \bigg(\frac{q_1^s}{\wt\Delta_1^s}+\sum_{r=1}^{k-1}\frac{1}{m_r}\log\frac{\wt\Delta^s_{r}}{\wt\Delta^s_{r+1}}+\log\wt\Delta^s_k\bigg)\bigg| \leq q_k\eps_2<\eps_2.
}
In addition, we can apply Lemma \ref{ibp_lemma} once more, specifically \eqref{by_parts_with_quantiles_0}, to see that
\eeq{ \label{for_density_final_1plus}
|\Delta^s(0)-\Delta^s_1|&=\Big|\int_0^1 \zeta\big([0,q]\big)(\Phi^s)'(q)\ \dd q-\int_0^1 \wt\zeta\big([0,q]\big)(\Phi^s)'(q)\ \dd q\Big| \\
&= \Big|\int_{0}^1 \Phi^s(Q_{\zeta}(z))\ \dd z-\int_0^1 \Phi^s(Q_{\wt\zeta}(z))\ \dd z\Big| \\
&\leq \int_0^1\|\Phi(Q_{\zeta}(z))-{\Phi}(Q_{\wt\zeta}(z))\|_1\ \dd z
\stackref{L_choice_consequence}{\leq} \eps_1.
}
Finally, we use \eqref{by_parts_with_quantiles_0} again to determine that
\eeq{ \label{for_density_final_2}
&\Big|\int_0^1 \zeta\big([0,q]\big)(\xi\circ \Phi)'(q)\ \dd q-\int_0^1 \wt\zeta\big([0,q]\big)(\xi\circ {\Phi})'(q)\ \dd q\Big| \\
&= \Big|\int_{0}^1 \xi(\Phi(Q_{\zeta}(z)))\ \dd z-\int_0^1 \xi({\Phi}(Q_{\wt\zeta}(z)))\ \dd z\Big| \\
&\leq C\int_0^1\|\Phi(Q_{\zeta}(z))-{\Phi}(Q_{\wt\zeta}(z))\|_1\ \dd z
\stackref{L_choice_consequence}{\leq} C\eps_1,
}
where $C$ depends only on $\xi$.
Once we recall the definition \eqref{B_def} of $B(\zeta,\Phi)$ and the value of $B(\wt\zeta,\Phi)$ from \eqref{B_discrete_def},
it follows from \eqref{for_density_final_1}--\eqref{for_density_final_2}
that
\eeq{ \label{last_inequality_for_density}
|B(\zeta,\Phi)-B(\wt\zeta,\Phi)| \leq \frac{1}{2}\Big(\eps_2 + \sum_{s\in\SSS}\lambda^s\red{h_s^2}\eps_1 + C\eps_1\Big).
}
By replacing $\eps_1$ with
\eq{
\min\Big\{\eps_1,\frac{\eps_2}{2}\Big(\sum_{s\in\SSS}\red{\lambda^sh_s^2} \Big)^{-1},\frac{\eps_2}{2C}\Big\},
}
we can ensure that the right-hand side of \eqref{last_inequality_for_density} is at most $\eps_2$, as needed for \eqref{need_for_density_2}.
\end{proof}

We are now ready to prove that $(\zeta,\Phi)\mapsto B(\zeta,\Phi)$ is locally Lipschitz with respect to $\DD$.

\begin{proof}[Proof of Proposition \ref{continuity_prop}]
Let $(\zeta_1,\Phi_1),(\zeta_2,\Phi_2)\in\AA(\bar q)$.
Given any $\eps>0$, use Lemma \ref{density_lemma} to identify finitely supported measures $\wt\zeta_1$ and $\wt\zeta_2$ such that
\eq{ 
\DD\big((\zeta_i,\Phi_i),(\wt\zeta_i,{\Phi}_i)\big) \leq \eps \quad \text{and} \quad
|B(\zeta_i,\Phi_i)-B(\wt\zeta_i,{\Phi}_i)|\leq \eps \quad \text{for $i\in\{1,2\}$},
}
and also $(\wt\zeta_i,\Phi_i)\in\AA(\bar q)$.
We may assume that $(\wt\zeta_1,\Phi_1)$ and $(\wt\zeta_2,\Phi_2)$ are of the form \eqref{discrete_1} and \eqref{discrete_2}, respectively, by possibly adding duplicate $q$'s and $p$'s in \eqref{points} and \eqref{tilde_points} so that the two representations use the same sequence \eqref{weights}.
Having reduced to this case, we appeal to Lemma \ref{continuity_lemma} to conclude that
\eq{
|B(\zeta_1,\Phi_1)-B(\zeta_2,\Phi_2)|
\leq 2\eps + C\Big(\DD\big((\zeta_1,\Phi),(\zeta_2,\Phi_2))+2\eps\Big).
}
By sending $\eps\to0$, we obtain the desired Lipschitz continuity.
\end{proof}

\section{Proof of minimizer identities} \label{identities_proof}
In this section, we prove Theorems \ref{parisi_min_thm} and \ref{cs_min_thm}.

\subsection{Identity satisfied by Parisi minimizers} \label{parisi_min_proof}

Here we prove Theorem \ref{parisi_min_thm}.
First we show that $A$ admits a minimizer, assuming only \eqref{decay_condition}.
The result \cite[Cor.~1.6]{bates-sohn22} says that there exists a $\vc\lambda$-admissible pair $(\zeta,\Phi)$ such that
\eeq{ \label{inf_b_inf}
\inf_{\vc b}A(\zeta,\Phi,\vc b) = \inf A.
}
The task now is to prove that there is $\vc b$ achieving this infimum, for which the following statement suffices.
For any $s\in\SSS$, the following inequalities hold uniformly in $(b^t)_{t\neq s}$:
\eeq{ \label{b_derivative_needs}
\frac{\dd A(\zeta,\Phi,\vc b)}{\dd b^s} &> 0 \quad \text{for all $b^s$ sufficiently large, and}\\
\frac{\dd A(\zeta,\Phi,\vc b)}{\dd b^s} &< 0 \quad \text{for all $b^s$ sufficiently close to $d^s(0)$.}
}
To this end, we can differentiate both sides of \eqref{A_def} to obtain
\eeq{  \label{b_derivative_calculation}
\frac{2}{\lambda^s}\frac{\dd A(\zeta,\Phi,\vc b)}{\dd b^s}
= -\frac{h_s^2+\red{\xi^s(\vc 0)}}{(b^s-d^s(0))^2}+ 1 - \frac{1}{b^s}-\int_0^1\frac{(\xi^s\circ\Phi)'(q)}{(b^s-d^s(q))^2}\ \dd q.
}
The right-hand side clearly tends to $1$ as $b^s\to\infty$, and so the first line of \eqref{b_derivative_needs} is true.
If $d^s(0)<1$, then the second line is also true, since in this case $1-1/b^s < 0$ for all $b^s$ sufficiently close to $d^s(0)$.
If instead $d^s(0)\geq 1$, then consider the point
\eq{
q_0 = \sup\{q\geq0:d^s(q)=d^s(0)\}.
}
As $d^s$ is non-increasing and continuous with $d^s(1)=0$, we must have $q_0<1$ and $d^s(q_0)=d^s(0)$.
By maximality of $q_0$, given any $\eps\in(0,1]$, we can identify $\delta>0$ such that $d^s(q_0+\delta)=d^s(q_0)-\eps$.
Therefore, if we choose $b^s = d^s(q_0)+\eps$,
then
\eq{
\int_{q_0}^{q_0+\delta}\frac{(\xi^s\circ\Phi)'(q)}{(b^s-d^s(q))^2}\ \dd q
&\geq \frac{1}{(2\eps)^2}\int_{q_0}^{q_0+\delta}(\xi^s\circ\Phi)'(q)\ \dd q \\
&\geq \frac{1}{(2\eps)^2}\int_{q_0}^{q_0+\delta}\zeta\big([0,q]\big)(\xi^s\circ\Phi)'(q)\ \dd q
= \frac{1}{4\eps}.
}
Consequently, for any $\eps<1/4$, the right-hand side of \eqref{b_derivative_calculation} is negative.
We have thus demonstrated the second line of \eqref{b_derivative_needs}, and so a minimizing $\vc b$ exists.
Furthermore, from \eqref{b_derivative_calculation} it is clear that any such $\vc b$ satisfies \eqref{parisi_identity_1}.

Now we assume \eqref{xi_strictly_convex} and look to prove that any minimizing triple $(\zeta,\Phi,\vc b)$ satisfies  \eqref{parisi_identity_2}.
That is, if we define
\eeq{ \label{little_phi_def}
\phi^s(q) \coloneqq \frac{h_s^2+\red{\xi^s(\vc0)}}{(b^s-d^s(0))^2}+\int_0^q \frac{(\xi^s\circ\Phi)'(q)}{(b^s-d^s(u))^2}\ \dd u,
}
then we wish to show that $\phi^s(q) = \Phi^s(q)$ for all $q\in\Supp(\zeta)$.
So let us write $\phi = (\phi^s)_{s\in\SSS}$ and consider the following function (mapping $[0,1]$ into $[0,1]^\SSS$) for any $\eps\in[0,1)$:
\eeq{ \label{wt_Phi_def}
\wt\Phi_\eps \coloneqq (1-\eps)\Phi + \eps\phi.
}
Now, each coordinate map $\wt\Phi_\eps^s = (1-\eps)\Phi^s+\eps\phi^s$ is non-decreasing, nonnegative, and satisfies the inequality
\eq{
\wt\Phi_\eps^s(1) = (1-\eps) + \eps\Big(1-\frac{1}{b^s}\Big) < 1.
}
In particular, $\Phi_\eps$ is not $\vc\lambda$-admissible.
To amend this, we define the scalar function
\eq{
\alpha_\eps(q) \coloneqq \sum_{s\in\SSS}\lambda^s\wt\Phi^s_\eps(q), \quad q\in[0,1],
}
which is clearly continuous, and also strictly increasing in $q$ because $\Phi$ is $\vc\lambda$-admissible.
Therefore, we can reparameterize $\wt\Phi_\eps$ as
\eq{
\Phi_\eps(q) \coloneqq \wt\Phi_\eps(\alpha_\eps^{-1}(q)) \quad \text{for $q\in[\alpha_\eps(0),\alpha_\eps(1)]$}.
}
On the intervals $[0,\alpha_\eps(0)]$ and $[\alpha_\eps(1),1]$, simply perform a linear interpolation 
to the endpoints $\Phi_\eps(0) = \vc 0$ and $\Phi_\eps(1) = \vc 1$.
By definition of $\alpha_\eps$, this new map $\Phi_\eps$ is $\vc\lambda$-admissible.
Corresponding to this reparameterization, we also define a new probability measure $\zeta_\eps$ by
\eeq{ \label{new_zeta}
\zeta_\eps\big([0,q]\big) \coloneqq \begin{cases}
0 &\text{if $q\in[0,\alpha_\eps(0))$}, \\
\zeta\big([0,\alpha_\eps^{-1}(q)]\big) &\text{if $q\in[\alpha_\eps(0),\alpha_\eps(1))$}, \\
1 &\text{if $q\in[\alpha_\eps(1),1]$}.
\end{cases}
}
The claimed identity \eqref{parisi_identity_2} will follow from the following claim.

\begin{claim} \label{parisi_derivative_claim}
We have the following right derivative:
\eeq{ \label{parisi_perturb_derivative}
\frac{\dd A(\zeta_\eps,\Phi_\eps,\vc b)}{\dd\eps}\Big|_{\eps=0^+} = -\frac{1}{2}\int\iprod{\Phi(q)-\phi(q)}{\nabla^2\xi(\Phi(q))(\Phi(q)-\phi(q))}\ \zeta(\dd q).
}
\end{claim}

Before checking Claim \ref{parisi_derivative_claim}, let us explain how to complete the proof of Theorem \ref{parisi_min_thm}.
Since $(\zeta,\Phi,\vc b)$ is a minimizer, the left-hand side of \eqref{parisi_perturb_derivative} is nonnegative.
But in light of \eqref{xi_strictly_convex}, the only way the right-hand side can be nonnegative is if
\eeq{ \label{equality_ae}
\Phi(q)=\phi(q) \quad \text{for $\zeta$-a.e. $q\neq0$}.
}
In particular, if $\zeta(\{q\})>0$ and $q\neq0$, then $\Phi(q)=\phi(q)$.
If $q\in\Supp(\zeta)$ but $\zeta(\{q\})=0$, then  \eqref{equality_ae} tells us that $q$ is a limit point of the locations at which $\Phi$ and $\phi$ coincide.
Since $\Phi$ and $\phi$ are both continuous, this is enough to  conclude $\Phi(q)=\phi(q)$.
Similarly, if $\Supp(\zeta)$ contains positive numbers arbitrarily close to $0$, then $\Phi(0)=\phi(0)$.

The only remaining scenario to consider is when $\Supp(\zeta)\setminus \{0\} \subset [q_1,1]$ for some $q_1>0$.
Let us choose $q_1$ maximally, so that $q_1\in\Supp(\zeta)\cup\{1\}$.
Now consider any $s\in\SSS$.
If $h_s=0$, then $\Phi^s(0) = 0 = \phi^s(0)$, as desired.
If instead $h_s^2>0$, then the following claim tells us that $0\notin\Supp(\zeta)$, and so it is not even necessary to check \eqref{parisi_identity_2} at $q=0$.

\begin{claim} \label{jump_claim}
Assuming $\Supp(\zeta)\setminus \{0\} \subset [q_1,1]$ and $h_s^2 > 0$, we must have $\zeta(\{0\})=0$.
\end{claim}

\begin{proofclaim}
First observe that $\Phi^s(q_1) > 0$.
Indeed, if $q_1\in\Supp(\zeta)$, then we already know $\Phi^s(q_1)=\phi^s(q_1)$, and $\phi^s(q_1)>0$ since $h_s^2>0$.
If instead $q_1=1$, then we trivially have $\Phi^s(q_1)>0$ because $\Phi^s(1)=1$.

Suppose toward a contradiction that $\zeta(\{0\}) > 0$.
We will argue that $(\zeta,\Phi)$ cannot satisfy \eqref{inf_b_inf}.
Thanks to \cite[Thm.~1.5]{bates-sohn22}, we can modify $\Phi$ off the support of $\zeta$ without changing the value of the left-hand side of \eqref{inf_b_inf}.
So let us fix some $p\in(0,\lambda^s\Phi^s(q_1)\wedge q_1)$, and then assume that
\eq{
\Phi^s(q) = \begin{cases}
\displaystyle q/\lambda^s &\text{if $q\in[0,p]$}, \\
\displaystyle\frac{q_1-q}{q_1-p}(p/\lambda^s)+\frac{q-p}{q_1-p}\Phi^s(q_1)
&\text{if $q\in(p,q_1]$}.
\end{cases}
}
Correspondingly, for $t\neq s$, we assume that
\eq{
\Phi^t(q) = \begin{cases}
 0 &\text{if $q\in[0,p]$}, \\
\displaystyle\frac{q-p}{q_1-p}\Phi^t(q_1)
&\text{if $q\in(p,q_1]$}.
\end{cases}
}
It is easy to check that these assumptions preserve $\vc\lambda$-admissibility.
Now consider the following perturbed measure for sufficiently small $\eps>0$:
\eq{
\wt\zeta_\eps \coloneqq \zeta - \eps\delta_0 + \eps\delta_{p}.
}
If we define
\eq{
d^t_\eps(q) \coloneqq \int_q^1\wt\zeta_\eps\big([0,u]\big)(\xi^t\circ\Phi)'(u)\ \dd u,
}
then by construction we have
\eq{
d^t_\eps(q) - d^t(q) = -\eps\one_{\{q<p\}}\int_q^{ p}(\xi^t\circ\Phi)'(u)\ \dd u
=-\eps\one_{\{q<p\}}\big(\xi^t(\Phi(p))-\xi^t(\Phi(q))\big).
}
This gives the derivative calculation
\eeq{ \label{meas_pert_derivative}
\frac{\dd}{\dd\eps} d^t_\eps(q)
= -\one_{\{q<p\}}\big(\xi^t(\Phi(p))-\xi^t(\Phi(q))\big).
}
Very similarly, we have
\eeq{ \label{meas_pert_theta}
\frac{\dd}{\dd\eps}\int_0^1\wt\zeta_\eps\big([0,q]\big)(\theta\circ\Phi)'(q)\ \dd q
= -\int_0^{p}(\theta\circ\Phi)'(q)\ \dd q
= -\theta(\Phi(p)). 
}
Now choose $\vc b$ such that 
\eq{
A(\zeta,\Phi,\vc b) = \inf_{\wt{\vc b}} A(\zeta,\Phi,\wt{\vc b}).
}
Referring to \eqref{A_def}, \eqref{meas_pert_derivative}, and \eqref{meas_pert_theta}, we have
\eq{
\frac{\dd A(\wt\zeta_\eps,\Phi,\vc b)}{\dd\eps}
=\sum_{t\in\SSS} \frac{\lambda^t}{2}\bigg[&-\frac{h_t^2+\red{\xi^t(\vc0)}}{(b^t-d^t(0))^2}\big(\xi^t(\Phi(p))-\xi^t(\vc 0)\big) \\
&- \int_0^{p}\frac{(\xi^t\circ \Phi)'(q)}{(b^t-d^t(q))^2}\big(\xi^t(\Phi(p))-\xi^t(\Phi(q))\big)\ \dd q\bigg]+ \frac{\theta(\Phi(p))}{2}.
}
Ignoring all species but $s$ and recalling the definition of $\theta$ from \eqref{gamma_theta_def}, we see from this identity that
\eq{
\frac{2}{\lambda^s}\frac{\dd A(\wt\zeta_\eps,\Phi,\vc b)}{\dd\eps}
\leq -\frac{h_s^2+\red{\xi^s(\vc0)}}{(b^s-d^s(0))^2}\big(\xi^s(\Phi(p))-\xi^s(\vc 0)\big)
+ \frac{p\xi^s(\Phi(p)) - \xi(\Phi(p))}{\lambda^s}. 
}
Since $\Phi^t(p) = (p/\lambda^s)\one_{\{t=s\}}$, the last term satisfies
\eq{
-\frac{\xi(\Phi(p))}{\lambda^s} = -\int_0^{p/\lambda^s}\xi^s\big|_{\{q^s=q,\, q^t=0\text{ for all $t\neq s$}\}}\ \dd q \leq -\frac{p\xi^s(\vc 0)}{\lambda^s}.
}
In light of this inequality, the previous estimate becomes
\eeq{ \label{A_deriv_meas}
\frac{2}{\lambda^s}\frac{\dd A(\wt\zeta_\eps,\Phi,\vc b)}{\dd\eps}
\leq\Big(-\frac{h_s^2+\red{\xi^s(\vc0)}}{(b^s-d^s(0))^2} + \frac{p}{\lambda^s}\Big)\big(\xi^s(\Phi(p))-\xi^s(\vc 0)\big).
}
Once again, since $\Phi^t(p) = (p/\lambda^s)\one_{\{t=s\}}$, we have
\eq{
\xi^s(\Phi(p))-\xi^s(\vc 0)
= \int_0^{p/\lambda^s}\frac{\partial \xi^s}{\partial q^s}\Big|_{\{q^s=q,\, q^t=0\text{ for all $t\neq s$}\}}\ \dd q \stackref{xi_strictly_convex}{>} 0.
}
Consequently, the derivative in \eqref{A_deriv_meas} is strictly negative whenever $p$ satisfies $p/\lambda^s < (h_s^2+\red{\xi^s(\vc0)})/((b^s-d^s(0))^2)$.
In particular, $A(\zeta,\Phi,\vc b)$ cannot be equal to $\inf A$.
\end{proofclaim}

\begin{proof}[Proof of Claim \ref{parisi_derivative_claim}]
First we make two preliminary calculations that will be used several times.
First, for any $C^1$ function $f:[0,1]^\SSS\to\R$,
by the chain rule together with the definition \eqref{wt_Phi_def} of $\wt\Phi_\eps$, we have
\eeq{ \label{general_eps_func}
\Big(\frac{\dd}{\dd\eps}(f\circ\wt\Phi_\eps)(q)\Big)\Big|_{\eps=0^+}
&= \sum_{t\in\SSS}\Big(\frac{\partial f}{\partial q^t}(\wt\Phi_\eps(q))\frac{\dd\wt\Phi_\eps^t(q)}{\dd\eps}\Big)\Big|_{\eps=0^+} \\
&= \sum_{t\in\SSS}\Big((\partial^t f\circ\Phi)\cdot(\phi^t-\Phi^t)\Big)(q).
}
Second, if $f$ is also $C^2$, then we include add an application of the product rule to obtain 
\eeq{ \label{general_eps_deriv}
&\Big(\frac{\dd}{\dd\eps}(f\circ\wt\Phi_\eps)'(q)\Big)\Big|_{\eps=0^+}
= \Big(\frac{\dd}{\dd\eps}\sum_{r\in\SSS} \frac{\partial f}{\partial q^r}(\wt\Phi_\eps(q))\cdot(\wt\Phi^r_\eps)'(q)\Big)\Big|_{\eps=0^+} \\
&=\sum_{r,t\in\SSS}\frac{\partial^2 f}{\partial q^t\partial q^r}(\wt\Phi_\eps(q))\frac{\dd\wt\Phi_\eps^t(q)}{\dd\eps}(\wt\Phi^r_\eps)'(q)\Big|_{\eps=0^+}
+ \sum_{r\in\SSS}\frac{\partial f}{\partial q^r}(\wt\Phi_\eps(q))\frac{\dd(\wt\Phi^r_\eps)'(q)}{\dd\eps}\Big|_{\eps=0^+} \\
&= \sum_{r,t\in\SSS}\frac{\partial^2 f}{\partial q^t\partial q^r}(\Phi(q))\big(\phi^t(q)-\Phi^t(q)\big)(\Phi^r)'(q)
+ \sum_{r\in\SSS}\frac{\partial f}{\partial q^r}(\Phi(q))\big((\phi^r)'(q)-(\Phi^r)'(q)\big) \\
&=\sum_{t,r\in\SSS}\frac{\partial^2 f}{\partial q^r\partial q^t}(\Phi(q))\big(\phi^t(q)-\Phi^t(q)\big)(\Phi^r)'(q)
+ \sum_{t\in\SSS}\frac{\partial f}{\partial q^t}(\Phi(q))\big((\phi^t)'(q)-(\Phi^t)'(q)\big) \\
&= \sum_{t\in\SSS}\Big((\partial^t f\circ\Phi)\cdot(\phi^t-\Phi^t)\Big)'(q).
}
Let us note once and for all that in subsequent calculations, whenever we differentiate an integral, the derivative of the integrand will be uniformly bounded.
Therefore, there is never any issue exchanging differentiation and integration.

In order to write down an expression for $A(\zeta_\eps,\Phi_\eps,\vc b)$, we define the following function:
\eq{
\wt d^s_\eps(q) \coloneqq \int_q^1\zeta\big([0,u]\big)(\xi^s\circ\wt\Phi_\eps)'(u)\ \dd u, \quad q\in[0,1].
}
Note that from \eqref{general_eps_deriv} we have
\eeq{ \label{d_eps_deriv}
\Big(\frac{\dd}{\dd\eps}\wt d^s_\eps(q)\Big)\Big|_{\eps=0+}
= \int_q^1\zeta\big([0,u]\big)\sum_{t\in\SSS}\Big((\partial^t\xi^s\circ\Phi)\cdot(\phi^t-\Phi^t)\Big)'(u)\ \dd u.
}
For convenience, we will write
\eeq{ \label{Xs_def}
X^s(q) \coloneqq \sum_{t\in\SSS}\Big((\partial^t\xi^s\circ\Phi)\cdot(\phi^t-\Phi^t)\Big)(q), \quad q\in[0,1].
}
Now, by a change of variables $v = \alpha_\eps^{-1}(u)$, we obtain the following identity for all $q\in[\alpha_\eps(0),\alpha_\eps(1)]$:
\eeq{ \label{ds_to_tilde}
d^s_\eps(q) &\coloneqq \int_q^1\zeta_\eps\big([0,u]\big)(\xi^s\circ\Phi_\eps)'(u)\ \dd u \\
&= \int_{\alpha_\eps^{-1}(q)}^{1}\zeta\big([0,u]\big)(\xi^s\circ\wt\Phi_\eps)'(v)\ \dd v
+ \int_{\alpha_\eps(1)}^1(\xi^s\circ\Phi_\eps)'(u)\ \dd u \\
&= \wt d^s_\eps(\alpha_\eps^{-1}(q))
+ \xi^s(\vc 1)-(\xi^s\circ\wt\Phi_\eps)(1).
}
If $q\in[0,\alpha_\eps(0)]$, then from the definition \eqref{new_zeta} of $\zeta_\eps$, we trivially have
\eeq{ \label{d_eps_constant}
d^s_\eps(q) = d^s_\eps(\alpha_\eps(0)) = \wt d^s_\eps(0)+\xi^s(\vc 1) - (\xi^s\circ\wt\Phi_\eps)(1).
}
Let us assume henceforth that $\eps$ is small enough that when $q=0$, this last expression is sufficiently close to $d^s(0)$ so as to be less than $b^s$.
For the remaining values of $q\in[\alpha_\eps(1),1]$, we simply have
\eeq{ \label{ds_late}
d^s_\eps(q) = \int_q^1(\xi^s\circ\Phi_\eps)'(u)\ \dd u
= \xi^s(\vc 1) - (\xi^s \circ \Phi_\eps)(q).
}
Putting together these observations, we obtain the desired expression:
\eeq{ \label{expressing_A}
A(\zeta_\eps,\Phi_\eps,\vc b)
= &\sum_{s\in\SSS}\frac{\lambda^s}{2}\bigg[\frac{h_s^2+\red{\xi^s(\vc0)}}{b^s-d^s_\eps(0)}+\int_0^{\alpha_\eps(0)}\frac{(\xi^s\circ\Phi_\eps)'(q)}{b^s-d^s_\eps(q)}\ \dd q+\int_{\alpha_\eps(0)}^{\alpha_\eps(1)}\frac{(\xi^s\circ\Phi_\eps)'(q)}{b^s-d^s_\eps(q)}\ \dd q\\
&+\int_{\alpha_\eps(1)}^1\frac{(\xi^s\circ\Phi_\eps)'(q)}{b^s-d^s_\eps(q)}\ \dd q+b^s-1-\log b^s\bigg] -\frac{1}{2}\int_0^1 \zeta_\eps\big([0,u]\big)(\theta\circ\Phi_\eps)'(u)\ \dd u.
}
Now we differentiate each term on the right-hand side.
First, by applying \eqref{general_eps_func} and \eqref{d_eps_deriv} to the right-hand side \eqref{d_eps_constant}, we have
\eeq{ \label{zero_outcome}
\frac{\dd}{\dd\eps}\Big(\frac{h_s^2+\red{\xi^s(\vc0)}}{b^s-d_\eps^s(0)}\Big)\Big|_{\eps=0^+}
= \frac{h_s^2+\red{\xi^s(\vc0)}}{(b^s-d^s(0))^2}\Big(\int_0^1\zeta\big([0,u]\big)(X^s)'(u)\ \dd u-X^s(1)\Big).
}
Next we consider the first integral on the right-hand side of \eqref{expressing_A}.
Recalling the constant from \eqref{d_eps_constant}, we trivially obtain
\eq{
\int_0^{\alpha_\eps(0)}\frac{(\xi^s\circ\Phi_\eps)'(q)}{b^s-d^s_\eps(q)}\ \dd q
= \frac{(\xi^s\circ\wt\Phi_\eps)(0)-\xi^s(0)}{b^s-\wt d_\eps^s(0)-\xi^s(\vc 1)+(\xi^s\circ\wt\Phi_\eps)(1)}.
}
We now differentiate the right-hand side with respect to $\eps$, and then evaluate at $\eps=0^+$.
In light of \eqref{general_eps_func}--\eqref{d_eps_deriv}, and the fact that $(\wt d^s_\eps,\wt\Phi_\eps)\to(d^s,\Phi)$ as $\eps\to0$, the result of this calculation is
\eeq{ \label{next_outcome}
&\Big(\frac{\dd}{\dd\eps}\int_0^{\alpha_\eps(0)}\frac{(\xi^s\circ\Phi_\eps)'(q)}{b^s-d^s_\eps(q)}\ \dd q\Big)\Big|_{\eps=0^+} 
= \frac{X^s(0)}{b^s-d^s(0)}.
}
Now consider the second integral on the right-hand side of \eqref{expressing_A}.
By using the same change of variables $\alpha_\eps^{-1}(q)\mapsto q$ as before and recalling \eqref{ds_to_tilde}, we deduce that
\eq{ 
\int_{\alpha_\eps(0)}^{\alpha_\eps(1)}\frac{(\xi^s\circ\Phi_\eps)'(q)}{b^s-d^s_\eps(q)}\ \dd q
&= \int_0^1\frac{(\xi^s\circ\wt\Phi_\eps)'(q)}{b^s-\wt d^s_\eps(q)-\xi^s(\vc 1)+(\xi^s\circ\wt\Phi_\eps)(1)}\ \dd q.
}
We now differentiate the right-hand side with respect to $\eps$, using the formulas \eqref{general_eps_func}, \eqref{general_eps_deriv}, and  \eqref{d_eps_deriv} to evaluate at $\eps=0^+$.
This results in
\eeq{ \label{part_1_after_derive}
\int_0^1\frac{(X^s)'(q)}{b^s- d^s(q)}\ \dd q+\int_0^1\frac{(\xi^s\circ\Phi)'(q)}{(b^s-d^s(q))^2}\Big(\int_q^1\zeta\big([0,u]\big)(X^s)'(u)\ \dd u-X^s(1)\Big)\, \dd q. 
}
The first integral appearing here can be rewritten using integration by parts:
\eq{
\int_0^1\frac{(X^s)'(q)}{b^s- d^s(q)}\ \dd q
&=\frac{X^s(1)}{b^s- d^s(1)}-\frac{X^s(0)}{b^s-d^s(0)} +\int_0^1\frac{X^s(q)}{(b^s- d^s(q))^2}\zeta\big([0,q])(\xi^s\circ\Phi)'(q)\ \dd q \\
&=\frac{X^s(1)}{b^s}-\frac{X^s(0)}{b^s-d^s(0)} +\int_0^1(\phi^s)'(q)X^s(q)\zeta\big([0,q]\big)\ \dd q.
}
Meanwhile, the second integral appearing in \eqref{part_1_after_derive} can be rewritten using \eqref{little_phi_def} as
\eq{
&\int_0^1\frac{(\xi^s\circ\Phi)'(q)}{(b^s-d^s(q))^2}\int_q^1\zeta\big([0,u]\big)(X^s)'(u)\ \dd u\, \dd q - \Big(\phi^s(1)-\frac{h_s^2+\red{\xi^s(\vc0)}}{(b^s-d^s(0))^2}\Big)X^s(1) \\
&=\int_0^1\Big(\phi^s(u)-\frac{h_s^2+\red{\xi^s(\vc0)}}{(b^s-d^s(0))^2}\Big)\zeta\big([0,u]\big)(X^s)'(u)\  \dd u - \Big(\phi^s(1)-\frac{h_s^2+\red{\xi^s(\vc0)}}{(b^s-d^s(0))^2}\Big)X^s(1) \\
&= -\frac{h_s^2+\red{\xi^s(\vc0)}}{(b^s-d^s(0))^2}\int_0^1\zeta\big([0,u]\big)(X^s)'(u)\ \dd u -\int_0^1(\phi^s)'(u)\zeta\big([0,u]\big)X^s(u)\ \dd u \\
&\phantom{=}\, -\int\phi^s(u)X^s(u)\ \zeta(\dd u) +  \frac{h_s^2+\red{\xi^s(\vc0)}}{(b^s-d^s(0))^2}X^s(1).
}
The outcome of the four previous displays, together with \eqref{zero_outcome} and  \eqref{next_outcome}, is
\eeq{ \label{first_outcome_1}
\frac{\dd}{\dd\eps}\Big(\frac{h_s^2+\red{\xi^s(\vc0)}}{b^s-d^s_\eps(0)}+\int_0^{\alpha_\eps(0)}\frac{(\xi^s\circ\Phi_\eps)'(q)}{b^s-d^s_\eps(q)}\ \dd q+\int_{\alpha_\eps(0)}^{\alpha_\eps(1)}\frac{(\xi^s\circ\Phi_\eps)'(q)}{b^s-d^s_\eps(q)}\ \dd q\Big)\Big|_{\eps=0^+}& \\
= \frac{X^s(1)}{b^s} - \int\phi^s(u)X^s(u)\ \zeta(\dd u&).
}
Next, because of \eqref{ds_late}, the third integral in \eqref{expressing_A} is
\eq{
\int_{\alpha_\eps(1)}^1\frac{(\xi^s\circ\Phi_\eps)'(q)}{b^s-d^s_\eps(q)}\ \dd q
= \log b^s - \log\big(b^s-\xi^s(\vc 1)+(\xi^s\circ\wt\Phi_\eps)(1)\big).
}
Performing the relevant differentiation on the right-hand side, and applying \eqref{general_eps_func}, we obtain
\eeq{ \label{first_outcome_2}
\Big(\frac{\dd}{\dd\eps}\int_{\alpha_\eps(1)}^1\frac{(\xi^s\circ\Phi_\eps)'(q)}{b^s-d^s_\eps(q)}\ \dd q\Big)\Big|_{\eps=0^+}
= \frac{-X^s(1)}{b^s}.
}
Finally, by the same change of variables $\alpha_\eps^{-1}(q)\mapsto q$ as before, the fourth integral in \eqref{expressing_A} is
\eq{
\int_0^1\zeta_\eps\big([0,q]\big)(\theta\circ\Phi_\eps)'(q)\ \dd q
= \int_0^{1}\zeta\big([0,q]\big)(\theta\circ\wt\Phi_\eps)'(q)\ \dd q
+ \theta(\vc 1)-(\theta\circ\wt\Phi_\eps)(1).
}
Thanks to \eqref{general_eps_deriv}, performing the relevant differentiation results in
\eq{
\Big(\frac{\dd}{\dd\eps}\int_0^1\zeta_\eps\big([0,q]\big)(\theta\circ\Phi_\eps)'(q)\ \dd q\Big)\Big|_{\eps=0^+}
&=\int_0^1\zeta\big([0,q]\big)Y'(q)\ \dd q - Y(1), }
where
\eq{
Y(q) &\stackrefp{gamma_theta_def}{\coloneqq} \sum_{t\in\SSS}\Big((\partial^t\theta\circ\Phi)\cdot(\phi^t-\Phi^t)\Big)(q) \\
&\stackref{gamma_theta_def}{=} \sum_{t,s\in\SSS}\lambda^s\Phi^s(q)(\partial^t\xi^s\circ\Phi)(q)\cdot(\phi^t-\Phi^t)(q) \stackref{Xs_def}{=}\sum_{s\in\SSS}\lambda^s\Phi^s(q)X^s(q).
}
Then integration by parts gives
\eq{ 
\int_0^1\zeta\big([0,q]\big)Y'(q)\ \dd q - Y(1)
= -\int_0^1 Y(q)\ \zeta(\dd q).
}
The outcome of the three previous displays is
\eeq{ \label{second_outcome}
\Big(\frac{\dd}{\dd\eps}\int_0^1\zeta_\eps\big([0,q]\big)(\theta\circ\Phi_\eps)'(q)\ \dd q\Big)\Big|_{\eps=0^+}
= -\int_0^1\sum_{s\in\SSS}\lambda^s\Phi^s(q)X^s(q)\ \zeta(\dd q).
}
Putting together \eqref{first_outcome_1}, \eqref{first_outcome_2}, and \eqref{second_outcome} results in
\eq{
\frac{\dd A(\zeta_\eps,\Phi_\eps,\vc b)}{\dd\eps}\Big|_{\eps=0^+}
&\stackrefp{Xs_def}{=} - \frac{1}{2}\int\sum_{s\in\SSS}\lambda^s(\phi^s(q)-\Phi^s(q))X^s(q)\ \zeta(\dd q) \\
&\stackref{Xs_def}{=} -\frac{1}{2}\int\sum_{s,t\in\SSS}(\phi^s(q)-\Phi^s(q))\cdot(\partial^t\partial^s\xi\circ\Phi)(q)\cdot(\phi^t(q)-\Phi^t(q))\ \zeta(\dd q).
}
Of course, the final line is simply the right-hand side of \eqref{parisi_perturb_derivative}.
\renewcommand{\qedsymbol}{$\square$ (Claim and Theorem)}
\end{proof}
\renewcommand{\qedsymbol}{$\square$}

\subsection{Identity satisfied by Crisanti--Sommers minimizers} \label{cs_min_proof}
The identity \eqref{cs_identity} is found by making two types of perturbations to $(\zeta,\Phi)$, which we call ``up-perturbations'' and ``down-perturbations''.
We will initially consider just up-perturbations, as down-perturbations will have a very similar treatment.

Fix $(\zeta,\Phi)$ which satisfies \eqref{gap_assumption} for some $q_*<1$.
Fix $s\in\SSS$, and consider a point $a\in(0,1)$ such that
\begin{subequations}
\label{up_assumptions}
\eeq{ \label{not_reach_1}
\Phi^t(a) < 1 \quad \text{for all $t\in\SSS$}.
}
Fix any $\delta\geq 0$ small enough that
\eeq{ \label{delta_small_assumption}
a-\delta>0.
}
We also assume that $a^+$ is a point of increase for $\Phi^s$:
\eeq{ \label{at_a_assumption}
\Phi^s(a) < \Phi^s(q)  \quad \text{for all $q\in(a,1]$}.
}
\end{subequations}
For sufficiently small $\eps>0$ (this parameter we will ultimately send to $0$), we can define
\eeq{ \label{a_hat_def}
\wh a_\eps \coloneqq \inf\{q\geq a:\, \Phi^s(q) = \Phi^s(a) + \eps + (q-a)\eps^2\}.
}
Then consider the new function
\eeq{ \label{up_perturbation}
 \wh\Phi_\eps^s(q) &\coloneqq \begin{cases}
\Phi^s(q) + (q-a+\delta+\eps) &\text{if }q\in(a-\delta-\eps,a-\delta], \\
\Phi^s(q) + \eps &\text{if } q\in(a-\delta,a], \\
\Phi^s(a) + \eps + (q-a)\eps^2 &\text{if } q\in(a,\wh a_\eps], \\
\Phi^s(q) &\text{otherwise}.
\end{cases}
}
See Figure \ref{up_fig} for an illustration.
Note that $ \wh\Phi^s$ retains the continuity and monotonicity of $\Phi^s$.
For any $t\in\SSS\setminus\{s\}$, we simply take $ \wh\Phi^t = \Phi^t$.

\begin{figure}
\subfloat[Up-perturbation at $a$, for species $s$]{
\label{up_fig}
\centering
\includegraphics[clip, trim = 0 0 1.3in 0, width=0.8\textwidth]{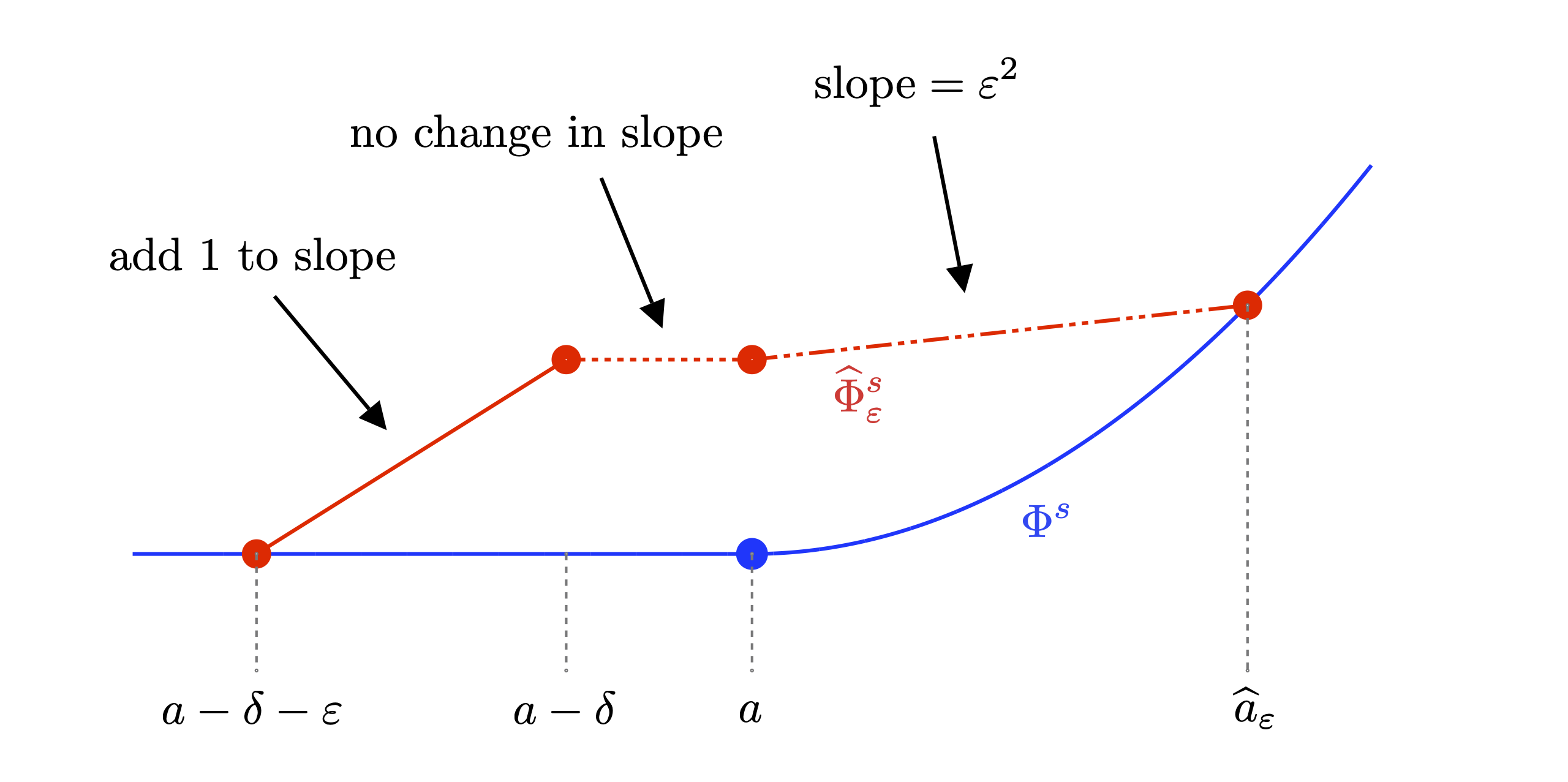}
} \\
\subfloat[Down-perturbation at $a$, for species $s$]{
\label{down_fig}
\centering
\includegraphics[clip, trim = 1.3in 0 0 0, width=0.8\textwidth]{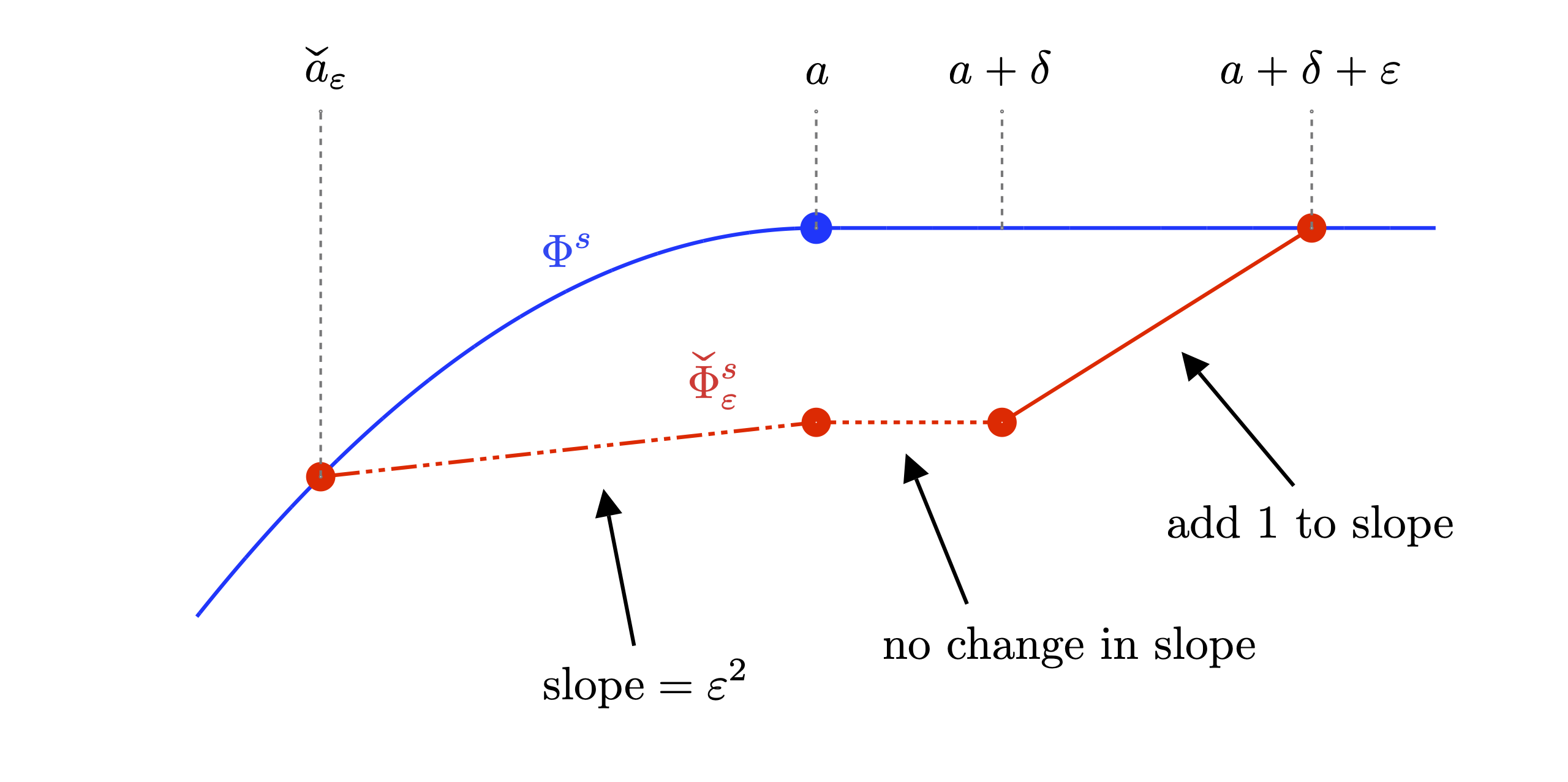}
}
\caption{Perturbations considered in the proof of Theorem \ref{cs_min_thm}. In (a), $\Phi^s$ must be strictly increasing to the right of $a$.
In (b), $\Phi^s$ must be strictly increasing to the left of $a$.
This is to ensure that $\wh a_\eps$ or $\wc a_\eps$ tends to $a$ as $\eps\to0$.}
\end{figure}

It is no longer the case that $\wh\Phi = (\wh\Phi^t)_{t\in\SSS}$ is $\vc\lambda$-admissible.
Therefore, we must perform a reparameterization as follows.
For $q\in[0,1]$, define
\eeq{ \label{alpha_up_def}
\alpha_\eps(q) \coloneqq \sum_{t\in\SSS}\lambda^t\wh\Phi_\eps^t(q)
&\stackref{admissible_def}{=}q + \lambda^s(\wh\Phi^s_\eps(q)-\Phi^s(q)).
}
From \eqref{up_perturbation} we have
\eeq{ \label{function_change}
\wh\Phi^s_\eps(q)-\Phi^s(q)
&= \one_{\{q\in(a-\delta-\eps,a-\delta]\}}(q-a+\delta+\eps)
+ \eps\one_{\{q\in(a-\delta,a]\}} \\
&\phantom{=}+ \one_{\{q\in(a,\wh a_\eps]\}}(\Phi^s(a)+\eps + (q-a)\eps^2-\Phi^s(q)).
}
In particular, for any $q\notin\{a-\delta-\eps,a-\delta,a,\wh a_\eps\}$ at which $\Phi^s$ is differentiable, we have
\eeq{ \label{derivative_change}
( \wh\Phi^s_\eps)'(q) - (\Phi^s)'(q)
= \one_{\{q\in(a-\delta-\eps,a-\delta)\}}
+ \eps^2\one_{\{q\in(a,\wh a_\eps)\}}
- \one_{\{q\in(a,\wh a_\eps)\}}(\Phi^s)'(q).
}
Notice that \eqref{admissible_def} forces $(\Phi^s)'(q) \leq 1/\lambda^s$, and so
\eq{
( \wh\Phi^s_\eps)'(q) - (\Phi^s)'(q)
> -1/\lambda^s.
}
Returning to \eqref{alpha_up_def}, we now see that
$q\mapsto \alpha_\eps(q)$ is strictly increasing.
Since this map is also continuous, the inverse $\alpha_\eps^{-1}$ is well-defined, strictly increasing, and continuous.
Clearly $\alpha_\eps(0)=0$, and so in order for the domain of $\alpha_\eps^{-1}$ to be all of $[0,1]$, we just need that $\alpha_\eps(1) = 1$.
Indeed, this follows from the simple observation that $\alpha_\eps$ disagrees with the identity function only on the interval $(a-\delta-\eps,a+\wh a_\eps)$.
The following lemma---which also demonstrates the purpose of assuming \eqref{at_a_assumption}---then suffices.

\begin{lemma} \label{a_hat_lemma}
As $\eps\searrow0$, we have $\wh a_\eps\searrow a$.
\end{lemma}

\begin{proof}
By \eqref{at_a_assumption}, for any $\eta>0$, there exists $\eps>0$ small enough that $\Phi^s(a+\eta) > \Phi^s(a)+2\eps$.
Whenever $\eps\leq1$, we trivially have $\eps + \eps^2(q-a) \leq 2\eps$.
Hence $\Phi^s(a+\eta)>\Phi^s(a)+\eps+\eps^2\eta$, which means $\wh a_\eps < a+\eta$ by definition \eqref{a_hat_def}.
\end{proof}

Now that we know $\alpha_\eps^{-1}\colon [0,1]\to[0,1]$ is an increasing bijection, we can define the \textit{up-perturbation} $(\zeta_\eps,\Phi_\eps)$ by
\eeq{ \label{up_perturbation_rescaled}
\zeta_\eps\big([0,q]\big) \coloneqq \zeta\big([0,\alpha_\eps^{-1}(q)]\big), \qquad
\Phi_\eps(q) \coloneqq \wh\Phi_\eps(\alpha_\eps^{-1}(q)).
}
By definition \eqref{alpha_up_def} of $\alpha_\eps$, the map $\Phi_\eps \colon[0,1]\to[0,1]^\SSS$ is $\vc\lambda$-admissible.
The principal calculation needed to prove Theorem \ref{cs_min_thm} is the following.

\begin{prop}[Up-perturbation] \label{derivative_prop}
Assuming \eqref{up_assumptions}, let $(\zeta_\eps,\Phi_\eps)$ be as in \eqref{up_perturbation_rescaled}.
We then have the following right derivative:
\eeq{ \label{B_deriv}
\frac{\dd B(\zeta_\eps,\Phi_\eps)}{\dd\eps}\Big|_{\eps=0^+}
=\frac{\lambda^s}{2}\bigg[&\Big(\zeta\big([0,a-\delta)\big)-\zeta\big([0,a]\big)\Big)\Big(\red{h_s^2}-\int_0^{a-\delta}\frac{(\Phi^s)'(q)}{(\Delta^2(q))^2}\ \dd q\Big)\\
&-\int_{a-\delta}^a\Big(\zeta\big([0,q]\big)-\zeta\big([0,a]\big)\Big)\frac{(\Phi^s)'(q)}{(\Delta^2(q))^2}\ \dd q \\
&+\int_{a-\delta}^a \zeta\big([0,q]\big)(\xi^s\circ\Phi)'(q)\ \dd q \\
&+\zeta\big([0,a-\delta)\big)\xi^s(\Phi(a-\delta))
-\zeta\big([0,a]\big)\xi^s(\Phi(a))\bigg].
\raisetag{3.5\baselineskip}
}
\end{prop}

To streamline the proof of Proposition \ref{derivative_prop}, we make one calculation beforehand.

\begin{lemma} \label{dirac_delta_lemma}
Suppose that $(f_\eps)_{\eps\geq0}$ is a family of real-valued functions on $[0,1]$ such that $f_\eps\to f_0$ uniformly as $\eps\searrow0$.
If $f_0$ is right-continuous at $a$, then
\eeq{ \label{dirac_delta}
\lim_{\eps\searrow0}\frac{1}{\eps}\int_a^{\wh a_\eps} f_\eps(q)(\Phi^s)'(q)\ \dd q = f_0(a).
}
\end{lemma}

\begin{proof}
Fix any $\eta>0$.
By right-continuity of $f_0$ and Lemma \ref{a_hat_lemma}, for all $\eps$ sufficiently small we have
\eq{
f(a) - \eta \leq f(q) \leq f(a)+\eta \quad \text{for all $q\in[a,\wh a_\eps]$}.
}
By the hypothesis of uniform convergence, this can be upgraded to
\eq{
f(a) - 2\eta \leq f_\eps(q) \leq f(a)+2\eta \quad \text{for all $q\in[a,\wh a_\eps]$}.
}
Therefore,
\eq{
\int_a^{\wh a_\eps} f_\eps(q)(\Phi^s)'(q)\ \dd q
&\leq \int_a^{\wh a_\eps} (f(a)+2\eta)(\Phi^s)'(q)\ \dd q \\
&= (f(a)+2\eta)(\Phi^s(\wh a_\eps)-\Phi^s(a))
\stackref{a_hat_def}{=} (f(a)+2\eta)(\eps + \eps^2(\wh a_\eps-a)),
}
and similarly
\eq{
\int_a^{\wh a_\eps} f_\eps(q)(\Phi^s)'(q)\ \dd q
&\geq (f(a)-2\eta)(\eps + \eps^2(\wh a_\eps-a)).
}
Since $\eta$ is arbitrary, \eqref{dirac_delta} follows from the two previous displays.
\end{proof}

\begin{proof}[Proof of Proposition \ref{derivative_prop}]
By \eqref{not_reach_1}, we can choose $q_*\in(a,1)$  such that \eqref{gap_assumption} is satisfied.
Then, thanks to Lemma \ref{a_hat_lemma}, we can assume throughout the proof that $\wh a_\eps < q_*$.
Define
\eq{
\Delta_\eps^t(q) \coloneqq \int_q^1\zeta_\eps\big([0,u]\big)(\Phi_\eps^t)'(u)\ \dd u, \quad q\in[0,1], t\in\SSS,
}
so that \eqref{B_def} reads as
\eeq{ \label{B_eps}
B(\zeta_\eps,\Phi_\eps)= \sum_{t\in\SSS}\frac{\lambda^s}{2}\Big[\red{h_t^2}\Delta_\eps^t(0)+\int_0^{q_*}\frac{(\Phi^t_\eps)'(q)}{\Delta^t_\eps(q)}\ \dd q+\log \Delta^t_\eps(q_*)&\Big]\\+\frac{1}{2}\int_0^1\zeta_\eps\big([0,q]\big)(\xi\circ \Phi_\eps)'(q)\ \dd q&.
}
Observe that by the definition \eqref{up_perturbation_rescaled}, we can execute a change of variables:
\eeq{ \label{first_cov}
\Delta^t_\eps(q) &= \int_q^1 \zeta\big([0,\alpha_\eps^{-1}(u)]\big)\frac{(\wh\Phi^t_\eps)'(\alpha_\eps^{-1}(u))}{\alpha_\eps'(\alpha_\eps^{-1}(u))}\ \dd u 
= \int_{\alpha_\eps^{-1}(q)}^1 \zeta\big([0,u]\big)(\wh\Phi_\eps^t)'(u)\ \dd u.
}
In particular, when $t\neq s$ we have $\wh\Phi^t_\eps = \Phi^t$, and so
\eeq{ \label{delta_unchanged}
\Delta^t_\eps(q) = \Delta^t(\alpha_\eps^{-1}(q)) \quad \text{for $t\neq s$}.
}
When $t=s$, we will instead interpret \eqref{first_cov} as
\eeq{ \label{delta_changed}
\Delta^s_\eps(q) = \wh \Delta_\eps^s(\alpha_\eps^{-1}(q)), \quad \text{where} \quad
\wh\Delta_\eps^s(q) &\coloneqq \int_q^1 \zeta\big([0,u]\big)(\wh\Phi^s_\eps)'(u)\ \dd u,
}
which by \eqref{derivative_change} is equal to
\eeq{ \label{rewriting_hat_delta}
\wh\Delta_\eps^s(q)
&=\int_q^1\zeta\big([0,u]\big)\Big[(\Phi^s)'(u)+\one_{\{u\in(a-\delta-\eps,a-\delta)\}}
+ \eps^2\one_{\{u\in(a,\wh a_\eps)\}}
- \one_{\{q\in(a,\wh a_\eps)\}}(\Phi^s)'(u)\Big]\ \dd u.
}
In any case, since $\alpha_\eps^{-1}(q_*) = q_*$ and $\wh\Phi_\eps(q) = \Phi(q)$ for all $q\in[q_*,1]$, \eqref{first_cov} shows that $\Delta^t_\eps(q_*) = \Delta^t(q_*)$ for all $t\in\SSS$.
That is, the logarithm appearing in \eqref{B_eps} does not depend on $\eps$.
To understand the integral appearing before the logarithm, we will need the following identity.

\begin{claim}  \label{delta_derivative_claim}
We have the following right derivative:
\eq{
\frac{\dd\wh\Delta^s_\eps(q)}{\dd\eps}\Big|_{\eps=0^+}
= \begin{cases}
\zeta\big([0,a-\delta)\big) - \zeta\big([0,a]\big) &\text{if $q\in[0,a-\delta)$}, \\
\phantom{\zeta\big([0,a-\delta)\big) - \zeta\big([0,a]\big)}\mathllap{-\zeta\big([0,a]\big)} &\text{if $q\in[a-\delta,a]$}, \\
\phantom{\zeta\big([0,a-\delta)\big) - \zeta\big([0,a]\big)}\mathllap{0} &\text{if $q\in(a,1]$}.
\end{cases}
}
\end{claim}

\begin{proofclaim}
Note that $\wh\Delta_0^s = \Delta^s$, and we can read off from \eqref{rewriting_hat_delta} that
\eeq{ \label{delta_difference}
\wh\Delta_\eps^s(q)-\Delta^s(q)
= \int_q^1\zeta\big([0,u]\big)\Big[\one_{\{u\in(a-\delta-\eps,a-\delta)\}}
+ \eps^2\one_{\{u\in(a,\wh a_\eps)\}}
- \one_{\{q\in(a,\wh a_\eps)\}}(\Phi^s)'(u)\Big]\ \dd u.
}
If $q<a-\delta$, then all indicator functions appearing in \eqref{delta_difference} are equal to $1$ for sufficiently small $\eps$.
By the left-continuity of the map $u\mapsto\zeta\big([0,u)\big)$, we have
\eq{
\lim_{\eps\searrow0}\frac{1}{\eps}\int_{a-\delta-\eps}^{a-\delta}\zeta\big([0,u]\big)\ \dd u
= \zeta\big([0,a-\delta)\big).
}
Of course, we trivially have
\eq{
\lim_{\eps\searrow0}\eps\int_{a}^{\wh a_\eps}\zeta\big([0,u]\big)\ \dd u = 0.
}
Finally, by Lemma \ref{dirac_delta_lemma} we have
\eq{
\lim_{\eps\searrow0}\frac{1}{\eps}\int_a^{\wh a_\eps}\zeta\big([0,u]\big)(\Phi^s)'(u)\ \dd u
= \zeta\big([0,a]\big).
}
These three limits together yield the desired result in the case of $q\in[0,a-\delta)$.
If instead $q\in[a-\delta,a]$, then only the final two indicator functions in \eqref{delta_difference} are nonzero, and so only the two previous displays apply.
Finally, if $q>a$, then for all $\eps$ sufficiently small, all indicators in \eqref{delta_difference} are equal to 0.
Here we have again used Lemma \ref{a_hat_lemma}.
\end{proofclaim}

We now return to analyzing $B(\zeta_\eps,\Phi_\eps)$.
Let us consider the first part of the right-hand side of \eqref{B_eps}.
For any $t\neq s$, nothing has changed:
\eq{
&\red{h_t^2}\Delta^t_\eps(0)+\int_0^{q_*}\frac{(\Phi^t_\eps)'(q)}{\Delta^t_\eps(q)}\ \dd q \\
&\stackref{delta_unchanged}{=}
\red{h_t^2}\Delta^t(0)+\int_0^{q_*}
\frac{1}{\Delta^t(\alpha_\eps^{-1}(q))}\cdot\frac{(\Phi^t_\eps)'(\alpha_\eps^{-1}(q)))}{\alpha_\eps'(\alpha_\eps^{-1}(q))}\ \dd q \\
&\stackrefp{delta_unchanged}{=}\red{h_t^2}\Delta^t(0)+\int_0^{q_*}\frac{(\Phi^t)'(q)}{\Delta^t(q)}\ \dd q.
}
If $t=s$, we instead have
\eq{
&\red{h_s^2}\Delta^s_\eps(0)+\int_0^{q_*}\frac{(\Phi^s_\eps)'(q)}{\Delta^s_\eps(q)}\ \dd q \\
&\stackref{delta_changed}{=}
\red{h_s^2}\wh\Delta^s_\eps(0)+\int_0^{q_*}\frac{1}{\wh\Delta^s_\eps(\alpha_\eps^{-1}(q))}\frac{(\wh\Phi^s_\eps)'(\alpha_\eps^{-1}(q)))}{\alpha_\eps'(\alpha_\eps^{-1}(q))}\ \dd q  \\
&\stackrefp{delta_changed}{=} \red{h_s^2}\wh\Delta^s_\eps(0)+\int_0^{q_*}\frac{(\wh\Phi^s_\eps)'(q)}{\wh\Delta^s_\eps(q)}\ \dd q.
}
Of course, the first term in the final line is subject to Claim \ref{delta_derivative_claim}.
For the second term, we make the following calculation.

\begin{claim} \label{B_deriv_1_claim}
We have the following right derivative:
\eeq{ \label{B_deriv_1}
\frac{\dd}{\dd\eps}\Big(\int_0^{q_*}\frac{(\wh\Phi^s_\eps)'(q)}{\wh\Delta^s_\eps(q)}\ \dd q\Big)\Big|_{\eps=0+}
= &-\Big(\zeta\big([0,a-\delta)\big)-\zeta\big([0,a]\big)\Big)\int_0^{a-\delta}\frac{(\Phi^s)'(q)}{(\Delta^2(q))^2}\ \dd q \\
&-\int_{a-\delta}^a\Big(\zeta\big([0,q]\big)-\zeta\big([0,a]\big)\Big)\frac{(\Phi^s)'(q)}{(\Delta^2(q))^2}\ \dd q.
}
\end{claim}

\begin{proofclaim}
Intuitively, we pass the derivative through the integral, and then apply the quotient rule.
To be perfectly rigorous, though, let us do this carefully.
We begin by writing
\eq{
\frac{1}{\eps}\Big(\frac{(\wh\Phi^s_\eps)'(q)}{\wh\Delta^s_\eps(q)}
- \frac{(\Phi^s)'(q)}{\Delta^s(q)}\Big)
= \underbrace{\frac{(\wh\Phi^s_\eps)'(q)-(\Phi^s)'(q)}{\eps\wh\Delta^s_\eps(q)}}_{X_\eps(q)}
- \underbrace{\frac{(\Phi^s)'(q)}{\Delta^s(q)\wh\Delta^s_\eps(q)}\Big(\frac{\wh\Delta^s_\eps(q)-\Delta^s(q)}{\eps}\Big)}_{Y_\eps(q)}.
}
Now we analyze $X_\eps$ and $Y_\eps$ separately.

We begin with $X_\eps$.
Observe that because of \eqref{derivative_change}, we have
\eq{
\int_0^{q_*}X_\eps(q)\ \dd q
= \frac{1}{\eps}\int_{a-\delta-\eps}^{a-\delta}\frac{1}{\wh\Delta^s_\eps(q)}\ \dd q
+ \eps\int_a^{\wh a_\eps}\frac{1}{\wh\Delta^s_\eps(q)}\ \dd q
- \int_a^{\wh a_\eps}\frac{(\Phi^s)'(q)}{\wh\Delta^s_\eps(q)}\ \dd q.
}
To control the first integral on the right-hand side, we observe that
\eeq{ \label{uniform_convergence}
\Big|\frac{1}{\wh\Delta^s_\eps(q)}-\frac{1}{\Delta^s(q)}\Big|
= \Big|\frac{\wh\Delta^s_\eps(q)-\Delta^s(q)}{\wh\Delta^s_\eps(q)\Delta^s(q)}\Big|
&\stackref{delta_difference}{\leq}
\frac{\eps+\eps^2+(\wh a_\eps-a)/\lambda^s}{\big(\Delta^s(q)-(\wh a_\eps-a)/\lambda^s\big)\Delta^s(q)} \\
&\stackrefp{delta_difference}{\leq}
\frac{\eps+\eps^2+(\wh a_\eps-a)/\lambda^s}{\big(\Delta^s(q_*)-(\wh a_\eps-a)/\lambda^s\big)\Delta^s(q_*)} = o(1).
}
Hence
\eq{
\frac{1}{\eps}\int_{a-\delta-\eps}^{a-\delta}\frac{1}{\wh\Delta^s_\eps(q)}\ \dd q
&= \frac{1}{\eps}\int_{a-\delta-\eps}^{a-\delta}\frac{1}{\Delta^s(q)}\ \dd q + o(1), \quad \text{and} \quad
\eps\int_a^{\wh a_\eps}\frac{1}{\wh\Delta^s_\eps(q)}\ \dd q
= O(\eps).
}
And by the continuity of $q\mapsto\Delta^s(q)$ we have
\eq{
\lim_{\eps\searrow0}\frac{1}{\eps}\int_{a-\delta-\eps}^{a-\delta}\frac{1}{\Delta^s(q)}\ \dd q
= \frac{1}{\Delta^s(a-\delta)}.
}
Finally, thanks to the uniform convergence shown in \eqref{uniform_convergence}, Lemma \ref{dirac_delta_lemma} gives
\eq{
\lim_{\eps\searrow0}\int_{a}^{\wh a_\eps}\frac{(\Phi^s)'(q)}{\wh\Delta^s_\eps(q)}\ \dd q
= \frac{1}{\Delta^s(q)}.
}
In summary, the three previous displays yield
\eeq{ \label{X_deriv}
\lim_{\eps\searrow0}\int_0^{q_*}X_\eps(q)\ \dd q
= \frac{1}{\Delta^s(a-\delta)}-\frac{1}{\Delta^s(a)}
&= -\int_{a-\delta}^a\Big(\frac{1}{\Delta^s(q)}\Big)'\ \dd q \\
&= -\int_{a-\delta}^a\zeta\big([0,q]\big)\frac{(\Phi^s)'(q)}{(\Delta^s(q))^2}\ \dd q.
}
Next we analyze $Y_\eps$.
We have
\eq{
&\bigg|Y_\eps(q) -\frac{(\Phi^s)'(q)}{(\Delta^s(q))^2}\frac{\dd\wh\Delta^s_\eps(q)}{\dd\eps}\Big|_{\eps=0}\bigg| \\
&\leq
\bigg|\frac{(\Phi^s)'(q)}{\Delta^s(q)\wh\Delta^s_\eps(q)}\Big(\frac{\wh\Delta^s_\eps(q)-\Delta^s(q)}{\eps}-\frac{\dd\wh\Delta^s_\eps(q)}{\dd\eps}\Big|_{\eps=0}\Big)\bigg|
+\bigg|\Big(\frac{(\Phi^s)'(q)}{\Delta^s(q)\wh\Delta^s_\eps(q)}-\frac{(\Phi^s)'(q)}{(\Delta^s(q))^2}\Big)\frac{\dd\wh\Delta^s_\eps(q)}{\dd\eps}\Big|_{\eps=0}\bigg|
 \\
&\leq 
\frac{1/\lambda^s}{\Delta^s(q_*)(\Delta^s(q_*)-(\wh a_\eps-a))}\bigg|\frac{\wh\Delta^s_\eps(q)-\Delta^s(q)}{\eps}-\frac{\dd\wh\Delta^s_\eps(q)}{\dd\eps}\Big|_{\eps=0}\bigg|
+\frac{1/\lambda^s}{\Delta^s(q_*)}\Big|\frac{1}{\wh\Delta^s_\eps(q)}-\frac{1}{\Delta^s(q)}\Big|.
}
This final line tends to $0$ as $\eps\searrow0$, and by \eqref{uniform_convergence} is bounded by a constant uniformly in $\eps$ and $q$.
Therefore, by dominated convergence and Claim \ref{delta_derivative_claim},
\eeq{ \label{Y_deriv}
\lim_{\eps\searrow0}\int_0^{q_*}Y_\eps(q)\ \dd q
= \Big(\zeta\big([0,a-\delta)\big) &- \zeta\big([0,a]\big)\Big)\int_0^{a-\delta}\frac{(\Phi^s)'(q)}{(\Delta^s(q))^2}\ \dd q \\
&- \zeta\big([0,a]\big)\int_{a-\delta}^a\frac{(\Phi^s)'(q)}{(\Delta^s(q))^2}\ \dd q.
}
Subtracting \eqref{Y_deriv} from \eqref{X_deriv} results in \eqref{B_deriv_1}.
\end{proofclaim}

We now turn our attention to the second integral in \eqref{B_eps}.

\begin{claim} \label{B_deriv_2_claim}
We have the following right derivative:
\eeq{ \label{B_deriv_2}
&\frac{\dd}{\dd\eps}\Big(\int_0^1\zeta_\eps\big([0,q]\big)(\xi\circ\Phi_\eps)'(q)\ \dd q\Big)\Big|_{\eps=0^+} \\
&=\lambda^s\bigg[\int_{a-\delta}^a \zeta\big([0,q]\big)(\xi^s\circ\Phi)'(q)\ \dd q
+\zeta\big([0,a-\delta)\big)\xi^s(\Phi(a-\delta))
-\zeta\big([0,a]\big)\xi^s(\Phi(a))\bigg].
}
\end{claim}

\begin{proofclaim}
By the same change of variables $u = \alpha_\eps^{-1}(q)$ we have used before, we can write
\eq{
\int_0^1\zeta_\eps\big([0,q]\big)(\xi\circ\Phi_\eps)'(q)\ \dd q
= \int_0^1 \zeta\big([0,u]\big)(\xi\circ\wh\Phi_\eps)'(u)\ \dd u.
}
Upon recalling the definition of $\xi^t$ from \eqref{gamma_theta_def}, we see that the chain rule gives
\eq{
(\xi\circ\wh\Phi_\eps)'(q)
= \sum_{t\in\SSS}\lambda^t\xi^t(\wh\Phi_\eps(q))(\wh\Phi^t_\eps)'(q).
}
Therefore, we wish to understand the quantity
\eeq{ \label{part_2}
&\sum_{t\in\SSS}\lambda^t\Big[\xi^t(\wh\Phi_\eps(q))(\wh\Phi^t_\eps)'(q)-\xi^t(\Phi(q))(\Phi^t)'(q)\Big] \\
&= \sum_{t\in\SSS}\lambda^t\Big[\big[\xi^t(\wh\Phi_\eps(q))-\xi^t(\Phi(q))\big](\wh\Phi^t_\eps)'(q) + \xi^t(\Phi(q))\big[(\wh\Phi^t_\eps)'(q)-(\Phi^t)'(q)\big] \\
&= \sum_{t\in\SSS}\lambda^t\big[\xi^t(\wh\Phi_\eps(q))-\xi^t(\Phi(q))\big](\wh\Phi^t_\eps)'(q) +
\lambda^s\xi^s(\Phi(q))\big[(\wh\Phi^s_\eps)'(q)-(\Phi^s)'(q)\big]\Big],
}
where in the last line we used the fact that $(\wh\Phi_\eps^t)'(q) - (\Phi^t)'(q) = 0$ for $t\neq s$.
Given any $t\in\SSS$, by the mean value theorem, there is some vector $\vc x_{\eps,t} = (x^r_{\eps,t})_{r\in\SSS}$ such that $x^r_{\eps,t} = \Phi^r(q)$ for all $r\neq s$, $\Phi^s(q) \leq x^s_{\eps,t} \leq \wh\Phi^s_\eps(q)$, and
\eq{
\xi^t(\wh\Phi_\eps(q))-\xi^t(\Phi(q))
= (\wh\Phi^s_\eps(q)-\Phi^s(q))\frac{\partial\xi^t}{\partial q^s}(\vc x_{\eps,t}).
}
Recalling \eqref{function_change} and using the fact that $\lambda^t\partial\xi^t/\partial q^s = \lambda^s\partial\xi^s/\partial q^t$, we have
\eeq{ \label{part_2_1}
\frac{1}{\eps}\int_0^1 \zeta\big([0,q]\big)&\sum_{t\in\SSS}\lambda^t\big[\xi^t(\wh\Phi_\eps(q))-\xi^t(\Phi(q))\big](\wh\Phi^t_\eps)'(q)\ \dd q \\
= \lambda^s\bigg[
\frac{1}{\eps}&\int_{a-\delta-\eps}^{a-\delta}\zeta\big([0,q]\big)(q-a+\delta+\eps)\frac{\partial\xi^s}{\partial q^t}(\vc x_{\eps,t})(\wh\Phi^t_\eps)'(q)\ \dd q \\
+&\int_{a-\delta}^a\zeta\big([0,q]\big)\sum_{t\in\SSS}\frac{\partial\xi^s}{\partial q^t}(\vc x_{\eps,t})(\wh\Phi^t_\eps)'(q)\ \dd q \\
-\frac{1}{\eps}&\int_{a}^{\wh a_\eps}\zeta\big([0,q]\big)(\Phi^s(a)+\eps+(q-a)\eps^2-\Phi^s(q))(\wh\Phi^t_\eps)'(q)\ \dd q\bigg].
}
The first integrand on the right-hand side is bounded by $C\eps$ for some constant $C$ depending only on $\xi$.
Since the interval of integration is itself of length $\eps$, we thus have
\eq{
\lim_{\eps\searrow0}\frac{1}{\eps}\int_{a-\delta-\eps}^{a-\delta}\zeta\big([0,q]\big)(q-a+\delta+\eps)\frac{\partial\xi^s}{\partial q^t}(\vc x_{\eps,t})(\wh\Phi^t_\eps)'(q)\ \dd q = 0.
}
Meanwhile, the third integrand on the right-hand side of \eqref{part_2_1} is at most $C(\eps + \eps^2)$, and the interval of integration has vanishing length by Lemma \ref{a_hat_lemma}.
Therefore, this integral also vanishes in the limit:
\eq{
\lim_{\eps\searrow0}\frac{1}{\eps}\int_{a}^{\wh a_\eps}\zeta\big([0,q]\big)(\Phi^s(a)+\eps+(q-a)\eps^2-\Phi^s(q))(\wh\Phi^t_\eps)'(q)\ \dd q = 0.
}
The second (and only remaining) integrand on the right-hand side of \eqref{part_2_1} is bounded by a constant, and converges (as $\eps\searrow0$) to
\eq{
\zeta\big([0,q]\big)\sum_{t\in\SSS}\frac{\partial\xi^s}{\partial q_t}(\Phi(q))(\Phi^t)'(q)
= \zeta\big([0,q]\big)(\xi^s\circ\Phi)'(q).
}
By dominated convergence, we have now argued that
\eeq{ \label{part_2_1_complete}
\lim_{\eps\searrow0}\frac{1}{\eps}\int_0^1 \zeta\big([0,q]\big)\sum_{t\in\SSS}\lambda^t\big[\xi^t(\wh\Phi_\eps(q))-\xi^t(\Phi(q))\big](\wh\Phi^t_\eps)'(q)\ \dd q
= \lambda^s\int_{a-\delta}^a\zeta\big([0,q]\big)(\xi^s\circ\Phi)'(q)\ \dd q.
}
It remains to consider the second term on the right-hand side of \eqref{part_2}.
Using \eqref{derivative_change}, we find that
\eq{
\frac{1}{\eps}
\int_0^1\zeta\big([0,q]\big)\xi^s(\Phi(q))\big[(\wh\Phi^s_\eps)'(q)-(\Phi^s)'(q)\big]\ \dd q
= \frac{1}{\eps}&\int_{a-\delta-\eps}^{a-\delta}\zeta\big([0,q]\big)\xi^s(\Phi(q))\ \dd q \\
+ \eps&\int_a^{\wh a_\eps}\zeta\big([0,q]\big)\xi^s(\Phi(q))\ \dd q \\
- \frac{1}{\eps}&\int_a^{\wh a_\eps} \zeta\big([0,q]\big)\xi^s(\Phi(q))(\Phi^s)'(q)\ \dd q.
}
By continuity of $\xi^s\circ\Phi$ and left-continuity of $q\mapsto\zeta\big([0,q)\big)$, the first integral on the right-hand side converges to $\zeta([0,a-\delta))\xi^s(\Phi(a-\delta))$.
The second integral clearly converges to $0$.
For the third and final integral, we can appeal to Lemma \ref{dirac_delta_lemma}.
Putting these facts together, we conclude
\eeq{ \label{part_2_2_complete}
&\lim_{\eps\searrow0}\frac{1}{\eps}
\int_0^1\zeta\big([0,q]\big)\xi^s(\Phi(q))\big[(\wh\Phi^s_\eps)'(q)-(\Phi^s)'(q)\big]\ \dd q \\
&= \zeta([0,a-\delta))\xi^s(\Phi(a-\delta))
- \zeta\big([0,a]\big)\xi^s(\Phi(a)).
}
Because of \eqref{part_2}, the claimed result \eqref{B_deriv_2} now follows by adding \eqref{part_2_1_complete} and \eqref{part_2_2_complete}.
\end{proofclaim}

The desired identity \eqref{B_deriv} is simply the sum result of Claims \ref{delta_derivative_claim}, \ref{B_deriv_1_claim}, and \ref{B_deriv_2_claim}.
\end{proof}

Now let us briefly discuss down-perturbations.
As the name suggests, we need to replace the assumptions from \eqref{up_assumptions} by
\begin{subequations}
\label{down_assumptions}
\eeqs{
\label{Phi_above_a_assumption}
\Phi^t(a+\delta)&<\mathrlap{1}\phantom{\Phi^s(q)} \quad \text{for all $t\in\SSS$}, \\
\label{Phi_at_a_assumption}
\Phi^s(a) &> 0, \\
\label{at_a_assumption_down}
\Phi^s(a) &> \Phi^s(q) \quad \text{for all $q\in[0,a)$}.
}
\end{subequations}
Under these conditions, for any $\eps>0$ small enough, we can make definitions analogous to those for up-perturbations:
\eq{
\wc a_\eps &\coloneqq \inf\{q\leq a:\, \Phi^s(q) = \Phi^s(a) - \eps - (a-q)\eps^2\}, \\
\wc\Phi^s_\eps(q) &\coloneqq \begin{cases}
\Phi^s(q) - (a+\delta+\eps-q) &\text{if }q\in[a+\delta,a+\delta+\eps), \\
\Phi^s(q) - \eps &\text{if } q\in[a,a+\delta), \\
\Phi^s(a) - \eps - (a-q)\eps^2 &\text{if } q\in[\wc a_\eps,a), \\
\Phi^s(q) &\text{otherwise}.
\end{cases}
}
See Figure \ref{down_fig} for an illustration.
Instead of \eqref{alpha_up_def}, we now take $\alpha_\eps(q)$ to be
\eq{ 
\alpha_\eps(q) \coloneqq \sum_{t\in\SSS}\lambda^t\wc\Phi_\eps^t(q).
}
The same argument as before will show $q\mapsto \alpha_\eps(q)$ is strictly increasing, and so $\alpha_\eps^{-1}$ is well-defined.
We then take
\eeq{ \label{down_perturbation_rescaled}
\zeta_\eps\big([0,q]\big) \coloneqq \zeta\big([0,\alpha_\eps^{-1}(q)]\big), \qquad
\Phi_\eps(q) \coloneqq \wc\Phi_\eps(\alpha_\eps^{-1}(q)).
}
By arguments parallel to those for Proposition \ref{derivative_prop}, we have the following calculation.

\begin{prop}[Down-perturbation] \label{derivative_prop_down}
Assuming \eqref{down_assumptions}, let $(\zeta_\eps,\Phi_\eps)$ be as in \eqref{down_perturbation_rescaled}.
We then have the following right derivative: 
\eeq{ \label{B_deriv_down}
\frac{\dd B(\zeta_\eps,\Phi_\eps)}{\dd\eps}\Big|_{\eps=0^+}
=\frac{\lambda^s}{2}\bigg[&\Big(\zeta\big([0,a+\delta]\big)-\zeta\big([0,a)\big)\Big)\Big(\red{h_s^2}-\int_0^{a}\frac{(\Phi^s)'(q)}{(\Delta^2(q))^2}\ \dd q\Big) \\
&+\int_{a}^{a+\delta}\Big(\zeta\big([0,q]\big)-\zeta\big([0,a+\delta]\big)\Big)\frac{(\Phi^s)'(q)}{(\Delta^2(q))^2}\ \dd q \\
&-\int_{a}^{a+\delta} \zeta\big([0,q]\big)(\xi^s\circ\Phi)'(q)\ \dd q \\
&+\zeta\big([0,a+\delta]\big)\xi^s(\Phi(a+\delta))
-\zeta\big([0,a)\big)\xi^s(\Phi(a))\bigg].
\raisetag{3.5\baselineskip}
}
\end{prop}

\begin{proof}
Let us just highlight the differences relative to the proof of Proposition \ref{derivative_prop}:
\begin{itemize}
    \item Replace Lemma \ref{dirac_delta_lemma} with the statement that
    \eq{
    \lim_{\eps\searrow0}\frac{1}{\eps}\int_{\wc a_\eps}^a f_\eps(q)(\Phi^s)'(q)\ \dd q = f_0(a),
    }
    provided $f_0$ is \textit{left}-continuous at $a$, because the interval $[\wc a_\eps,a]$ collapses to $a$ from the left.
    This is why the first and last instances of $\zeta\big([0,a]\big)$ in \eqref{B_deriv} are replaced by $\zeta\big([0,a)\big)$ in \eqref{B_deriv_down}.
    Similarly, we replace $\zeta\big([0,a-\delta)\big)$ with $\zeta\big([0,a+\delta])$ because $q\mapsto\zeta\big([0,q]\big)$ is \textit{right}-continuous, and the interval $[a+\delta,a+\delta+\eps]$ collapses to $a+\delta$ from the right.
    \item Replace Claim \ref{delta_derivative_claim} with the statement that
    \eq{
    \frac{\dd\wc\Delta^s_\eps(q)}{\dd\eps}\Big|_{\eps=0^+}=\begin{cases}
    \zeta\big([0,a+\delta]\big)-\zeta\big([0,a)\big) &\text{if }q\in[0,a), \\
    \zeta\big([0,a+\delta]\big) &\text{if }q\in[a,a+\delta], \\
    0 &\text{otherwise}.
    \end{cases}
    }
    The endpoints of intervals here correspond with the intervals of integration in \eqref{B_deriv_down}.
    Also notice that the middle case is positive rather than negative; this is why $\zeta(\big[0,a+\delta]\big)$ appears in the second line of \eqref{B_deriv_down} with a sign opposite that of $\zeta\big([0,a])$ in \eqref{B_deriv}.
    \item Notice that \eqref{X_deriv} will be replaced with
    \eq{
    \frac{1}{\Delta^s(a+\delta)}-\frac{1}{\Delta^s(a)} = \int_a^{a+\delta}\zeta\big([0,q]\big)\frac{(\Phi^s)'(q)}{(\Delta^s(q))^2}\ \dd q.
    }
    This is why $\zeta\big([0,q]\big)$ appears in second line of \eqref{B_deriv_down} with a sign opposite that of the same term in \eqref{B_deriv}.
    \item Whereas before $\wh\Phi^s_\eps\geq \Phi^s$, we now have $\wc\Phi^s_\eps\leq\Phi^s$.
    This is why the sign of the third line in \eqref{B_deriv_down} has changed.
    No other terms are affected because \eqref{function_change} was only used once---to obtain \eqref{part_2_1}---and the only surviving integral on the right-hand side of \eqref{part_2_1} was the middle one.
\end{itemize}
\end{proof}

We now use the calculations we have made to prove the identity \eqref{cs_identity} for minimizers of the C--S functional.
Recall the definitions of $\AA$ and $\AA_k$ from \eqref{AA_def} and \eqref{AA_k_def}.

\begin{proof}[Proof of Theorem \ref{cs_min_thm}]
The first claim is that a minimizer exists.
By Lemma \ref{density_lemma} every $(\zeta,\Phi)\in\AA$ satisfying \eqref{gap_assumption} can be approximated (with respect to $\DD$) by $(\zeta_k,\Phi)\in\AA_k$ as $k\to\infty$, and furthermore $B(\zeta_k,\Phi)\to B(\zeta,\Phi)$.
Hence
\eeq{ \label{discrete_inf_converge}
\lim_{k\to\infty} \inf_{\AA_k} B = \inf_\AA B.
}
On the other hand, Lemma \ref{away_from_1_lemma} gives the existence of some $\bar q<1$ such that the following is true.
For every $k\geq1$, there is $(\zeta_k,\Phi_k)\in\AA_k(\bar q)$ such that
\eq{
B(\zeta_k,\Phi_k) = \inf_{\AA_k} B.
}
So let $(\zeta,\Phi)\in\AA(\bar q)$ be some subsequential limit of $(\zeta_k,\Phi_k)$ as $k\to\infty$, which exists because the quotient space of $\AA(\bar q)$ is compact under $\DD$.
Since continuity of $B$ is guaranteed by Proposition \ref{continuity_prop}, we have
\eq{
B(\zeta,\Phi) = \lim_{k\to\infty} B(\zeta_k,\Phi_k) = \inf_\AA B.
}

For the remainder of the proof, assume $(\zeta,\Phi)\in\AA(\bar q)$ is a minimizer for $B$, and let $q_*\in(0,1)$ be such that $\zeta\big([0,q_*]\big) = 1$ and $\Phi^s(q_*) \leq \bar q$ for each $s$.

\begin{claim} \label{up_claim}
For any $a\in(0,q_*]$ and $\delta\geq 0$ at which \eqref{delta_small_assumption} and \eqref{at_a_assumption} hold, we have
\eeq{ \label{B_deriv_up_optimal}
0\leq \Big(\zeta\big([0,a-\delta)\big)-\zeta\big([0,a]\big)\Big)\Big(\red{h_s^2}-\int_0^{a}\frac{(\Phi^s)'(q)}{(\Delta^2(q))^2}\ \dd q +\xi^s(\Phi(a-\delta))\Big).
}
\end{claim}

\begin{proofclaim}
The hypotheses of the claim allow us to apply Proposition \ref{derivative_prop}.
Using the assumption that $(\zeta,\Phi)$ is a minimizer, we learn from \eqref{B_deriv} that
\eeq{ \label{B_deriv_up_consequence}
0\leq&\Big(\zeta\big([0,a-\delta)\big)-\zeta\big([0,a]\big)\Big)\Big(\red{h_s^2}-\int_0^{a-\delta}\frac{(\Phi^s)'(q)}{(\Delta^2(q))^2}\ \dd q\Big) \\
&-\int_{a-\delta}^a\Big(\zeta\big([0,q]\big)-\zeta\big([0,a]\big)\Big)\frac{(\Phi^s)'(q)}{(\Delta^2(q))^2}\ \dd q \\
&+\int_{a-\delta}^a \zeta\big([0,q]\big)(\xi^s\circ\Phi)'(q)\ \dd q \\
&+\zeta\big([0,a-\delta)\big)\xi^s(\Phi(a-\delta))
-\zeta\big([0,a]\big)\xi^s(\Phi(a)).
}
Now take note of the following trivial inequalities:
\eq{
\int_{a-\delta}^a\Big(\zeta\big([0,a]\big)-\zeta\big([0,q]\big)\Big)\frac{(\Phi^s)'(q)}{(\Delta^2(q))^2}\ \dd q
&\leq \Big(\zeta\big([0,a]\big)-\zeta\big([0,a-\delta)\big)\Big)\int_{a-\delta}^a\frac{(\Phi^s)'(q)}{(\Delta^2(q))^2}\ \dd q,
}
as well as
\eq{
\int_{a-\delta}^a \zeta\big([0,q]\big)(\xi^s\circ\Phi)'(q)\ \dd q 
&\leq \zeta\big([0,a]\big)\int_{a-\delta}^a (\xi^s\circ\Phi)'(q)\ \dd q  \\
&= \zeta\big([0,a]\big)\big(\xi^s(\Phi(a))-\xi^s(\Phi(a-\delta))\big).
}
Using these inequalities in \eqref{B_deriv_up_consequence} results in \eqref{B_deriv_up_optimal}.
\end{proofclaim}

By completely analogous arguments (just using Proposition \ref{derivative_prop_down} instead of Proposition \ref{derivative_prop}) we also obtain the next claim.

\begin{claim} \label{down_claim}
For any $a\in(0,q_*]$ and $\delta\geq 0$ at which \eqref{down_assumptions} holds, we have
\eq{ 
0\leq \Big(\zeta\big([0,a+\delta]\big)-\zeta\big([0,a)\big)\Big)\Big(\red{h_s^2}-\int_0^{a}\frac{(\Phi^s)'(q)}{(\Delta^2(q))^2}\ \dd q +\xi^s(\Phi(a+\delta))\Big).
}
\end{claim}

For convenience, let us write $K = \Supp(\zeta)\subset[0,q_*]$, and also $K^\cc=(0,1)\setminus\Supp(\zeta)$.
Since $K^\cc$ is open, it is a disjoint union of countably many intervals of the form $(a_0,a_1)$, where $a_0\in\{0\}\cup K$ and $a_1\in K \cup \{1\}$.
Note that altering $\Phi$ off of $K$ does not change $(\zeta,\Phi)$ under the pseudometric $\DD$. 
Therefore, we can make the following modification.
For every one of the disjoint intervals $(a_0,a_1)\subset K^\cc$ just described, replace $\Phi\big|_{[a_0,a_1]}$ with the linear interpolation between $\Phi(a_0)$ and $\Phi(a_1)$.
In this way, we may assume the following for each $s\in\SSS$:
\eeq{ \label{off_of_K_assumption}
\text{on every $(a_0,a_1)\subset K^\cc$, $\Phi^s$ is either strictly increasing or constant}.
}
In particular, \eqref{at_a_assumption} holds for $a =q_*$, since $\Phi^s(q_*)<1=\Phi^s(1)$ for all $s\in\SSS$. 

Now suppose $q'\in K$ and $\Phi^s(q')=x\in[0,1)$.
Define
\eq{
a_0=\inf\{q\geq0:\, \Phi^s(q) = x\} \qquad \text{and} \qquad
a_1=\sup\{q\geq0:\, \Phi^s(q) = x\},
}
so that $[a_0,a_1]$ is the maximal interval containing $q'$ on which $\Phi^s$ is constant.
Note that $a_1\leq q_*$ by \eqref{off_of_K_assumption}.

\begin{claim} \label{in_support_claim}
Assuming \eqref{off_of_K_assumption}, we must have $a_0\in\{0\}\cup K$ and $a_1\in K$.
\end{claim}

\begin{proofclaim}
Let us just argue for $a_0$, as the argument for $a_1$ is identical.
If $a_0=q'$, then we are done.
Otherwise, we have $a_0 < a_1$.
And if $a_0$ were an element of $K^\cc$, then there would exist $\eps\in(0,a_1-a_0)$ such that $(a_0-\eps,a_0+\eps)\subset K^\cc$.
But then $(a_0-\eps,a_0+\eps)$ would contain both an interval $(a_0-\eps,a_0]$ of non-constancy for $\Phi^s$, and an interval $[a_0,a_0+\eps)$ of constancy for $\Phi^s$.
This scenario contradicts \eqref{off_of_K_assumption}.
\end{proofclaim}

\begin{claim} \label{almost_done}
If $a_1>0$, then we have
\eeq{ \label{to_contradict}
\xi^s(\Phi(q))+\red{h_s^2}  = \int_0^{q}\frac{(\Phi^s)'(q)}{(\Delta^2(q))^2}\ \dd q \quad \text{for all $q\in[a_0,a_1]$}.
}
\end{claim}

\begin{proofclaim}
By maximality of $a_1$, we must have \eqref{at_a_assumption} for $a=a_1$.
Moreover, since $\Phi^s$ is continuous, for any $\delta_0>0$, there exists $\delta\in(0,\delta_0)$ such that \eqref{at_a_assumption} holds for $a=a_1+\delta$.
By Claim \ref{up_claim}, we then have
\eq{
0\leq \Big(\zeta\big([0,a_1-\delta)\big)-\zeta\big([0,a_1+\delta]\big)\Big)\Big(\red{h_s^2}-\int_0^{a_1+\delta}\frac{(\Phi^s)'(q)}{(\Delta^2(q))^2}\ \dd q + \xi^s(\Phi(a_1-\delta))\Big).
}
Since $a_1\in K$, the first factor on the right-hand side is negative.  So dividing it out results in
\eq{
\xi^s(\Phi(a_1-\delta))+\red{h_s^2} \leq \int_0^{a_1+\delta}\frac{(\Phi^s)'(q)}{(\Delta^2(q))^2}\ \dd q.
}
As this conclusion holds for all $\delta$ along some sequence tending to $0$, we conclude that
\eq{
\xi^s(\Phi(a_1))+\red{h_s^2} \leq \int_0^{a_1}\frac{(\Phi^s)'(q)}{(\Delta^2(q))^2}\ \dd q.
}
But we have supposed $(\Phi^s)'(q)=0$ for all $q\in(a_0,a_1)$, and so the integral can actually be taken over just the interval $[0,a_0]$, meaning 
\eeq{ \label{contradict_1}
\xi^s(\Phi(a_1))+\red{h_s^2}  \leq \int_0^{a_0}\frac{(\Phi^s)'(q)}{(\Delta^2(q))^2}\ \dd q.
}
If $\Phi^s(a_0) > 0$, then by parallel reasoning (using the minimality of $a_0$ and Claim \ref{down_claim}), we obtain
\eeq{ \label{contradict_2}
\xi^s(\Phi(a_0))+\red{h_s^2}  \geq \int_0^{a_0}\frac{(\Phi^s)'(q)}{(\Delta^2(q))^2}\ \dd q.
}
This inequality remains true if $\Phi^s(a_0) = 0$, since the right-hand side is zero.
Finally, note that $\xi^s(\Phi(a_0))\leq\xi^s(\Phi(a_1))$ simply because $\xi^s$ is non-decreasing in all coordinates.
Consequently, \eqref{contradict_1} and \eqref{contradict_2} together yield
\eq{
\xi^s(\Phi(a_0))+\red{h_s^2}  = \int_0^{a_0}\frac{(\Phi^s)'(q)}{(\Delta^2(q))^2}\ \dd q
=
\int_0^{a_1}\frac{(\Phi^s)'(q)}{(\Delta^2(q))^2}\ \dd q
= \xi^s(\Phi(a_1))+\red{h_s^2} .
}
This identity extends to \eqref{to_contradict} because for any $q\in(a_0,a_1)$, we have
$(\Phi^s)'(q) = 0$ and $\xi^s(\Phi(a_0))\leq\xi^s(\Phi(q))\leq\xi^s(\Phi(a_1))$.
\end{proofclaim}

Now recall the point $q'\in K\cap[a_0,a_1]$.
If $q'>0$, then Claim \ref{almost_done} has established the desired identity \eqref{cs_identity}.
In the case $0\in\Supp(\zeta)$, we must provide a separate argument to guarantee that \eqref{cs_identity} holds for $q=0$.
First observe that if $\Supp(\zeta)$ contains positive numbers arbitrarily close to $0$, then continuity ensures that \eqref{cs_identity} continues to hold at $q=0$.
So we may assume that $\Supp(\zeta) \setminus \{0\} \subset [q_1,1]$ for some $q_1>0$.
Now, if $h_s^2=0$, then \eqref{cs_identity} holds trivially at $q=0$, with both sides equal to $0$.
If instead $h_s^2>0$, then the following claim tells us that $0\notin\Supp(\zeta)$, and so it is not even necessary to check \eqref{cs_identity} at $q=0$.

\begin{claim}
Asumming $\Supp(\zeta) \setminus \{0\} \subset [q_1,1]$ for some $q_1>0$, and $h_s^2>0$, we must have $\zeta(\{0\})=0$.
\end{claim}

\begin{proof}
The argument is similar to that of Claim \ref{jump_claim}.
Let us choose $q_1$ maximally so that $q_1\in\Supp(\zeta)\cup\{1\}$.
If $q_1\in\Supp(\zeta)$, then we already know
\eq{
0<\xi^s(\Phi(q_1))+\red{h_s^2}  = \int_0^{q_1}\frac{(\Phi^s)'(u)}{(\Delta^s(u))^2}\ \dd u,
}
and so we must have $\Phi^s(q_1)>0$.
If instead $q_1=1$, then of course $\Phi^s(q_1)=\Phi^s(1)=1>0$.

Again because of Proposition \ref{continuity_prop}, we can modify $\Phi$ off the support of $\zeta$ without changing the value of $B(\zeta,\Phi)$.
So let us fix $s\in\SSS$ and some $p\in(0,\lambda^s\Phi^s(q_1)\wedge q_1)$, and then assume that
\eq{
\Phi^s(q) = \begin{cases}
\displaystyle q/\lambda^s &\text{if $q\in[0,p]$}, \\
\displaystyle\frac{q_1-q}{q_1-p}(p/\lambda^s)+\frac{q-p}{q_1-p}\Phi^s(q_1)
&\text{if $q\in(p,q_1]$}.
\end{cases}
}
Correspondingly, for $t\neq s$, we assume that
\eq{
\Phi^t(q) = \begin{cases}
 0 &\text{if $q\in[0,p]$}, \\
\displaystyle\frac{q-p}{q_1-p}\Phi^t(q_1)
&\text{if $q\in(p,q_1]$}.
\end{cases}
}
It is easy to check that these assumptions preserve $\vc\lambda$-admissibility.

Now suppose toward a contradiction that $\zeta(\{0\})>0$, and consider the following perturbed measure for sufficiently small $\eps>0$:
\eq{
\wt\zeta_\eps \coloneqq \zeta - \eps\delta_0 + \eps\delta_{p}.
}
If we define
\eq{
\Delta^t_\eps(q) \coloneqq \int_q^1\wt\zeta_\eps\big([0,u]\big)(\Phi^t)'(u)\ \dd u,
}
then by construction we have
\eq{
\Delta^t_\eps(q) - \Delta^t(q) = -\eps\one_{\{q<p\}}\int_q^{ p}(\Phi^t)'(u)\ \dd u
=-\eps\one_{\{q<p\}}(\Phi^t(p)-\Phi^t(q)).
}
This gives the derivative calculation
\eq{ 
\frac{\dd}{\dd\eps} \Delta^t_\eps(q)
= -\one_{\{q<p\}}(\Phi^t(p)-\Phi^t(q)) = -\one_{\{q<p\}}\one_{\{t=s\}}\frac{p-q}{\lambda^s}.
}
Very similarly, we have
\eq{ 
\frac{\dd}{\dd\eps}\int_0^1\wt\zeta_\eps\big([0,q]\big)(\xi\circ\Phi)'(q)\ \dd q
= -\int_0^{p}(\xi\circ\Phi)'(q)\ \dd q
= -\xi(\Phi(p)). 
}
Referring to the two previous displays, we have 
\eeq{ \label{first_pass_derivative}
\frac{\dd B(\wt\zeta_\eps,\Phi)}{\dd\eps}
= \frac{\lambda^s}{2}\Big[-\red{h_s^2}\frac{p}{\lambda^s}
+ \int_0^p \frac{(\Phi^s)'(q)}{(\Delta^s(q))^2}\cdot\frac{p-q}{\lambda^s}\ \dd q\Big] -\frac{\xi(\Phi(p))}{2}.
}
Since $\Phi^t(p) = (p/\lambda^s)\one_{\{t=s\}}$, the last term is given by
\eq{
\xi(\Phi(p))
= \int_0^{p/\lambda^s}\lambda^s\xi^s\big|_{\{q^s=q,\, q^t = 0\text{ for all $t\neq s$}\}}\ \dd q \geq p\xi^s(\vc 0).
}
Furthermore, since $(\Phi^s)' \leq 1/\lambda^s$, we have
\eq{
\int_0^p\frac{(\Phi^s)'(q)}{(\Delta^s(q))^2}\cdot\frac{p-q}{\lambda^s}\ \dd q \leq \Big(\frac{p}{\lambda^s\Delta^s(p)}\Big)^2.
}
Using the two previous displays in \eqref{first_pass_derivative}, we obtain
\eq{
\frac{\dd B(\wt\zeta_\eps,\Phi)}{\dd\eps}
\leq \frac{\lambda^s}{2}\Big[-h_s^2\frac{p}{\lambda^s} + \Big(\frac{p}{\lambda^s\Delta^s(p)}\Big)^2\Big].
}
Since $h_s^2>0$, we can choose $p$ sufficiently small that the right-hand side is negative, thereby contradicting the assumption that $(\zeta,\Phi)$ is a minimizer.
\renewcommand{\qedsymbol}{$\square$ (Claim and Theorem)}
\end{proof}
\renewcommand{\qedsymbol}{}
\end{proof}
\renewcommand{\qedsymbol}{$\square$}
\vspace{-1\baselineskip}

\section{Proof of Crisanti--Sommers formula}
\label{parisi_cs_proof}
This final section has the single goal of proving Theorem \ref{parisi_cs}. 
The result will follow from the following three lemmas.

\begin{lemma} \label{big_lemma}
If $(\zeta,\Phi)$ satisfies \eqref{gap_assumption} and
\eeq{ \label{bd_Delta}
b^s-d^s(q) = \frac{1}{\Delta^s(q)} \quad \text{for all $q\in\Supp(\zeta)$},
}
then $A(\zeta,\Phi,\vc b) = B(\zeta,\Phi)$.
\end{lemma}

\begin{lemma} \label{small_lemma_parisi}
If $(\zeta,\Phi,\vc b)$ satisfies $\eqref{parisi_identity}$, then \eqref{bd_Delta} holds.
\end{lemma}

\begin{lemma} \label{small_lemma_cs}
If $(\zeta,\Phi)$ satisfies \eqref{gap_assumption}, \eqref{cs_identity}, and for each $s\in\SSS$ we have
\eeq{ \label{initial_bd_Delta}
b^s - d^s(q_*) = \frac{1}{\Delta^s(q_*)} \qquad \text{and} \qquad
\xi^s(\Phi(q_*))+\red{h_s^2} = \int_0^{q_*}\frac{(\Phi^s)'(u)}{(\Delta^s(u))^2}\ \dd u,
}
then \eqref{bd_Delta} holds.
\end{lemma}

Before proving the lemmas, let us give the argument for Theorem \ref{parisi_cs}.

\begin{proof}[Proof of Theorem \ref{parisi_cs}]
Let us first assume \eqref{xi_strictly_convex} holds. 
By Theorem \ref{parisi_min_thm}, there exists a triple $(\zeta,\Phi,\vc b)$ which minimizes $A$ and satisfies \eqref{parisi_identity}.
It then follows from Lemma \ref{small_lemma_parisi} that \eqref{bd_Delta} holds.
Furthermore, \eqref{parisi_identity} excludes the possibility that $1\in\Supp(\zeta)$, since otherwise we would have $\Phi^s(1) = 1 - 1/b^s < 1$.
So the maximum element of $\Supp(\zeta)$ is some $q_*\in[0,1)$, and for each $s\in\SSS$ we have
\eq{
\Phi^s(q_*) &\stackref{parisi_identity_2}{=}\frac{h_s^2+\red{\xi^s(\vc0)}}{(b^s-d^s(0))^2}+ \int_0^{q_*}\frac{(\xi^s\circ\Phi)'(u)}{(b^s-d^s(u))^2}\ \dd u \\
&\stackrefp{parisi_identity_2}{\leq}
\frac{h_s^2+\red{\xi^s(\vc0)}}{(b^s-d^s(0))^2}+\int_0^1\frac{(\xi^s\circ\Phi)'(u)}{(b^s-d^s(u))^2}\ \dd u
\stackref{parisi_identity_1}{=}1-\frac{1}{b^s}<1.
}
Now that we have verified \eqref{gap_assumption}, we can apply Lemma \ref{big_lemma} to conclude that
\eq{ 
\inf A = A(\zeta,\Phi,\vc b) = B(\zeta,\Phi) \geq \inf B.
}
To obtain the reverse inequality, we take $(\zeta,\Phi)$ to be the minimizer of $B$ guaranteed by Theorem \ref{cs_min_thm}, which necessarily satisfies \eqref{cs_identity}.
Let $q_*\in[0,1)$ be the maximum of $\Supp(\zeta)$, so that \eqref{cs_identity} implies the second statement in \eqref{initial_bd_Delta},
 and then choose $b^s$ to satisfy the first statement in \eqref{initial_bd_Delta}.
By Lemma \ref{small_lemma_cs}, it follows that \eqref{bd_Delta} holds, and then Lemma \ref{big_lemma} gives
\eq{ 
\inf A \leq A(\zeta,\Phi,\vc b) = B(\zeta,\Phi) = \inf B.
}
Note that this second inequality did not rely on \eqref{xi_strictly_convex} or even \eqref{xi_convex}.

Now we must argue that $\inf A = \inf B$ even if we relax \eqref{xi_strictly_convex} to \eqref{xi_convex}.
So assume the covariance function $\xi$ satisfies \eqref{xi_convex}.
For $\eps>0$, consider the replacement of $\xi$ by
\eq{
\xi_\eps(\vc q) = \xi(\vc q) + \eps\sum_{s\in\SSS}(\lambda^sq^s)^2.
}
This is equivalent to replacing the Hamiltonian $H_N:\T_N\to\R$ of \eqref{HN_def} with
\eq{
H_{N,\eps}(\sigma) \coloneqq H_N(\sigma) + \sqrt{\eps}\sum_{s\in\SSS}g^s(\sigma),
}
where $(g^s)_{s\in\SSS}$ are independent Gaussian processes such that
\eq{
\E[g^s(\sigma^1)g^s(\sigma^2)] = N(\lambda^sR^s(\sigma^1,\sigma^2))^2, \quad \text{with} \quad 
R^s(\sigma^1,\sigma^2) = \frac{\iprod{\sigma^1(s)}{\sigma^2(s)}}{N^s}.
}
(The process $g^s$ does indeed exist, since it is just a spherical SK model on $S_{N^s}$.)
By \cite[Lem.~A.1]{bates-sohn22}, this affects the free energy in \eqref{free_energy_def} as follows:
\eq{
\E F_N \leq \E F_{N,\eps} \leq \E F_N + \eps/2.
}
By Theorem \ref{parisi_thm}, we then have
\eeq{ \label{inf_A_converges}
\inf A \leq \inf A_\eps \leq \inf A + \eps/2,
}
where $A_\eps$ is the result of replacing $\xi$ with $\xi_\eps$ in \eqref{A_def}.
Furthermore, $\xi_\eps$ clearly satisfies \eqref{xi_strictly_convex} since $\xi$ already satisfies \eqref{xi_convex}, and so the first part of this proof gives
\eeq{ \label{inf_A_equals_B}
\inf A_\eps = \inf B_\eps,
}
where $B_\eps$ is result of replacing $\xi$ with $\xi_\eps$ in \eqref{B_def}.

For each $\eps>0$, \eqref{discrete_inf_converge} permits us to choose an integer $k_\eps$ large enough that
\eeq{ \label{k_eps_choice}
\inf_{\AA_{k_\eps}} B_\eps \leq \inf_{\AA} B_\eps+ \eps.
}
Take any $\bar q<1$ which satisfies
\eq{
\bar q &> \limsup_{\eps\searrow0}
\max_{s\in\SSS}\frac{(1-u^s_\eps)(\red{h_s^2}+\xi^s_\eps(\vc 1))+u^s_\eps}{(1-u^s_\eps)(\red{h_s^2}+\xi^s_\eps(\vc 1) )+1}, \quad \text{where} \\
u^s_\eps &\coloneqq 1-\frac{\sqrt{1+4(\red{h_s^2}+\xi^s_\eps(\vc 1))}-1}{2(\red{h_s^2}+\xi^s_\eps(\vc 1))}.
}
We make this choice so that for all $\eps>0$ sufficiently small, \eqref{k_eps_choice} and Lemma \ref{away_from_1_lemma} allow us to find $(\zeta_\eps,\Phi_\eps)\in\AA_{k_\eps}(\bar q)$ such that
\eeq{ \label{almost_inf_B_eps}
B_\eps(\zeta_\eps,\Phi_\eps) \leq \inf B_\eps + \eps.
}
Finally, let $(\zeta,\Phi)\in\AA(\bar q)$ be any subsequential limit (with respect to $\DD$) of $(\zeta_\eps,\Phi_\eps)$ as $\eps\searrow0$.
That is, as laws on $[0,1]^\SSS$, $\zeta_\eps\circ\Phi_\eps^{-1}$ converges weakly to $\zeta\circ\Phi^{-1}$.
Since $\xi_\eps$ converges uniformly to $\xi$ on $[0,1]^\SSS$, it follows from Proposition \ref{continuity_prop} that
\eq{
B(\zeta,\Phi)
=
\lim_{\eps\searrow0} B_\eps(\zeta_\eps,\Phi_\eps)
\stackref{almost_inf_B_eps}{=}
\lim_{\eps\searrow0} B_\eps
\stackref{inf_A_equals_B}{=}
\lim_{\eps\searrow0} A_\eps
\stackref{inf_A_converges}{=}
\inf A.
}
Thus $\inf A \geq \inf B$ even when \eqref{xi_strictly_convex} is relaxed to \eqref{xi_convex}.
As the reverse inequality holds regardless, we are done.
\end{proof}

Now we must prove the three lemmas we have just used.
We begin with Lemma \ref{big_lemma}, which is the technical heart of this section.

\begin{proof}[Proof of Lemma \ref{big_lemma}]
Note that if $q_0$ is the minimium of $\Supp(\zeta)$, then \eqref{bd_Delta} implies
\eq{
b^s - d^s(0) = b^s - d^s(q_0) = \frac{1}{\Delta^s(q_0)}=\frac{1}{\Delta^s(0)} > 0,
}
and so $A(\zeta,\Phi,\vc b)$ is well-defined.
Furthermore, we have
\eeq{ \label{external_terms_match}
\sum_{s\in\SSS}\lambda^s\frac{h_s^2+\red{\xi^s(\vc0)}}{(b^s-d^s(0))^2} = \sum_{s\in\SSS}\lambda^s (h_s^2+\red{\xi^s(\vc0)})\Delta_s(0).
}
For convenience, let us choose $q_*$ minimally; that is, $q_*$ is the maximum of $\Supp(\zeta)$.
First note that because $\zeta\big([0,q])=1$ for all $q\in[q_*,1]$, we have
\eeq{ \label{above_qstar_identities}
d^s(q) = \xi^s(\vc 1)-\xi^s(\Phi(q)) \quad \text{and} \quad
\Delta^s(q) = 1 - \Phi^s(q) \quad \text{for all $q\in[q_*,1]$}.
}
Since we chose $q_*$ to belong to $\Supp(\zeta)$, the assumption \eqref{bd_Delta} now gives
\eeq{ \label{at_qstar_identities}
\Phi^s(q_*)\cdot(b^s-d^s(q_*))
&= \frac{\Phi^s(q_*)}{1-\Phi^s(q_*)} \\
&= -1 + \frac{1}{\Delta^s(q_*)}
=-1 + b^s - \xi^s(\vc 1) + \xi^s(\Phi(q_*)).
}
To condense notation, let us write $\Supp(\zeta) =  K$ and $K^\cc = (0,q_*]\setminus\Supp(\zeta)$.
Consider the first integral appearing in \eqref{A_def}:
\eeq{ \label{parisi_int_1_lower}
&\sum_{s\in\SSS}\lambda^s\int_0^{q_*}\frac{(\xi^s\circ\Phi)'(q)}{b^s-d^s(q)}\ \dd q  \stackref{bd_Delta}{=}
\sum_{s\in\SSS}\lambda^s\Big[\int_{K} \Delta^s(q)(\xi^s\circ\Phi)'(q)\ \dd q 
+ \int_{K^\cc}\frac{(\xi^s\circ\Phi)'(q)}{b^s-d^s(q)}\ \dd q\Big] \\
&=
\sum_{s\in\SSS}\lambda^s\Big[\int_0^{q_*} \Delta^s(q)(\xi^s\circ\Phi)'(q)\ \dd q-\int_{K^\cc} \Delta^s(q)(\xi^s\circ\Phi)'(q)\ \dd q 
+ \int_{K^\cc}\frac{(\xi^s\circ\Phi)'(q)}{b^s-d^s(q)}\ \dd q\Big]. 
}
For the first integral on the final line, we use integration by parts:
\eeq{ \label{first_ibp}
&\int_0^{q_*}\Delta^s(q)(\xi^s\circ\Phi)'(q)\ \dd q \\
&= \Delta^s(q_*)\xi^s(\Phi(q_*))-\Delta^s(0)\xi^s(\vc0)+\int_0^{q_*}
\zeta\big([0,q]\big)(\Phi^s)'(q)(\xi^s\circ\Phi)(q)\ \dd q.
}
Recall from \eqref{above_qstar_identities} that $\Delta^s(q_*) = 1-\Phi^s(q_*)$, and also observe that
\eeq{ \label{combine_identity_1}
\sum_{s\in\SSS}\lambda^s(\xi^s\circ\Phi)(q)(\Phi^s)'(q)
=(\xi\circ\Phi)'(q)\ \dd q.
}
Consequently, when we sum over the various species, \eqref{first_ibp} becomes
\eq{
\sum_{s\in\SSS}\lambda^s
\int_0^{q_*}\Delta^s(q)(\xi^s\circ\Phi)'(q)\ \dd q
= (\vc 1 - \Phi(q_*))\cdot\nabla\xi(\Phi(q_*))
+ \int_0^{q_*}\zeta\big([0,q]\big)(\xi\circ \Phi)'(q)\ \dd q.
}
Notice that a portion of the second integral in \eqref{B_def} has appeared on the right-hand side.
The remaining portion is
\eq{
\int_{q_*}^1\zeta\big([0,q]\big)(\xi\circ \Phi)'(q)\ \dd q
= \int_{q_*}^1(\xi\circ \Phi)'(q)\ \dd q
= \xi(\vc 1)-\xi(\Phi(q_*)).
}
Adding this quantity to both sides of \eqref{parisi_int_1_lower}, we obtain
\eq{
&\sum_{s\in\SSS}\lambda^s\int_0^{q_*}\frac{(\xi^s\circ\Phi)'(q)}{b^s-d^s(q)}\ \dd q + \xi(\vc 1) - \xi(\Phi(q_*)) \\
&= (\vc 1-\Phi(q_*))\cdot\nabla\xi(\Phi(q_*))
+\int_0^1\zeta\big([0,q]\big)(\xi\circ\Phi)'(q)\ \dd q \\
&\phantom{=} +\sum_{s\in\SSS}\lambda^s\Big[-\Delta^s(0)\xi^s(\vc0)-\int_{K^\cc} \Delta^s(q)(\xi^s\circ\Phi)'(q)\ \dd q 
+ \int_{K^\cc}\frac{(\xi^s\circ\Phi)'(q)}{b^s-d^s(q)}\ \dd q\Big].
}
We must also calculate 
\eq{
\int_{q_*}^1\frac{(\xi^s\circ\Phi)'(q)}{b^s-d^s(q)}\ \dd q
\stackref{above_qstar_identities}{=} \log b^s - \log(b^s-d^s(q_*))
&\stackref{bd_Delta}{=} \log b^s + \log \Delta^s(q_*).
}
Combining the two previous displays, we arrive at
\eeq{ \label{AB_1}
&\sum_{s\in\SSS}\lambda^s\Big[-\log b^s+\int_0^{1}\frac{(\xi^s\circ\Phi)'(q)}{b^s-d^s(q)}\ \dd q\Big] \\
&= (\vc 1-\Phi(q_*))\cdot\nabla\xi(\Phi(q_*))
+ \sum_{s\in\SSS}\lambda^s\log \Delta^s(q_*)+ \int_0^1\zeta\big([0,q]\big)(\xi\circ\Phi)'(q)\ \dd q \\
&\phantom{=} +\sum_{s\in\SSS}\lambda^s\Big[-\Delta^s(0)\xi^s(\vc0)-\int_{K^\cc} \Delta^s(q)(\xi^s\circ\Phi)'(q)\ \dd q 
+ \int_{K^\cc}\frac{(\xi^s\circ\Phi)'(q)}{b^s-d^s(q)}\ \dd q\Big].
}
Next we consider the first integral in \eqref{B_def}:
\eeq{ \label{pre_pre_theta_conversion}
&\sum_{s\in\SSS}\lambda^s\int_0^{q_*}\frac{(\Phi^s)'(q)}{\Delta^s(q)}\ \dd q 
\stackref{bd_Delta}{=} \sum_{s\in\SSS}\lambda^s\Big[\int_{K}(b^s-d^s(q))(\Phi^s)'(q)\ \dd q
+ \int_{K^\cc}\frac{(\Phi^s)'(q)}{\Delta^s(q)}\ \dd q\Big] \\
&= \sum_{s\in\SSS}\lambda^s\Big[\int_0^{q_*}(b^s-d^s(q))(\Phi^s)'(q)\ \dd q
- \int_{K^\cc}(b^s-d^s(q))(\Phi^s)'(q)\ \dd q
+ \int_{K^\cc}\frac{(\Phi^s)'(q)}{\Delta^s(q)}\ \dd q\Big]. \raisetag{3.5\baselineskip}
}
For the first integral on the final line, we integrate by parts:
\eq{ 
&\int_0^{q_*}(b^s-d^s(q))(\Phi^s)'(q)\ \dd q \\
&\stackref{at_qstar_identities}{=}-1 + b^s - \xi^s(\vc 1) + \xi^s(\Phi(q_*))- \int_0^{q_*}\zeta\big([0,q]\big)(\xi^s\circ\Phi)'(q)\Phi^s(q)\ \dd q.
}
Observe by direct calculation (using definition \eqref{gamma_theta_def}) that
\eeq{ \label{combine_identity_2}
(\theta\circ\Phi)'(q) = \sum_{s\in\SSS}\Phi^s(q)\cdot\lambda^s(\xi^s\circ\Phi)'(q).
}
Therefore, we can rewrite \eqref{pre_pre_theta_conversion} as
\eeq{ \label{after_summing}
&\sum_{s\in\SSS}\lambda^s\int_0^{q_*}\frac{(\Phi^s)'(q)}{\Delta^s(q)}\ \dd q \\
&= \sum_{s\in\SSS}\lambda^s(b^s-1) - \vc 1\cdot\nabla\xi(\vc 1)
+ \vc 1\cdot\nabla\xi(\Phi(q_*))-\int_0^{q_*}\zeta\big([0,q]\big)(\theta\circ\Phi)'(q)\ \dd q \\
&\phantom{=} + \sum_{s\in\SSS}\lambda^s\Big[- \int_{K^\cc}(b^s-d^s(q))(\Phi^s)'(q)\ \dd q+\int_{K^\cc}\frac{(\Phi^s)'(q)}{\Delta^s(q)}\ \dd q\Big].
}
We also have
\eq{
\int_{q_*}^1\zeta\big([0,q]\big)(\theta\circ\Phi)'(q)\ \dd q
&= \theta(\vc 1) - \theta(\Phi(q_*)). \\
&= \vc1\cdot\nabla\xi(\vc 1) - \xi(\vc 1) - \Phi(q_*)\cdot\nabla\xi(\Phi(q_*)) + \xi(\Phi(q_*)).
}
Subtracting this quantity from both sides of \eqref{after_summing}, and then rearranging terms, we arrive at
\eeq{ \label{AB_2}
&\sum_{s\in\SSS}\lambda^s(b^s-1)-\int_0^1\zeta\big([0,q]\big)(\theta\circ\Phi)'(q)\ \dd q - \xi(\vc 1) + \xi(\Phi(q_*))\\
&= -(\vc 1-\Phi(q_*))\cdot\nabla\xi(\Phi(q_*))+\sum_{s\in\SSS}\lambda^s\int_0^{q_*}\frac{(\Phi^s)'(q)}{\Delta^s(q)}\ \dd q \\
&\phantom{=} + \sum_{s\in\SSS}\lambda^s\Big[ \int_{K^\cc}(b^s-d^s(q))(\Phi^s)'(q)\ \dd q-\int_{K^\cc}\frac{(\Phi^s)'(q)}{\Delta^s(q)}\ \dd q\Big].
}
By adding \eqref{external_terms_match}, \eqref{AB_1}, and \eqref{AB_2}, and then recalling definitions \eqref{A_def} and \eqref{B_def}, we obtain
\eq{
2A(\zeta,\Phi,\vc b)
= 2B(\zeta,\Phi)
+ \sum_{s\in\SSS}\lambda^s\Big[ &\int_{K^\cc}(b^s-d^s(q))(\Phi^s)'(q)\ \dd q
-\int_{K^\cc} \Delta^s(q)(\xi^s\circ\Phi)'(q)\ \dd q \\
+&\int_{K^\cc}\frac{(\xi^s\circ\Phi)'(q)}{b^s-d^s(q)}\ \dd q
-\int_{K^\cc}\frac{(\Phi^s)'(q)}{\Delta^s(q)}\ \dd q
\Big].
}
Therefore, the proof will be complete once we show that the additional terms on the right-hand side sum to zero.

Since $K$ is closed, its complement $K^\cc = (0,q_*]\setminus K$ is a countable union of disjoint open intervals of the form $(a_0,a_1)$, where $a_0$ and $a_1$ are both elements of $K\cup\{0\}$. 
We claim that for each such interval, we have
\eeq{ \label{on_each_interval}
\sum_{s\in\SSS}\lambda^s\Big[ &\int_{a_0}^{a_1}(b^s-d^s(q))(\Phi^s)'(q)\ \dd q
-\int_{a_0}^{a_1} \Delta^s(q)(\xi^s\circ\Phi)'(q)\ \dd q \\
+&\int_{a_0}^{a_1}\frac{(\xi^s\circ\Phi)'(q)}{b^s-d^s(q)}\ \dd q
-\int_{a_0}^{a_1}\frac{(\Phi^s)'(q)}{\Delta^s(q)}\ \dd q
\Big]=0.
}
Of course, this claim is sufficient to conclude the proof.

Since $(a_0,a_1)\subset K^\cc$, the map $u\mapsto\zeta\big([0,u]\big)$ is constant on $[a_0,a_1)$.
Referring to definition \eqref{ds_def}, we see that
\eeq{ \label{ds_constant}
d^s(q) = d^s(a_1) + \zeta\big([0,a_0]\big)\big[\xi^s(\Phi(a_1)) - \xi^s(\Phi(q))\big] \quad \text{for all $q\in[a_0,a_1]$}.
}
Similarly, referring to definition \eqref{delta_def}, we have
\eeq{ \label{delta_constant}
\Delta^s(q) = \Delta^s(a_1) + \zeta\big([0,a_0]\big)\big[\Phi^s(a_1) - \Phi^s(q)\big] \quad \text{for all $q\in[a_0,a_1]$}.
}
If $\zeta\big([0,a_0]\big)=0$, then in fact $\zeta\big([0,a_1)\big)=0$, which means $(0,a_1)\subset K^\cc$ and so $a_0$ must be 0.
In this case, we have
\eq{
\int_0^{a_1}\frac{(\xi^s\circ\Phi)'(q)}{b^s-d^s(q)}\ \dd q
&\stackref{ds_constant}{=} \frac{\xi^s(\Phi(a_1))-\xi^s(\vc0)}{b^s-d^s(a_1)} \\
&\stackrefpp{bd_Delta}{ds_constant}{=}
\Delta^s(a_1)\big(\xi^s(\Phi(a_1))-\xi^s(\vc0)\big)
\stackref{delta_constant}{=}
\int_0^{a_1}\Delta^s(q)(\xi^s\circ\Phi)'(q)\ \dd q,
}
as well as
\eq{
\int_0^{a_1}\frac{(\Phi^s)'(q)}{\Delta^s(q)}\ \dd q
&\stackref{delta_constant}{=} \frac{\Phi^s(a_1)}{\Delta^s(a_1)} \\
&\stackrefpp{bd_Delta}{ds_constant}{=}
(b^s-d^s(a_1))\Phi^s(a_1)
\stackref{ds_constant}{=}
\int_0^{a_1}(b^s-d^s(q))(\Phi^s)'(q)\ \dd q.
}
Hence \eqref{on_each_interval} is true if $\zeta\big([0,a_0]\big)=0$.

If $\zeta\big([0,a_0]\big) > 0$, then we instead have
\eq{
\int_{a_0}^{a_1}\frac{(\xi^s\circ\Phi)'(q)}{b^s-d^s(q)}\ \dd q
&\stackref{ds_constant}{=} \frac{1}{\zeta\big([0,a_0])}\log\frac{b^s-d^s(a_1)}{b^s-d^s(a_0)} \\
&\stackrefpp{bd_Delta}{ds_constant}{=} \frac{1}{\zeta\big([0,a_0]\big)}\log\frac{\Delta^s(a_0)}{\Delta^s(a_1)}
\stackref{delta_constant}{=}
\int_{a_0}^{a_1}\frac{(\Phi^s)'(q)}{\Delta^s(q)}\ \dd q.
}
That is, the second line of \eqref{on_each_interval} vanishes, and so we wish to show that the first line also vanishes.
By another application of \eqref{ds_constant}, we find
\eq{
\int_{a_0}^{a_1}(b^s-d^s(a_1))(\Phi^s)'(q)\ \dd q
= \big[b^s - d^s(q) - \zeta\big([0,a_0]\big)\xi^s(\Phi(a_1))\big](\Phi^s(a_1)-\Phi^s(a_0)) \\
+ \zeta\big([0,a_0]\big)\int_{a_0}^{a_1}\xi^s(\Phi(q))(\Phi^s)'(q)\ \dd q.
}
Once we sum over the various species and apply \eqref{combine_identity_1}, this identity becomes
\eeq{ \label{lhs_transformed}
&\sum_{s\in\SSS}\lambda^s
\int_{a_0}^{a_1}(b^s-d^s(q))(\Phi^s)'(q)\ \dd q \\
&=\sum_{s\in\SSS}\lambda^s\big[b^s - d^s(a_1) - \zeta\big([0,a_0]\big)\xi^s(\Phi(a_1))\big](\Phi^s(a_1)-\Phi^s(a_0)) \\
&\phantom{=}+\zeta\big([0,a_0]\big)\big(\xi(\Phi(a_1))-\xi(\Phi(a_0))\big).
}
By analogous computations using \eqref{delta_constant} and \eqref{combine_identity_2}, we also have
\eeq{ \label{rhs_transformed}
&\sum_{s\in\SSS}\lambda^s
\int_{a_0}^{a_1}\Delta^s(q)(\xi^s\circ\Phi)'(q)\ \dd q \\
&\stackrefp{gamma_theta_def}{=}\sum_{s\in\SSS}\lambda^s\big[\Delta^s(a_1) + \zeta\big([0,a_0]\big)\Phi^s(a_1)\big]\big(\xi^s(\Phi(a_1))-\xi^s(\Phi(a_0))\big) \\
&\phantom{\stackref{gamma_theta_def}{=}}-\zeta\big([0,a_0]\big)\big(\theta(\Phi(a_1))-\theta(\Phi(a_0))\big) \\
&\stackref{gamma_theta_def}{=}\sum_{s\in\SSS}\lambda^s\big[\Delta^s(a_1) + \zeta\big([0,a_0]\big)\Phi^s(a_1)\big]\big(\xi^s(\Phi(a_1))-\xi^s(\Phi(a_0))\big) \\
&\phantom{\stackref{gamma_theta_def}{=}}+\zeta\big([0,a_0]\big)\bigg[\big(\xi(\Phi(a_1))-\xi(\Phi(a_1))\big)
-\sum_{s\in\SSS}\lambda^s\big(\Phi^s(a_1)\xi^s(\Phi(a_1))-\Phi^s(a_0)\xi^s(\Phi(a_0))\big)\bigg].
\raisetag{5\baselineskip}
}
Now observe that
\eq{
\zeta\big([0, a_0]\big)\big(\xi^s(\Phi(a_1))-\xi^s(\Phi(a_0))\big)
&\stackref{ds_constant}{=} d^s(a_0)-d^s(a_1), \quad \text{as well as} \\
\zeta\big([0, a_0]\big)(\Phi^s(a_1) - \Phi^s(a_0))
&\stackref{delta_constant}{=}
\Delta^s(a_0)-\Delta^s(a_1).
}
Using these identities in conjunction with \eqref{bd_Delta}, we find that
\eeq{ \label{first_difference}
&(b^s-d^s(a_1))(\Phi^s(a_1)-\Phi^s(a_0))
- \Delta^s(a_1)\big(\xi^s(\Phi(a_1))-\xi^s(\Phi(a_0))\big) \\
&= \frac{1}{\zeta\big([0,a_0]\big)}\bigg[\frac{1}{\Delta^s(a_1)}\Big(\Delta^s(a_0)-\Delta^s(a_1)\Big)
- \Delta^s(a_1)\Big(\frac{1}{\Delta^s(a_1)}-\frac{1}{\Delta^s(a_0)}\Big)\bigg] \\
&= \frac{1}{\zeta\big([0,a_0]\big)}\cdot\frac{(\Delta^s(a_0)-\Delta^s(a_1))^2}{\Delta^s(a_1)\Delta^s(a_0)} \\
&= (\Phi^s(a_1)-\Phi^s(a_0))\Big(\frac{1}{\Delta^s(a_1)}-\frac{1}{\Delta^s(a_0)}\Big) \\
&= \zeta\big([0, a_0]\big)(\Phi^s(a_1)-\Phi^s(a_0))\big(\xi^s(\Phi(a_1))-\xi^s(\Phi(a_0))\big).
}
Now subtract \eqref{rhs_transformed} from \eqref{lhs_transformed}, and divide by $\zeta\big([0,a_0]\big)$.
In light of \eqref{first_difference}, this results in
\eq{
\sum_{s\in\SSS}\lambda^s\bigg[&(\Phi^s(a_1)-\Phi^s(a_0))\big(\xi^s(\Phi(a_1))-\xi^s(\Phi(a_0))\big)
-\xi^s(\Phi(a_1))(\Phi^s(a_1)-\Phi^s(a_0)) \\
&-\Phi^s(a_1)\big(\xi^s(\Phi(a_1))-\xi^s(\Phi(a_0))\big)
+ \Phi^s(a_1)\xi^s(\Phi(a_1))-\Phi^s(a_0)\xi^s(\Phi(a_0))\bigg]=0.
}
That is, the first line of \eqref{on_each_interval} vanishes, and so we are done.
\end{proof}

The reader will notice a parallel structure in the proofs of our last two lemmas.

\begin{proof}[Proof of Lemma \ref{small_lemma_parisi}]
We start at (the inverse of) the right-hand side of \eqref{bd_Delta}, and will transform it to (the inverse of) the left-hand side.
For any $q\in[0,1]$, integration by parts gives
\eeq{ \label{first_writing_delta}
\Delta^s(q) = \int_q^1 \zeta\big([0,u]\big)(\Phi^s)'(u)\ \dd u
= 1 - \zeta\big([0,q]\big)\Phi^s(q) - \int_{(q,1]}\Phi^s(u)\ \zeta(\dd u).
}
Using the hypothesis \eqref{parisi_identity_2}, the integral on the right-hand side can be rewritten as
\eq{
\int_{(q,1]}\Phi^s(u)\ \zeta(\dd u)
= \int_{(q,1]}\int_0^u \frac{(\xi^s\circ\Phi)'(v)}{(b^s-d^s(v))^2}\ \dd v\, \zeta(\dd u) + \frac{h_s^2+\red{\xi^s(\vc0)}}{(b^s-d^s(0))^2}\zeta\big((q,1]\big).
}
Now use a reverse integration by parts:
\eq{
&\int_{(q,1]}\int_0^u \frac{(\xi^s\circ\Phi)'(v)}{(b^s-d^s(v))^2}\ \dd v\, \zeta(\dd u)  \\
&= \int_0^1 \frac{(\xi^s\circ\Phi)'(v)}{(b^s-d^s(v))^2}\ \dd v
- \zeta([0,q]\big)\int_0^q \frac{(\xi^s\circ\Phi)'(v)}{(b^s-d^s(v))^2}\ \dd v
- \int_q^1\zeta\big([0,u]\big) \frac{(\xi^s\circ\Phi)'(u)}{(b^s-d^s(u))^2}\ \dd u.
}
Invoking \eqref{parisi_identity_1} and \eqref{parisi_identity_2} under the assumption that $q\in\Supp(\zeta)$, we simplify the right-hand side to obtain
\eq{
1 - \frac{1}{b^s}-\frac{h_s^2+\red{\xi^s(\vc0)}}{(b^s-d^s(0))^2} - \zeta\big([0,q]\big)\Big(\Phi^s(q)-\frac{h_s^2+\red{\xi^s(\vc0)}}{(b^s-d^s(0))^2}\Big)
+ \int_q^1\Big(\frac{1}{b^s-d^s(u)}\Big)'\ \dd u.
}
In light of the three previous displays, \eqref{first_writing_delta} now reads as
\eq{
\Delta^s(q) = \frac{1}{b^s} - \int_q^1\Big(\frac{1}{b^s-d^s(u)}\Big)'\ \dd u
= \frac{1}{b^s-d^s(q)} \quad \text{for all $q\in\Supp(\zeta)$}.\tag*{\qedhere}
}
\end{proof}

\begin{proof}[Proof of Lemma \ref{small_lemma_cs}]
We start at the left-hand side of \eqref{bd_Delta}, and will transform it to the right-hand side.
For any $q\in[0,q_*]$, we have
\eeq{ \label{first_writing_bd}
b^s - d^s(q)
&\stackrefp{initial_bd_Delta}{=} b^s - d^s(q_*) + d^s(q_*) - d^s(q) \\
&\stackref{initial_bd_Delta}{=} \frac{1}{\Delta^s(q_*)} - \int_{q}^{q_*}\zeta\big([0,u]\big)(\xi^s\circ\Phi)'(u)\ \dd u.
}
Now use integration by parts:
\eq{
\int_{q}^{q_*}\zeta\big([0,u]\big)(\xi^s\circ\Phi)'(u)\ \dd u
= \xi^s(\Phi(q_*)) - \zeta\big([0,q]\big)\xi^s(\Phi(q))
- \int_{(q,q_*]}\xi^s(\Phi(u))\ \zeta(\dd u).
}
Using \eqref{cs_identity} and a reverse integration by parts, we find that
\eq{
 \int_{(q,q_*]}\xi^s(\Phi(u))\ \zeta(\dd u)
 &= \int_{(q,q_*]}\int_0^u\frac{(\Phi^s)'(v)}{(\Delta^s(v))^2}\ \dd v\, \zeta(\dd u) - \red{h_s^2}\zeta\big((q,q_*]\big) \\
 &= \int_0^{q_*}\frac{(\Phi^s)'(v)}{(\Delta^s(v))^2}\ \dd v
 - \zeta\big([0,q]\big)\int_0^q\frac{(\Phi^s)'(v)}{(\Delta^s(v))^s}\ \dd v \\
 &\phantom{=}- \int_q^{q_*}\zeta\big([0,u]\big)\frac{(\Phi^s)'(u)}{(\Delta^s(u))^2}\ \dd u-\red{h_s^2} \zeta\big((q,q_*]\big).
}
Using \eqref{initial_bd_Delta} and \eqref{cs_identity} under the assumption that $q\in\Supp(\zeta)$, we simplify the final line to
\eq{
&\xi^s(\Phi(q_*))+\red{h_s^2} - \zeta\big([0,q]\big)\big(\xi^s(\Phi(q))+\red{h_s^2}\big)
- \int_{q}^{q_*}\Big(\frac{1}{\Delta^s(u)}\Big)'\ \dd u
- \red{h_s^2}\zeta\big((q,q_*]\big) \\
&= \xi^s(\Phi(q_*)) - \zeta\big([0,q]\big)\xi^s(\Phi(q))
- \int_{q}^{q_*}\Big(\frac{1}{\Delta^s(u)}\Big)'\ \dd u.
}
In light of the three previous displays, \eqref{first_writing_bd} now reads as
\eq{
b^s - d^s(q) = \frac{1}{\Delta^s(q_*)} - 
\int_{q}^{q_*}\Big(\frac{1}{\Delta^s(u)}\Big)'\ \dd u
= \frac{1}{\Delta^s(q)} \quad \text{for all $q\in\Supp(\zeta)$}.
\tag*{\qedhere}}
\end{proof}

\section*{Acknowledgments}
We are grateful to Amir Dembo for valuable feedback and suggestions, to the referees for useful comments, and to Pax Kivimae for the detection of a computational error in a previous draft.

\bibliography{spin_glasses}

\end{document}